\documentclass[10pt, journal]{IEEEtran}

\usepackage[english]{babel}
\usepackage{blindtext}
\usepackage{graphicx}
\usepackage{color}
\usepackage{multirow}
\usepackage{algorithm}
\usepackage{algorithmicx}
\usepackage[noend]{algpseudocode}
\usepackage{enumitem}
\usepackage{microtype}
\usepackage{multicol}
\usepackage{amsmath, amsthm}
\usepackage{amssymb}
\usepackage[font=small,labelfont=bf]{caption}
\usepackage{cite}
\usepackage{url}

\newtheorem{theorem}{Theorem}

\newtheorem{lemma}{Lemma}
\newtheorem{definition}{Definition}

\usepackage{mathptmx}
\DeclareMathAlphabet{\mathcal}{OMS}{cmsy}{m}{n}
\DeclareMathAlphabet{\mathrm}{OT1}{bch}{m}{n}
\DeclareMathAlphabet{\mathit}{OT1}{bch}{m}{it}

\DisableLigatures{encoding=*, family=*}

\newcommand{\sysname}{\textsf{MVPipe}\xspace}

\renewcommand{\paragraph}[1]{\noindent\textbf{#1}}

\makeatletter
\newcommand*{\algrule}[1][\algorithmicindent]{%
\hspace*{.2em}% <------------- This is where the rule starts from
\vrule %height .75\baselineskip depth .25\baselineskip
\hspace*{\dimexpr#1-.2em-.4pt}%
}

\newcommand{\StatePar}[1]{%
    \State\parbox[t]{\dimexpr\linewidth-\ALG@thistlm}{\strut #1\strut}%
}
\renewcommand{\ALG@beginalgorithmic}{\offinterlineskip}% Remove all interline skips

\newcount\ALG@printindent@tempcnta
\def\ALG@printindent{%
    \ifnum \theALG@nested>0% is there anything to print
    \ifx\ALG@text\ALG@x@notext% is this an end group without any text?
    \else
    \unskip
    \ALG@printindent@tempcnta=1
    \loop
    \algrule[\csname ALG@ind@\the\ALG@printindent@tempcnta\endcsname]%
    \advance \ALG@printindent@tempcnta 1
    \ifnum
    \ALG@printindent@tempcnta<\numexpr\theALG@nested+1\relax
    \repeat
    \fi
    \fi
}
\makeatother
\algrenewcommand\algorithmicend{\strut\textbf{end}}
\algrenewcommand\algorithmicdo{\strut\textbf{do}}
\algrenewcommand\algorithmicwhile{\strut\textbf{while}}
\algrenewcommand\algorithmicfor{\strut\textbf{for}}
\algrenewcommand\algorithmicforall{\strut\textbf{for all}}
\algrenewcommand\algorithmicloop{\strut\textbf{loop}}
\algrenewcommand\algorithmicrepeat{\strut\textbf{repeat}}
\algrenewcommand\algorithmicuntil{\strut\textbf{until}}
\algrenewcommand\algorithmicprocedure{\strut\textbf{procedure}}
\algrenewcommand\algorithmicfunction{\strut\textbf{function}}
\algrenewcommand\algorithmicif{\strut\textbf{if}}
\algrenewcommand\algorithmicthen{\strut\textbf{then}}
\algrenewcommand\algorithmicelse{\strut\textbf{else}}

\algrenewcommand\algorithmicrequire{\strut\textbf{Input:}}
\algrenewcommand\algorithmicensure{\strut\textbf{Output:}}

\let\oldState\State
\renewcommand{\State}{\oldState\strut}
\begin{document}

\title{MVPipe: Enabling Lightweight Updates and Fast Convergence in
Hierarchical Heavy Hitter Detection}

\author{Lu Tang,
        Qun Huang,
        Patrick P. C. Lee%
\thanks{The work was supported by Key-Area Research and Development Program
of Guangdong Province 2020B0101390001, Joint Funds of the National Natural
Science Foundation of China (U20A20179), National Natural Science Foundation
of China (62172007), the Fundamental Research Funds for the Central
Universities (20720210072), and the Natural Science Foundation of Fujian
Province of China (2021J05002).}
\thanks{L. Tang is with Department of Computer Science and Technology, Xiamen 
University, Xiamen, China (Email: tanglu@xmu.edu.cn).} 
\thanks{Q. Huang is with Department of Computer Science and Technology, Peking
University, Beijing, China (Email: huangqun@pku.edu.cn).}
\thanks{P. P. C. Lee is with the Department of Computer Science
and Engineering, The Chinese University of Hong Kong, Hong Kong, China
(Email:pclee@cse.cuhk.edu.hk).}% 
\thanks{Corresponding author: Patrick P. C. Lee.}
}

% The paper headers
%\markboth{
%   IEEE/ACM TRANSACTIONS ON NETWORKING}
%{Tang \MakeLowercase{\textit{et al.}}: MVPipe: Enabling Lightweight Updates and Fast Convergence in
%Hierarchical Heavy Hitter Detection}

\maketitle

\begin{abstract}
Finding hierarchical heavy hitters (HHHs) (i.e., hierarchical aggregates with
exceptionally huge amounts of traffic) is critical to network management, yet
it is often challenged by the requirements of fast packet processing,
real-time and accurate detection, as well as resource efficiency. Existing HHH
detection schemes either incur expensive packet updates for multiple
aggregation levels in the IP address hierarchy, or need to process sufficient
packets to converge to the required detection accuracy.  We present \sysname,
an invertible sketch that achieves both lightweight updates and
fast convergence in HHH detection.  \sysname builds on the skewness property
of IP traffic to process packets via a pipeline of majority voting executions,
such that most packets can be updated for only one or few aggregation levels
in the IP address hierarchy.  We show how \sysname
can be feasibly deployed in P4-based programmable switches subject to limited
switch resources.  We also theoretically analyze the accuracy and coverage
properties of \sysname.  Evaluation with real-world Internet traces shows that
\sysname achieves high accuracy, high throughput, and fast convergence
compared to six state-of-the-art HHH detection schemes.  It also incurs low
resource overhead in the Tofino switch deployment. 
\end{abstract}

%----------------------------------------------------------------------------
% Section I: Introduction
%----------------------------------------------------------------------------
\section{Introduction} 
\label{sec:introduction}

Network administrators often need to measure and characterize the anomalous
behaviors of IP traffic in operational networks.  IP traffic is inherently
{\em hierarchical}. It can be organized in hierarchical forms in one or
multiple dimensions. For example, it can be aggregated either by source IP
address prefixes (i.e., one-dimensional (1D)), or by the source-destination IP
address prefixes (i.e., two-dimensional (2D)).  Given the
hierarchical nature of IP traffic, finding {\em hierarchical heavy hitters
(HHHs)} (i.e., the hierarchical aggregates with exceptionally huge amounts
of traffic) is of particular interest to network measurement
\cite{Estan2003,Zhang2004,Jose2011,Moshref2013}.  One notable application of
HHH detection is to identify distributed denial-of-service (DDoS) or botnet
attacks \cite{Sekar2006,Fayaz2015}, in which the traffic aggregates of
multiple attack flows can bring substantial damage to a network.  

Unlike the classical {\em heavy hitter (HH)} detection problem
\cite{Estan2002,Karp2003,Metwally2005}, whose goal is to identify individual
large-sized flows (i.e., HHs),  HHH detection is a much more challenging task
as it needs to identify not only the HHs, but also the set of flows that have
small sizes each but have a huge aggregate size when combined together.
As there are many possibilities for aggregating traffic at different levels in
the IP address hierarchy (e.g., multiple lengths of prefixes in IP addresses),
enumerating all possible combinations of traffic aggregates is infeasible for
HHH detection.  This motivates the need for specialized algorithmic designs for
HHH detection. 

Like most network measurement tasks, practical HHH detection schemes need to
address the challenges of managing the ever-increasing speed and size of IP
traffic in modern networks. For a typical backbone link with a
bandwidth of tens or hundreds of Gigabits per second, network measurement tasks
should efficiently track millions of concurrently active flows at any time.
Maintaining per-flow states, or even any combination of traffic aggregates in
HHH detection, inevitably has tremendous resource demands.  In addition, with
the emergence of programmable networking, new network measurement solutions
(e.g., \cite{Sivaraman2017,Gupta2018,Narayana2017}) often offload packet
processing to programmable hardware switches for scalable
network measurement.  Unfortunately, the available switch resources 
are scarce (e.g., less than 2\,MB of SRAM per processing stage
\cite{Bosshart2013,Sivaraman2017}), thereby complicating the use of
HHH detection in programmable hardware switches.   To this end, HHH detection
should aim for the following design requirements: (i) {\em fast packet
processing} (i.e., keeping pace with the line rate of operational networks),
(ii) {\em real-time and accurate detection} (i.e., identify all HHHs in
real-time with low false positive/negative rates), and (iii) {\em resource
efficiency} (i.e., the computational and memory resources should be limited
within their available capacities in both hardware and software). 

HHH detection has been extensively studied in the literature for more than a
decade (see \S\ref{sec:related} for details).  
One class of HHH detection schemes is {\em streaming-based}
\cite{Lin2007,Cormode2008,Mitzenmacher2012,Zhang2004,Truong2009,Moraney2020}, in
which they use memory-efficient stream data structures to process IP traffic
and detect HHHs at short time scales, at the expense of incurring bounded
errors on detection.  However, such schemes often have high processing costs
to capture multiple aggregation levels of HHHs in stream data structures, and
hence cannot be readily scaled to line-rate processing in modern networks.
Another class of HHH detection schemes is {\em sampling-based}, by updating a
sketch instance
with only a sampled subset of packets \cite{BenBasat2017,BenBasat2020}.  One
representative example is randomized HHH (RHHH) \cite{BenBasat2017}, which
detects HHHs at long time scales. RHHH maintains multiple instances of sketches
for different aggregation levels and randomly selects one instance to update per
packet. It has high update performance, but has slow convergence, as the HHHs
cannot be detected until sufficient packets have been processed.  
Furthermore, both streaming-based and sampling-based HHH detection schemes
have been deployed in hardware
\cite{Jose2011,Moshref2013,Popescu2017,Kuvcera2020,BenBasat2020}, but 
their designs often face different limitations, such as relying on a
controller to specify what HHHs are monitored \cite{Jose2011,Moshref2013},
requiring specialized hardware (e.g., TCAM) to maintain high update throughput
\cite{Popescu2017, Kuvcera2020}, or sampling packets to trade convergence for
resource efficiency in hardware \cite{BenBasat2020}.

We present \sysname, an invertible sketch that achieves lightweight updates,
fast convergence, and resource efficiency in HHH detection, in both
software and hardware.  By ``invertible'', we mean that \sysname can
directly return all HHHs (with high accuracy) from the data structure itself.
\sysname's design builds on the observation that IP traffic is highly skewed
across multiple aggregation levels, in which most IP traffic belonging to large
flows can be aggregated in a single level, while only a small fraction of
traffic needs to be aggregated across all levels in the IP address hierarchy.
Specifically, \sysname maintains its sketch with small and static memory
allocation (i.e., the memory can be pre-allocated a priori) and tracks
aggregates via the pipelined executions of the {\em majority vote algorithm
(MJRTY)} \cite{Boyer1991}.  For most packets, \sysname only updates a single
level with a single MJRTY execution, while only for a small fraction of
packets, \sysname needs to update more levels along the pipeline with multiple
MJRTY executions (i.e., lightweight updates).  In addition, \sysname processes
all packets within its sketch, so it can detect HHHs at short time scales
(i.e., fast convergence) as opposed to sampling-based approaches.  Furthermore,
with small and static memory allocation, \sysname can be readily deployed in
programmable hardware switches.  

While we motivate our HHH detection problem from the hierarchical nature of IP
traffic, we expect that our \sysname design is also applicable to general
types of hierarchical datasets, such as geographic or temporal datasets
\cite{Mitzenmacher2012}.

The contributions of this paper are summarized as follows.
\begin{itemize}[leftmargin=*]
\item 
We design \sysname, a novel invertible sketch for HHH detection with three
major design features: lightweight updates, fast convergence, and resource
efficiency for deployment in both software and hardware.
\item
We implement \sysname on P4-based programmable switches \cite{p4} and compile
our prototype in the Tofino chipset \cite{tofino}, subject to limited
hardware resources. 
\item
We conduct theoretical analysis on \sysname, including its space and time
complexities, accuracy, and coverage.
\item 
We conduct trace-driven evaluation on \sysname in both software and hardware.
Evaluation in software shows that \sysname achieves higher detection accuracy,
faster convergence, and up to $22.13\times$ throughput gain compared to six
state-of-the-art HHH detection schemes. \sysname also incurs limited resource
overhead in the Tofino switch deployment. 
\end{itemize}

We open-source our \sysname prototype in both software and P4 at
{\bf \url{https://github.com/Grace-TL/MVPipe}}.

%--------------------------------------------------------------------------
% Section II: Problem Formulation
%--------------------------------------------------------------------------
\section{Problem Formulation}
\label{sec:problem}

We formulate the HHH detection problem; similar formulations are also found in
\cite{Cormode2008,Mitzenmacher2012,BenBasat2017}.  
We focus on IP traffic, which can be aggregated by different prefixes in the
IP address space.  We model IP traffic as a stream of packets. Each packet is
denoted by a tuple $(f, v_f)$ and is allowed to be processed only once.  In
network measurement, $f$ identifies a flow,  and $v_f$ is
either one (for packet counting) or the packet size (for byte counting).  In
this work, we consider one-dimensional (1D) and two-dimensional (2D) HHH
detection: for 1D HHH detection, $f$ refers to a source IP address (the same
arguments hold for a destination IP address); for 2D HHH detection, $f$ refers
to a source-destination IP address pair.  

We aggregate source addresses or source-destination address pairs by the
address prefixes at either byte-level or bit-level granularities; we refer to
them as {\em 1D-byte}, {\em 2D-byte}, {\em 1D-bit}, and {\em 2D-bit}
hierarchies.  Figure~\ref{fig:hierarchy} shows the 1D-byte and 2D-byte
hierarchies.  Each {\em node} corresponds to the {\em key} of a flow at a
certain aggregation level in a hierarchy.  We define the {\em level} of a node
as the position in a hierarchy of depth $d$, where the level ranges between 0
and $d-1$.  The key at the lowest level~0 is the most specific and
refers to an exact address (1D) or an address pair (2D), while the key at the
highest level $d-1$ corresponds to the most general aggregate (i.e., all
addresses or address pairs).  In general, a key refers to an address
prefix (1D) or a pair of address prefixes (2D).  The keys of lower-level nodes
(a.k.a.  {\em descendants}) can be {\em generalized} to the keys of their
higher-level nodes (a.k.a.  {\em ancestors}); for example, the key $1.2.3.4$
can be generalized by one byte to $1.2.3.*$.  Let $\prec$ be the
generalization relation of two keys.  For any keys $x$ and $y$, we say
that $x\prec y$ if $x$ can be generalized to $y$, and that $x\preceq y$ if
$x\prec y$ or $x=y$.

\begin{figure}[!t]
\centering
\begin{tabular}{@{\ }c@{\ }c}
\includegraphics[width=0.5in]{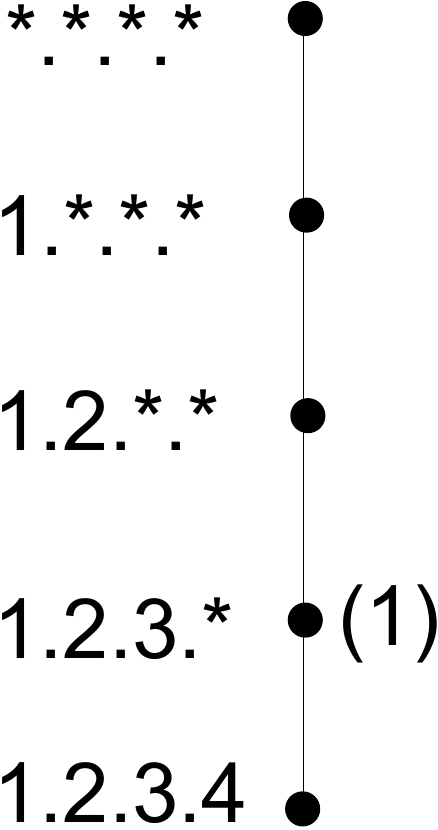} \hspace{20pt} &
\includegraphics[width=2in]{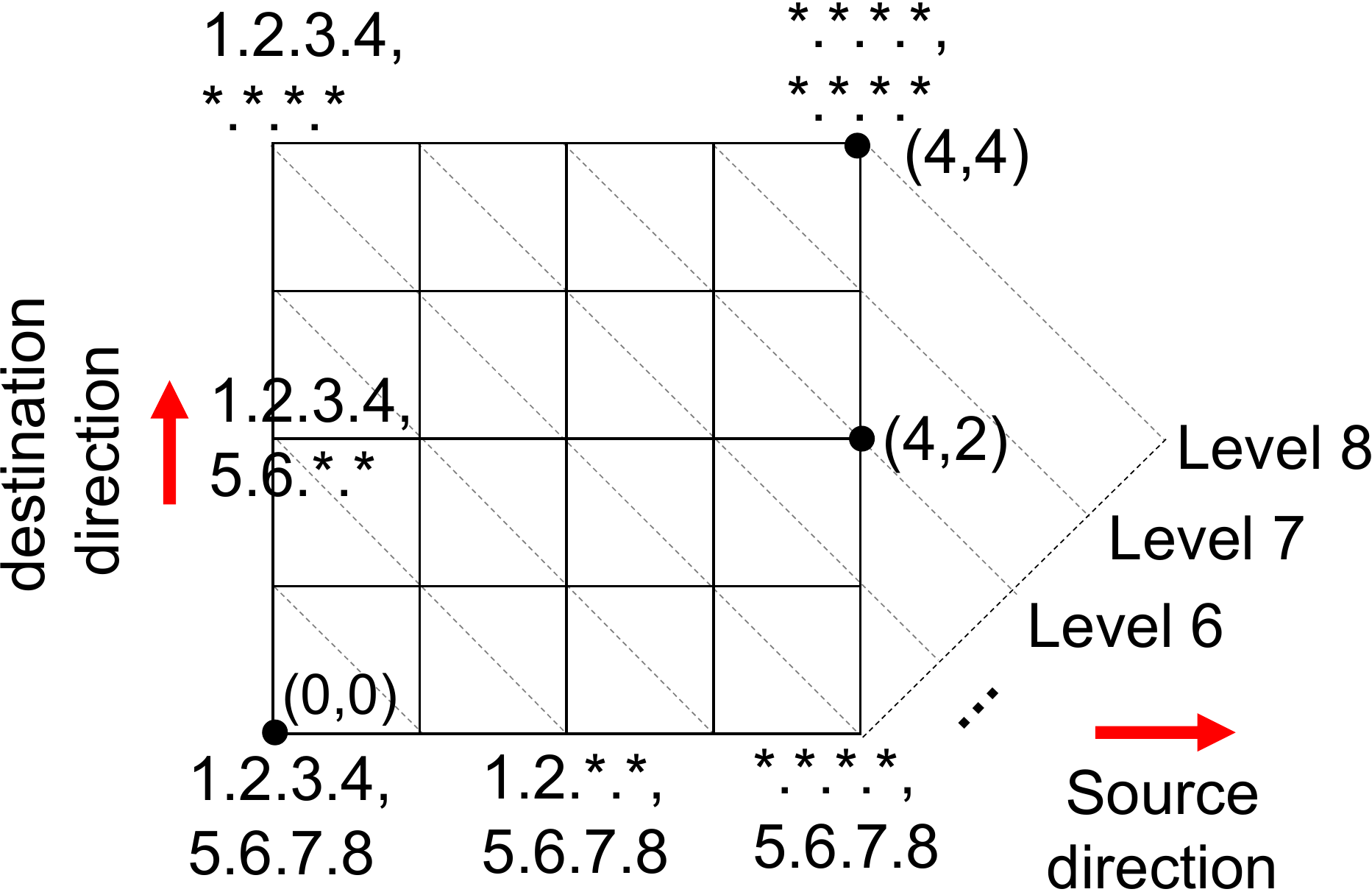} \\
{\small (a) 1D-byte hierarchy} \hspace{20pt} & 
{\small (b) 2D-byte hierarchy}
\end{tabular}
\vspace{-3pt}
\caption{1D-byte and 2D-byte hierarchies.}
\label{fig:hierarchy}
\vspace{-6pt}
\end{figure}

To quantify the level of a node, we associate each node with a {\em
coordinate} in multi-dimensional space as shown in
Figure~\ref{fig:hierarchy}. The $i$-th element of the coordinate represents
the degree of generalization in the $i$-th dimension.  Then the level of a
node is the sum of all elements of the node's coordinate.  For example, in
Figure~\ref{fig:hierarchy}(a), the node with coordinate $(1)$ is at level~1;
in Figure~\ref{fig:hierarchy}(b), the node with coordinate $(4,2)$ is at
level~6.  In multi-dimensional space, multiple nodes can reside at the same
level (e.g., see the 2D case in Figure~\ref{fig:hierarchy}(b)).  We denote the
set of nodes at level~$i$ by $\mathcal{L}(i)$.  

We now formally define HHHs.  Let $S(x) = \sum_{f\preceq x}S(f)$ be the
{\em count} of a key $x$ (i.e., packet count or byte count), 
where $S(f)$ denotes the sum of all $v_f$'s for every flow $f$ under $x$; for
example, if $x$ refers to a subnet, $S(x)$ is the total count of all flows
under the subnet.  Intuitively, a key $x$ is an HHH if $S(x)$
exceeds some pre-defined threshold.  However, if the count of a key exceeds a
threshold, so do the counts of all its ancestors, which cover the key itself.
To concisely define HHHs, we focus on the {\em conditioned count} of a key
\cite{Mitzenmacher2012,BenBasat2017}, defined as the total count of all its
associated flows that do not belong to any HHH.  Specifically, for a key $x$
and a set of HHHs $\mathcal{H}$, the conditioned count of $x$ with respect to
$\mathcal{H}$ is $S_{\mathcal{H}}(x) = \sum S(f): (f\preceq x) \wedge (\nexists
y\in \mathcal{H},$ where $f\preceq y)$.  Thus, we can formally define an HHH
in an inductive fashion: 
\begin{definition} (Hierarchical heavy hitters (HHHs) \cite{Cormode2008}).
\label{def:hhh}
Let $\mathcal{S}$ be the total count of all flows and $\phi$ be a fractional
threshold (where $0 < \phi < 1$). We define $\mathcal{H}_i$ as the set of HHHs
at level~$i$ (where $0 \le i < d$), such that:
\begin{itemize}[leftmargin=*]
\item
$\mathcal{H}_0$ is a set of flows, in which each flow $f$ has count $S(f)
\ge \phi \mathcal{S}$ (i.e., $f$ is a heavy hitter); 
\item
$\mathcal{H}_i = \mathcal{H}_{i-1} \cup \{x: (x\in \mathcal{L}(i)) 
\wedge (S_{\mathcal{H}_{i-1}}(x) \ge \phi\mathcal{S})\}$; and
\item
$\mathcal{H}_{d-1}$ is the set of all HHHs.
\end{itemize}
\end{definition}

We perform HHH detection at fixed-time intervals called {\em epochs}.  Our
goal is to find: (i) the set of all HHHs (whose conditioned counts exceed
$\phi\mathcal{S}$) at the end of each epoch and (ii) the count $S(x)$ of each
key $x$ that is identified as an HHH. 

%-------------------------------------------------------------------------
% Section III: MVPipe Design
%-------------------------------------------------------------------------
\section{MVPipe Design}
\label{sec:overview}

\sysname is a novel invertible sketch for HHH detection, with three major
design goals: lightweight updates, fast convergence, and resource efficiency. 

\sysname builds on the skewness of IP traffic to find HHHs.  Field studies
\cite{Fang1999,Zhang2002,Sarrar2012} show that IP traffic is highly skewed, in
which a small fraction of flows accounts for a majority of traffic. We
argue that the skewness property also holds across aggregation levels.  To
justify, we evaluate the real-world IP traffic traces from CAIDA \cite{caida}
(see \S\ref{subsec:methodology} for details) on the cumulative percentage of
packet counts versus the top percentage of keys at different aggregation levels.
Figure~\ref{fig:skew} plots the results for the 1D-byte and 2D-byte
hierarchies for some aggregation levels in IPv4 traffic.  We observe that the
top-10\% of keys at each level all account for more than 65\% of IP traffic at
that level.

\begin{figure}[!t]
\centering
\begin{tabular}{@{\ }c@{\ }c}
\includegraphics[width=1.65in]{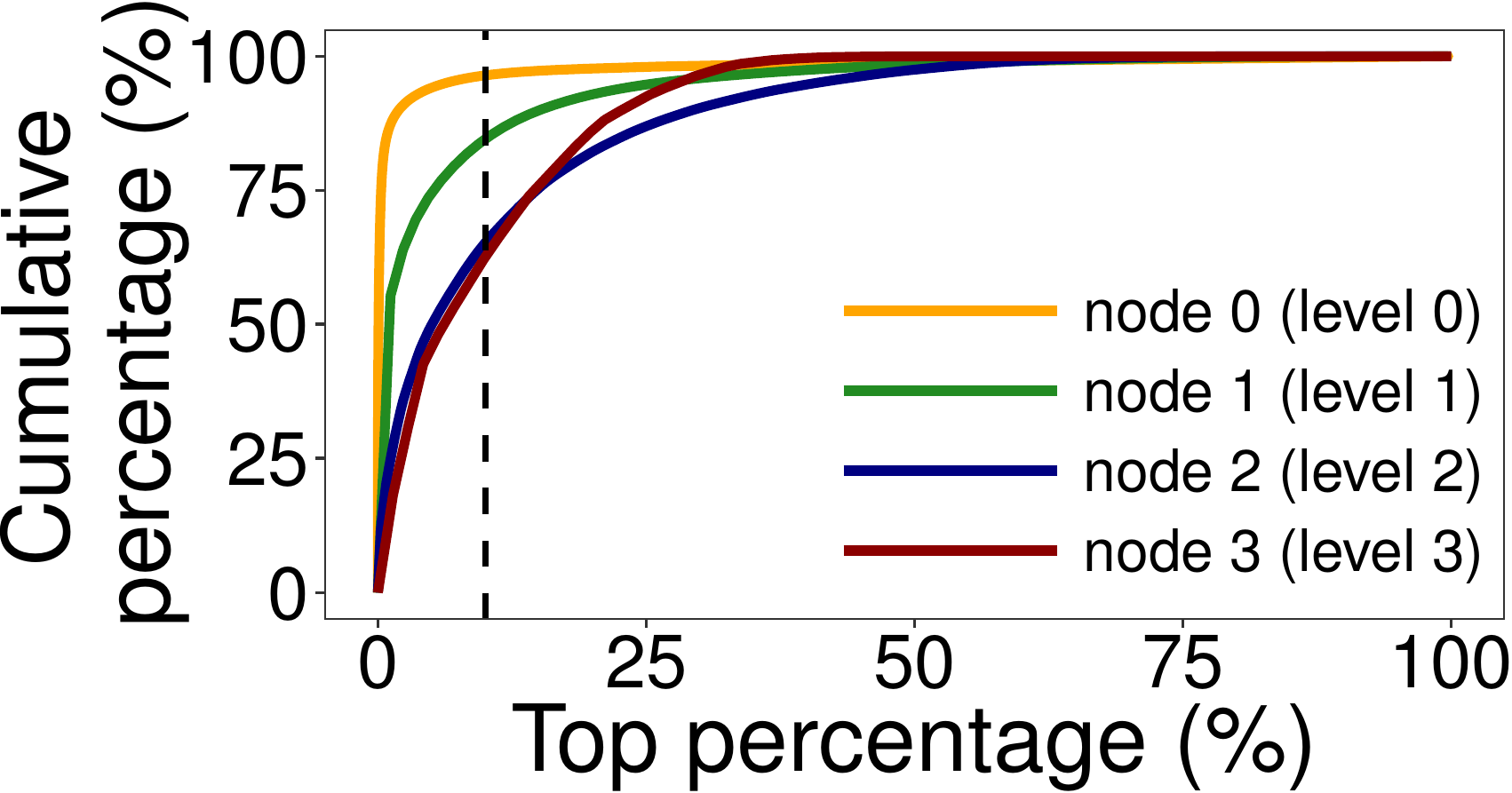} &
\includegraphics[width=1.65in]{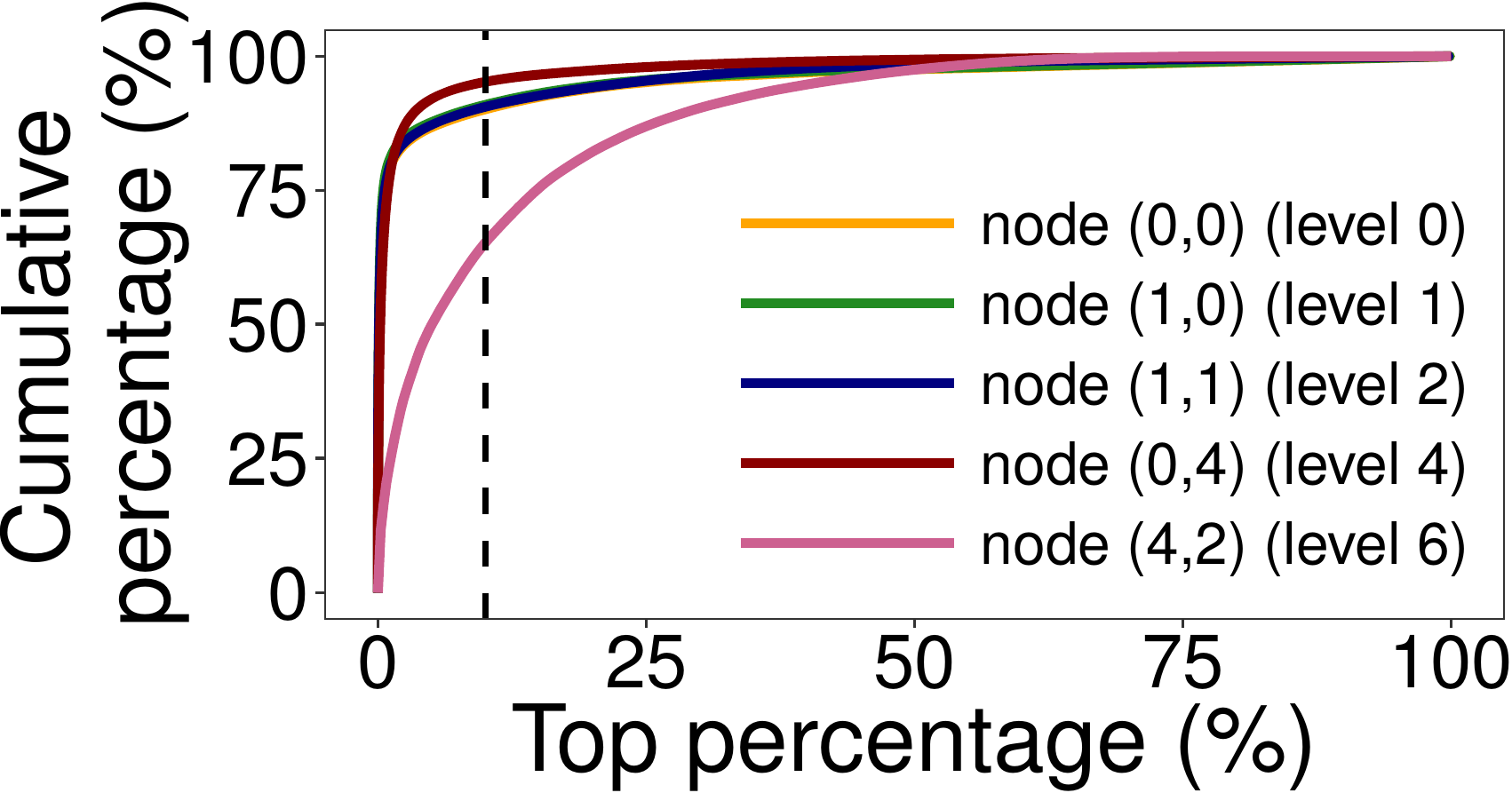} \\
{\small (a) 1D-byte hierarchy} & 
{\small (b) 2D-byte hierarchy}
\end{tabular}
\vspace{-3pt}
\caption{Cumulative percentage of packet counts versus the top-percentage of
keys at different aggregation levels. The dashed line denotes the top-10\%
mark. Here, we focus on IPv4 traffic.}
\label{fig:skew}
\vspace{-6pt}
\end{figure}

\sysname tracks the candidate HHHs that are likely the true HHHs via the
pipelined executions of the {\em majority vote algorithm (MJRTY)}
\cite{Boyer1991}. MJRTY is a one-pass, constant-memory algorithm that finds
the item that has more than half of the occurrences (i.e., the majority item)
in a data stream.  It is proven that if the majority item exists, MJRTY can
always find the majority item \cite{Boyer1991}.  

Based on MJRTY, \sysname maintains an array of buckets for each node in a
hierarchy.  Each bucket performs MJRTY to find the dominant key among all
packets that are hashed to the bucket itself (i.e., the majority item in
MJRTY) as the candidate HHH for the bucket.  Then \sysname processes each
packet $(f, v_f)$ starting from the lowest level~0 (i.e., node~(0) in the 1D
hierarchy or node~(0,0) in the 2D hierarchy) in the hierarchy.  If $f$
does not belong to any candidate HHH at a lower level, \sysname generalizes
$f$ to its ancestor at the next higher level and checks if the ancestor
is a candidate HHH at that level.  \sysname proceeds toward higher levels,
until $f$ is {\em admitted} by a candidate HHH (i.e., the value $v_f$ is
included in the count of the candidate HHH).  

We justify how \sysname achieves its design goals:
\begin{itemize}[leftmargin=*]
\item
{\em Lightweight updates:} By the skewness of IP traffic, the pipelined design
of \sysname stops processing most of the packets at lower-level arrays and
passes only a small fraction of packets to higher-level arrays.  Also, the
processing of each packet in each array of \sysname only contains one hash
computation and one memory access.  Thus, the amortized processing cost is
low. 
\item 
{\em Fast convergence:} \sysname processes every packet (without sampling) in
the same data structure and ensures that any HHH can be detected with high
accuracy at short time scales.  
\item
{\em Resource efficiency:} \sysname requires only primitive computations in
packet processing (e.g., hashing, addition, and subtraction).  Also, \sysname
supports static memory allocation (i.e., its memory space can be pre-allocated
in advance) and incurs limited memory usage.  Such features allow \sysname to
be readily implemented in both hardware and software (\S\ref{sec:imp}).  
\end{itemize}

\subsection{Data Structure}
\label{subsec:data}

Figure~\ref{fig:mvpipe} shows the data structure of \sysname. It comprises $H$
arrays, denoted by $A_0, A_1, \cdots, A_{H-1}$, where $H$ is the number of
nodes in the hierarchy.  Each array $A_i$ (where $0\le i < H$) contains $w_i$
buckets and corresponds to one node in the hierarchy.  Let $B(i,j)$ be the
$j$-th bucket in array $A_i$, where $0 \le j < w_i$. Each bucket $B(i,j)$
consists of four fields: (i) $K_{i,j}$, which stores the candidate HHH in the
bucket; (ii) $V_{i,j}$, which is the total count of all keys hashed to the
bucket;  (iii) $I_{i,j}$, which is the indicator counter that checks if the
current candidate HHH in $K_{i,j}$ should be kept or replaced by MJRTY; and
(iv) $C_{i,j}$, which is the cumulative count of the candidate HHH since it is
stored in $K_{i,j}$.  \sysname is associated with $H$ pairwise-independent
hash functions $h_0, h_1 \cdots h_{H-1}$, such that each $h_i$ (where $0\le i
< H$) hashes the generalization of the key of each incoming packet to one of
the $w_i$ buckets in $A_i$.

\begin{figure}[!t]
\centering
\includegraphics[width=3.2in]{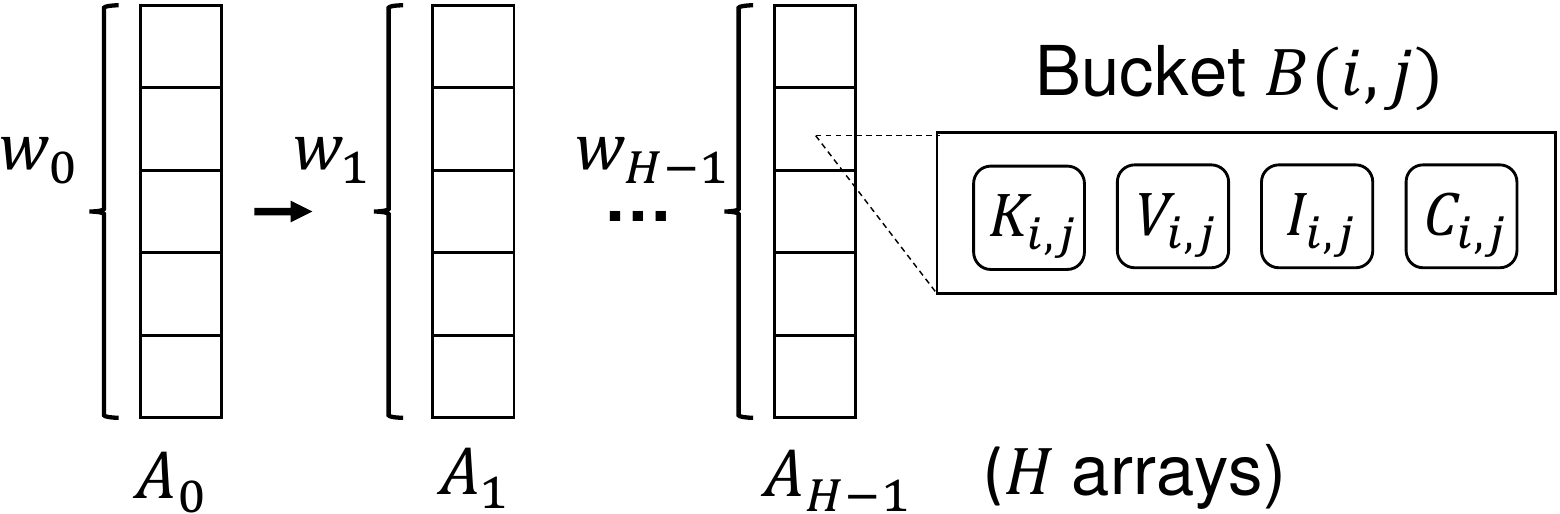}
\caption{Data structure of \sysname.  Each bucket $B(i,j)$ (where $0\le i < H$
and $0\le j < w_i$) has four fields: $K_{i,j}$, $V_{i,j}$, $I_{i,j}$,
and $C_{i,j}$.}
\label{fig:mvpipe}
\vspace{-6pt}
\end{figure}

\sysname currently associates a single array with each node in the hierarchy.
It can improve the HHH detection accuracy by associating multiple arrays with
each node in the hierarchy, at the expense of degraded update performance.  We
discuss the trade-off in \S\ref{subsec:comparisons}. 

\subsection{1D HHH Detection}
\label{subsec:1d}

We first consider the operations of \sysname in 1D HHH detection, whose
pseudo-code is shown in Figure~\ref{fig:ops1d}. \sysname
supports two major operations: (i) {\em Update}, which updates each packet
$(f, v_f)$ into the data structure; and (ii) {\em Detect}, which returns the
set of all HHHs and their respective estimated counts from the data structure. 
It also builds on two functions: (i) {\em Push},
which pushes the updates of a key and its corresponding count along the
arrays; and (ii) {\em Estimate}, which returns an estimated count of a key 
in its hashed bucket.  Note that in 1D HHH detection, each node in a hierarchy
corresponds to a distinct level.  Thus, each array $A_i$ corresponds to
level~$i$, where $0\le i < H$.

\begin{figure}[t]
\begin{algorithmic}[1]
\small
\Function{Push}{key, value, $l$} 
\State $(x, v_x) \leftarrow (\textrm{key}, \textrm{value})$
\For {$i = l$ to $H-1$}
  \State $x \leftarrow $ generalize $x$ at level $i$
  \State $V_{i,h_i(x)} \leftarrow V_{i,h_i(x)} + v_{x}$
  \If {$K_{i,h_i(x)} = x$} 
    \State $I_{i,h_i(x)} \leftarrow I_{i,h_i(x)} + v_{x}$
    \State $C_{i,h_i(x)} \leftarrow C_{i,h_i(x)} + v_{x}$
    \State \Return 
    \ElsIf {$I_{i, h_i(x)} \ge v_{x}$}
    \State $I_{i,h_i(x)} \leftarrow I_{i,h_i(x)} - v_{x}$
    \Else  \Comment{$K_{i,h_i(x)} \ne x$ and $I_{i,h_i(x)} < v_x$}
      \State $I_{i,h_i(x)} \leftarrow v_{x} - I_{i,h_i(x)}$
      \State swap $(K_{i,h_i(x)}, C_{i,h_i(x)})$ and $(x, v_x)$
  \EndIf
\EndFor
\EndFunction
\smallskip
\Function{Estimate}{$x$, $l$, $t$}
  \State $U_l\leftarrow (V_{l,h_l(x)} + I_{l,h_l(x)})/2$  
  \State $v\leftarrow C_{l,h_l(x)}$  
  \For {$i = l+1$ to $\min\{l+t, H-1\}$}
    \State $y \leftarrow $ generalize $x$ at level~$i$
    \If {$K_{i,h_i(y)}$ = $y$}
      \State $U_i = (V_{i,h_i(y)} + I_{i,h_i(y)})/2 + v$
      \State $v \leftarrow v + C_{i,h_i(y)}$
    \Else
      \State $U_i = (V_{i,h_i(y)} - I_{i,h_i(y)})/2 + v$
    \EndIf
    \EndFor
    \State \Return $\min_{l\le i\le \min\{l+t, H-1\}}\{U_i\}$
\EndFunction
\smallskip
\Procedure{Update}{$f,v_{f}$}
\State\textsc{Push}($f, v_{f}, 0$)
\EndProcedure
\smallskip
\Procedure{Detect}{}
    \State $\mathcal{A}, \mathcal{H}_o, \mathcal{H}_1, \dots , \mathcal{H}_{H-1} \leftarrow \emptyset$
\For {$i = 0$ to $H-1$}
    \For {$j = 0$ to $w_i-1$}
      \State $x \leftarrow K_{i,j}$
      \State $\hat{S}_{\mathcal{H}_{i-1}}(x) \leftarrow $ \textsc{Estimate}($x$, $i$, $t$)
      \Comment $\mathcal{H}_{i-1}=\emptyset$ for $i=0$
      \If {$\hat{S}_{\mathcal{H}_{i-1}}(x) > \phi\mathcal{S}$}
        \State $\hat{S}(x) \leftarrow \hat{S}_{\mathcal{H}_{i-1}}(x) +
        \sum_{x'\in \mathcal{H}_{i-1} \wedge x'\prec x}
        C_{i', h_{i'}(x')}$ 
		\Comment $i'$ is the level of $x'$
        \State $\mathcal{A} \leftarrow \mathcal{A} \cup {(x, \hat{S}(x))}$
      \Else
        \State\textsc{Push}($x, V_{i,j}, i+1$)
      \EndIf
    \EndFor
    \State $\mathcal{H}_{i}  \leftarrow \mathcal{H}_{i-1} \cup \mathcal{A}$
    \State $\mathcal{A}  \leftarrow \emptyset$
\EndFor
    \State \Return $\mathcal{H}_{H-1}$
\EndProcedure
\end{algorithmic}
\vspace{-3pt}
\caption{Major operations for 1D HHH detection.}
\label{fig:ops1d}
\vspace{-6pt}
\end{figure}

\paragraph{Update operation.}  We apply the Update operation to insert each
incoming packet to \sysname, starting from $A_0$.  At a high level, we hash
the flow key of each packet to one of the buckets in $A_0$ and check if the
flow key is a candidate HHH in the bucket via MJRTY. If so, we end the update; 
otherwise, we generalize the flow key to its ancestor at the next higher
level~1 and continue to insert the ancestor to $A_1$.  We update $A_1$ and the
remaining arrays in a similar way until the flow key is admitted by a
candidate HHH.  During the process, if the original candidate HHH stored in
\sysname is replaced by the current generalized flow key due to MJRTY, we
generalize the original candidate HHH and insert it into higher-level arrays. 

We elaborate on the Update operation (Lines~26-27 of Figure~\ref{fig:ops1d}) as
follows.  At the beginning of each epoch, we initialize the counters of all
buckets of \sysname to zeros.  We update each incoming packet $(f, v_f)$
starting from array $A_0$ of \sysname by calling the {\sf Push}($f, v_f, 0$)
function, which processes $f$ starting from level~0 until $f$ is admitted by a
candidate HHH in one of the levels. 

The Push function (Lines~1-14 of Figure~\ref{fig:ops1d}) takes a key, its
associated value, and the array index $l$ (where $0\le l< H$) as input.
We first initialize $(x, v_x)$ from the input, where $(x, v_x)$ is passed
along the arrays at higher levels (Line~2).  For each array $A_i$ (where $l
\le i < H$), we generalize $x$ at level~$i$ (Line~4) and hash the new $x$ into
the bucket $B(i, h_i(x))$.  We increment $V_{i,h_i(x)}$ by $v_x$ and check if
$x$ should become the candidate HHH of the bucket based on MJRTY.
Specifically, if $K_{i,h_i(x)}$ equals $x$ (i.e., $x$ is already a candidate
HHH), we increment both $I_{i,h_i(x)}$ and $C_{i,h_i(x)}$ by $v_x$ and return
(Lines~6-9); else if $x$ is not the candidate HHH and $I_{i,h_i(x)}$ is at
least $v_x$, we decrement $I_{i,h_i(x)}$ by $v_x$ (Lines~10-11); otherwise, if
$x$ is not the candidate HHH and $I_{i,h_i(x)}$ is below $v_x$, it means that
$x$ should now become the new candidate HHH in $B(i,j)$.  Then we should
update $K_{i,h_i(x)}$ and $C_{i,h_i(x)}$ with $x$ and $v_x$, respectively, and
aggregate the count of the original key in $K_{i,h_i(x)}$ to the next higher
level.  More precisely, we set $I_{i,h_i(x)}$ to $v_x - I_{i,h_i(x)}$ and swap
$(K_{i,h_i(x)}, C_{i,h_i(x)})$ and $(x, v_x)$ (Lines~12-14).  

From the Update operation, it is clear that once the candidate HHH is stored
in $K_{i,j}$, its subsequent values received at level $i$ (i.e. $C_{i,j}$) are
not pushed to higher levels. In other words, the cumulative count of a
candidate HHH is not aggregated to any of its ancestors at higher levels. 

\begin{figure}[!t]
\centering
\includegraphics[width=3.3in]{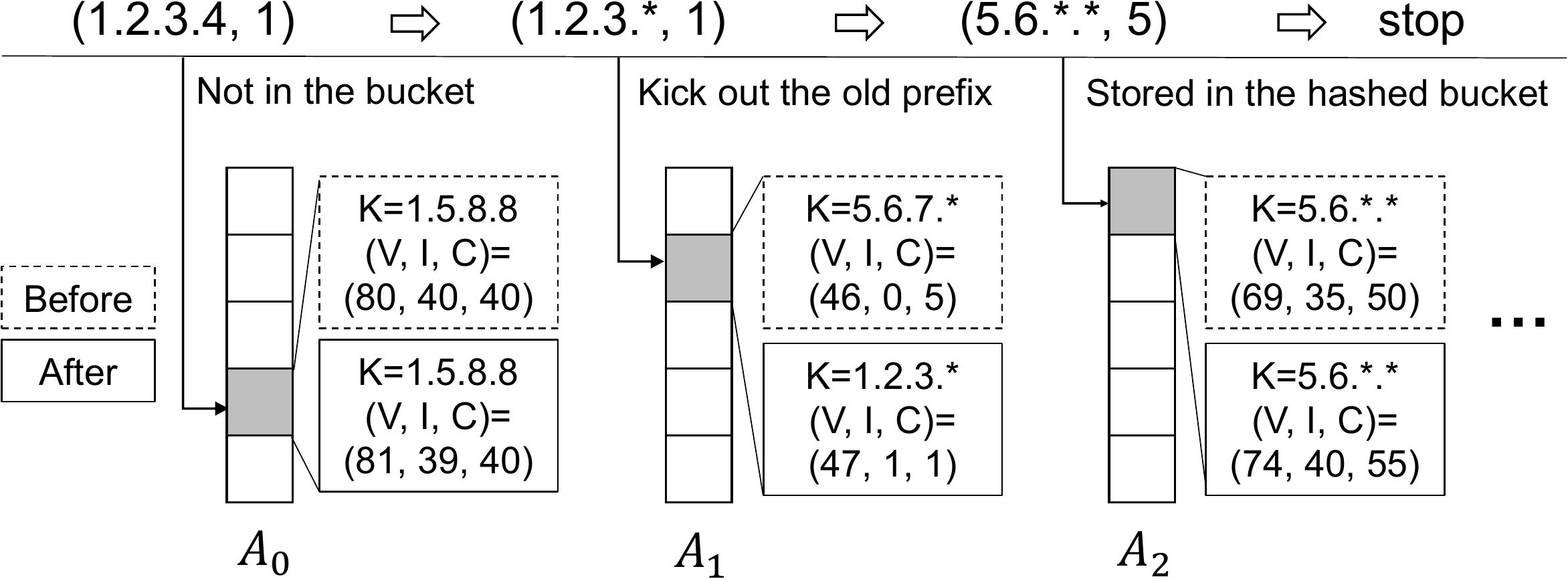} 
\vspace{-3pt}
\caption{Example of the Update operation in \sysname.}
\label{fig:update}
\vspace{-6pt}
\end{figure}

{\em Example.} Figure~\ref{fig:update} depicts the Update operation.  Suppose
that a packet $(f, v_f) = (1.2.3.4, 1)$ arrives.  If, say, 1.2.3.4 is not the
candidate HHH in the hashed bucket in $A_0$ and the branch in Lines~10-11
holds, we generalize 1.2.3.4 into 1.2.3.* and proceed to the next level.  If,
say, 1.2.3.* is not the candidate HHH in the hashed bucket in $A_1$ but the
branch in Line~12-14 holds (i.e., it now becomes the candidate HHH), we
substitute the original candidate HHH 5.6.7.* by 1.2.3.*.  We now generalize
5.6.7.* into 5.6.*.* and push (5.6.*.*, 5) to the next level.  If,
say, 5.6.*.* is the candidate HHH in $A_2$ (i.e., the branch in Lines~6-9
holds), the Update operation updates the counters and terminates (i.e., $f$ is
now admitted by 5.6.*.*). 

\paragraph{Detect operation.} We apply the Detect operation at the end of each
epoch to find all HHHs and their estimated counts.  At a high level, we
traverse all arrays of \sysname and check if the candidate HHH in each bucket
should be reported as an HHH.  We start from the array $A_0$ and check each
candidate HHH: if the estimated conditioned count of the candidate HHH exceeds
the threshold, we treat the candidate HHH as an HHH and include both the
candidate HHH and its estimated count
into the output set $\mathcal{H}_{H-1}$; otherwise, the candidate HHH should
not be reported as an HHH and its key and cumulative count should be further
pushed to higher-level arrays (as one of its ancestors may be reported as an
HHH).  

We elaborate on the Detect operation (Lines~28-41 of Figure~\ref{fig:ops1d}) as
follows.  Let $\hat{S}_{\mathcal{H}_{i-1}}(x)$ and $ \hat{S}(x)$ be the
estimated conditioned count at level~$i$ and the estimated count of key $x$,
respectively. For each bucket $B(i,j)$ (where $0 \le i < H$ and $0 \le j <
w_i$), we call the Estimate function to return the estimated conditioned count
$\hat{S}_{\mathcal{H}_{i-1}}(x)$ of key $x$ stored in $K_{i,j}$
(Lines~32-33). If $\hat{S}_{\mathcal{H}_{i-1}}(x)$ exceeds the threshold
$\phi\mathcal{S}$, we further calculate the estimated count $\hat{S}(x)$
and add $x$ to $\mathcal{A}$, where $\mathcal{A}$ is the set of detected HHHs
at level~$i$ (Lines~34-36).  Otherwise, we call the Push function to push ($x,
C_{i,j}$) to the next higher level $i+1$ (Lines~37-38). After processing array
$A_i$, we update $\mathcal{H}_i$ as the union of $\mathcal{H}_{i-1}$ and
$\mathcal{A}$, and reset $\mathcal{A}$ to empty (Line~39-40).

For each reported HHH $x$ at level~$i$, we need to calculate its
estimated count $\hat{S}(x)$.  From the Update operation, we note that the
cumulative count of each of $x$'s descendants that are reported as HHHs in
$\mathcal{H}_{i-1}$ is not pushed to the hashed bucket of $x$.  
Thus, we need to include such ``missing'' counts in the calculation of
$\hat{S}(x)$.  We calculate $\hat{S}(x)$ as the sum of: (i) the estimated
conditioned count $\hat{S}_{\mathcal{H}_{i-1}}(x)$ is returned by the Estimate
function (see details below) and (ii) the cumulative counts of $x$'s
descendants that are reported as HHHs in lower levels (Line~35).  We analyze
the error bound of $\hat{S}(x)$ in \S\ref{sec:theory}. 

The Estimate function (Lines~15-25 of Figure~\ref{fig:ops1d}) returns the
estimated conditioned count of a given key $x$ in its hashed bucket.  The
function takes a key $x$, the level $l$, and the configurable number $t$ of
the ancestors that are checked in estimation (see details below).  We start
from array $A_l$ and obtain the upper bound $U_l$ of $x$ in $B_{l,h_l(x)}$ as
$(V_{l,h_l(x)}+I_{l,h_l(x)})/2$ (Line~16); in MJRTY, the count of $x$ in
$I_{l,h_l(x)}$ is decremented by other keys in the same bucket by at most
$(V_{l,h_l(y)}-I_{l,h_l(y)})/2$.  We show that $U_l$ is the upper bound of the
true count of $x$ tracked in its hashed bucket (Lemma~\ref{lem:bound} of
\S\ref{sec:theory}).

Note that $x$ may have hash collisions with some large keys in $A_l$, and
$U_l$ severely overestimates the true count of $x$.  To reduce the collision
error, we introduce the configurable parameter $t$, through which we access
$t$ additional arrays and further check the estimated counts of the $t$
closest ancestors of $x$ from $A_{l+1}$ to $A_{l+t}$ (if $l+t$ is beyond the
maximum number of arrays $H-1$, we stop in $A_{H-1}$).  The idea here is that
when we include the cumulative count of $x$ in the estimated count of each
ancestor, if the estimated count of the ancestor is smaller than $U_l$, it
implies that $x$ collides with some large keys in $A_l$ and hence $U_l$ is
severely overestimated.  Thus, we use the minimum value of $U_l$ (the
estimated count of $x$ in $A_l$) and $U_i$'s (the estimated counts of $t$
ancestors of $x$ from $A_{l+1}$ to $A_{l+t}$) as the final estimated count of
$x$ (Line~25).  The parameter $t$ determines the performance-accuracy
trade-off in HHH detection: a large $t$ means fewer false positives, but
incurs more time to find all HHHs. 

Given $t$, the Estimate function proceeds as follows. For each array $A_i$,
where $l+1 \le i \le \min\{l+t, H-1\}$, we set $y$ as the generalization of
$x$ at level~$i$ (Line~19). We calculate the estimate of $y$:  if
$K_{i,h_i(y)}$ equals $y$, $U_i$ is $(V_{i,h_i(y)}+I_{i,h_i(y)})/2+v$
(Line~21); otherwise, $U_i$ is $(V_{i,h_i(y)}-I_{i,h_i(y)})/2+v$ (Line~24).
The term $v$, which we initialize as the cumulative counter value
$C_{l,h_l(x)}$ (Line~17), refers to the cumulative count of $x$ that should be
included in $y$ when $U_i$ is calculated.  Finally, we return the minimum
value among $U_i$'s, where $l \le i \le \min\{l+t, H-1\}$ (Line~25). 

{\em Example.} Figure~\ref{fig:detect} depicts the Detect operation.  We set
the threshold as 30 and $t=2$. At the end of an epoch, we start from $A_0$ to
check every bucket.  Suppose that we check bucket $B(0,3)$.  We first
calculate the estimated conditioned count of the candidate HHH 1.5.8.8 in
$B(0,3)$ via the Estimate function. Suppose that we have $U_0=32$ and $v=8$.
We generalize 1.5.8.8 to 1.5.8.* and find that it is the candidate HHH in the
hashed bucket in $A_1$ (i.e., the branch in Lines~20-22 holds).  We obtain
$U_1=23$ and add the cumulative count $C_{1,3}=14$ to $v$, so we now have
$v=22$.  We continue to generalize 1.5.8.* to 1.5.*.* and find that the
generalized key is not the candidate HHH (which is now 1.6.*.*) in $A_2$
(i.e., the branch in Lines~23-24 holds).  We obtain $U_2=26$.  Thus, the
returned estimated conditioned count of 1.5.8.8 is $\min\{U_0, U_1, U_2\}=23$,
which is smaller than the threshold 30 (i.e., the branch Lines~37-38 holds).
We then push 1.5.8.8 to higher levels: we generalize it to 1.5.8.* and update
the hashed bucket in $A_1$.  If, say, 1.5.8.* is already stored in the bucket,
we increment each of the three counters of the bucket by 8 and finish the
checking of $B(0,3)$. 

\begin{figure}[!t]
\centering
\includegraphics[width=3.3in]{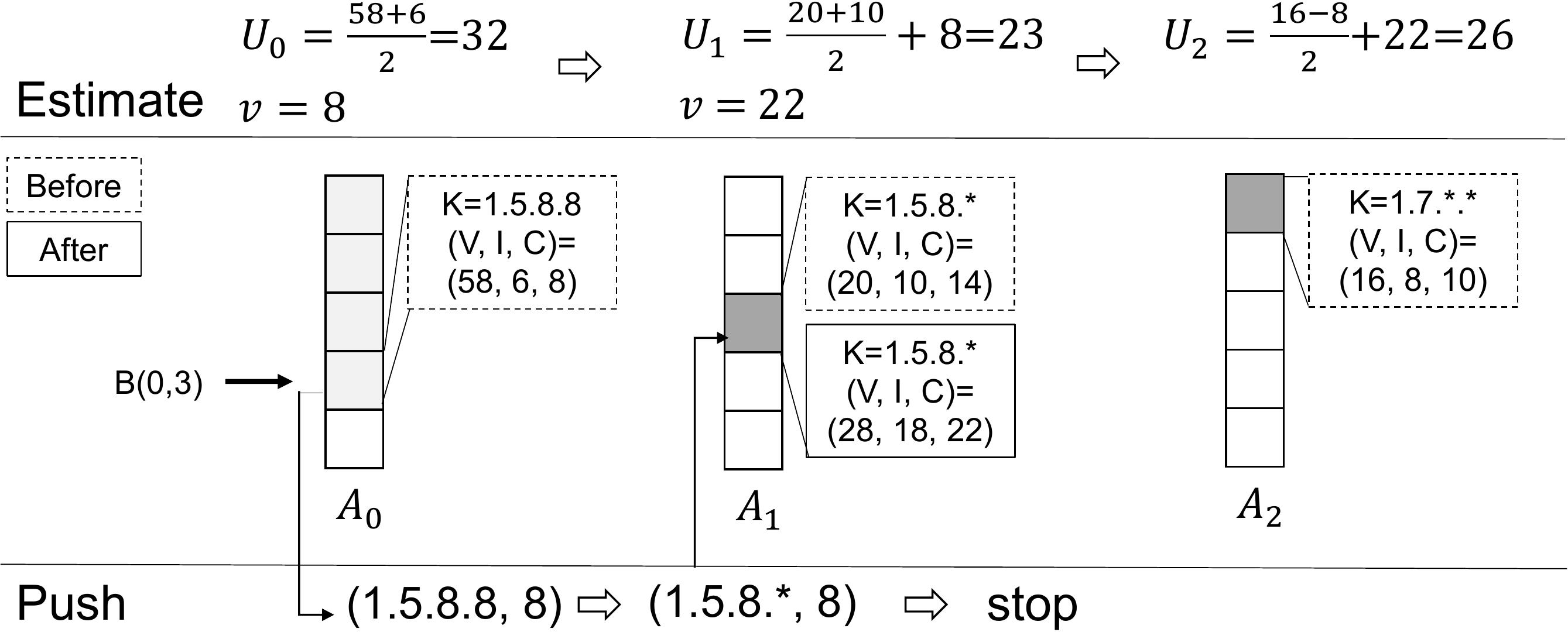} 
\vspace{-3pt}
\caption{Example of the Detect operation in \sysname.}
\label{fig:detect}
\vspace{-6pt}
\end{figure}

\subsection{2D HHH Detection}
\label{subsec:2d}

We extend \sysname to 2D HHH detection, in which the generalization relation 
now forms a lattice structure (Figure~\ref{fig:hierarchy}(b)).  Similar to
1D HHH detection, we maintain an array of buckets for each node of the lattice
to track the candidate HHHs.  We briefly describe the Update and Detect
operations in 2D HHH detection; their pseudo-code is in the supplementary
file.

\paragraph{Update operation.}  In 2D HHH detection, we need to address the
generalization order and the stop condition in the Update operation.  Unlike
1D HHH detection, which only has one generalization direction, a key in 2D HHH
detection has two
generalization directions: the source direction (i.e., from left to right) and
the destination direction (from bottom to up).  For example, the address pair
(1.2.3.4, 5.6.7.8) can be generalized to either (1.2.3.*, 5.6.7.8) or
(1.2.3.4, 5.6.7.*) by a single byte. To represent the lattice structure in
\sysname, we index the arrays of \sysname first along the destination
direction (from bottom to up), followed by along the source direction (from
left to right).  For example, array $A_0$ corresponds to node (0,0), while 
array $A_{22}$ corresponds to node (4,2) in Figure~\ref{fig:hierarchy}(b)
(recall that the number of arrays is the number of nodes in the lattice). 

We enforce that the generalization of a key along the source direction only
applies to the {\em bottom nodes} in the lattice (i.e., nodes (0,0), (0,1),
\dots, (0,4) in Figure~\ref{fig:hierarchy}(b)).  Specifically, we update each
incoming packet in \sysname starting from level~0 (i.e., node (0,0)).  If a
key is not admitted by a candidate HHH, we push the key in two generalization
directions: (i) along the destination direction from the bottom node and (ii)
the next bottom node.  We stop the generalization until a key is admitted by
the candidate HHH in a bottom node. 
  
We use Figure~\ref{fig:example} to show the idea of the Update
operation.  For the key (1.2.3.4, 5.6.7.8), we first insert it to array $A_0$
in node (0,0). If it is not a candidate HHH, we push the key along the
destination direction until it is admitted by some ancestor (e.g., node~$o$).
We also push the key to the next bottom node (1,0). If, say, the generalized
key (1.2.3.*, 5.6.7.8) is still not a candidate HHH, we push it along the
destination direction until it is admitted by some ancestor (e.g., node~$x$).
Similar operations apply to node (2,0) and its destination direction.  We
terminate until the key is admitted by a bottom node (e.g., node~$p$).   

\paragraph{Detect operation.} Since we now push the count of a key along both
the source direction along the bottom nodes and the destination direction from
the bottom nodes, the count of a key can contribute to multiple ancestors,
leading to the {\em double counting} problem
\cite{Cormode2008,Mitzenmacher2012}.  For example, the key (1.2.3.4, 5.6.7.8) is
pushed along the arrows in Figure~\ref{fig:example} according to the Update
operation. Suppose that its count is counted by both of its ancestors $x$ and
$y$ that are both reported to be HHHs.  Now, we calculate the conditioned count
of $z$, which is the common ancestor of both $x$ and $y$.  By the definition
of the conditioned count, we need to subtract the counts of both $x$ and $y$
from $z$, but doing so will deduct the count of (1.2.3.4, 5.6.7.8) twice from
$z$.

The Detect operation addresses the double counting problem based on the
inclusion-exclusion principle \cite{Cormode2008, Mitzenmacher2012}, whose idea
is that after subtracting all descendants that are HHHs, we add back the
count that is discounted twice.  At the end of an epoch, the Detect operation 
checks the candidate HHH in each bucket.  If the estimated conditioned count
of a candidate HHH exceeds the threshold, we add it to the output set;
otherwise, we push the candidate HHH to higher-level nodes, either in the
source direction along the bottom nodes or in the destination direction from
the bottom nodes. 

\begin{figure}[!t]
\centering
\includegraphics[width=1.5in]{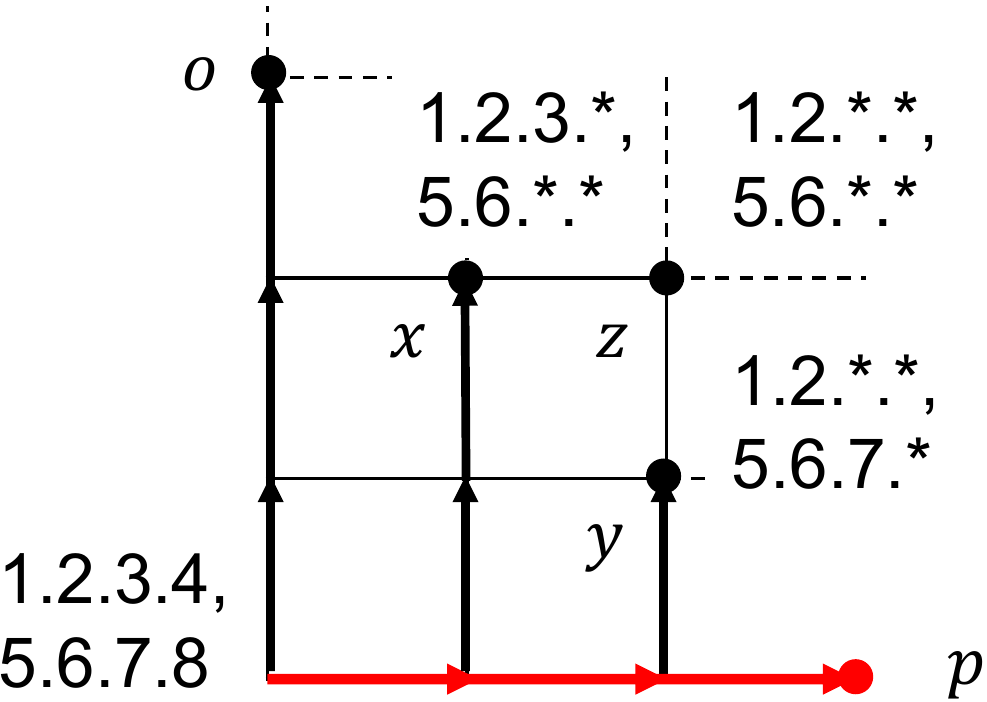} 
\vspace{-3pt} 
\caption{Illustration for the major operations in 2D HHH
detection. The generalization is applied either in the source direction
along the bottom nodes or in the destination direction from the bottom nodes.}
\label{fig:example}
\vspace{-6pt}
\end{figure}

\subsection{Discussion}

MJRTY has also been adopted in HH detection (i.e., no hierarchy
awareness) \cite{Kallitsis2016,Tang2019}.  In
particular, MV-Sketch \cite{Tang2019} maintains multiple rows of buckets, in
which each bucket tracks a candidate HH using MJRTY.  It hashes each packet
into a bucket in each of the rows, and updates the counters in each bucket.
It uses multiple rows of buckets to resolve the hash collisions of HHs into
the same bucket.  In contrast to \sysname, MV-Sketch implements a single stage
of MJRTY and is not hierarchy-aware, while \sysname forms multiple pipelined
stages of MJRTY for HHH detection. 

\sysname is much beyond a simple extension of MV-Sketch in HHH detection. A
na\"ive way to extend MV-Sketch for HHH detection is to maintain an
instance of MV-Sketch for every node of the IP address hierarchy.  For each
packet, we compute all its generalizations and insert each of them
independently into the corresponding instances, so that the HHHs can later be
recovered as the HHs of each aggregation level.  However, such a
na\"ive approach incurs substantial update overhead, as it updates all
MV-Sketch instances for each incoming packet.   In contrast, \sysname forms
multiple pipelined stages of MJRTY that correspond to different nodes in the
IP address hierarchy, and maintains one array for each node.  With the
skewness of IP traffic, \sysname only needs to update a small number of arrays
per packet (and one array for most of the time) and hence achieves lightweight
updates.  Also, to resolve hash collisions, \sysname checks
additional arrays in higher levels when estimating the value of a candidate
HHH (see the Estimate function in Figure~\ref{fig:ops1d}) and hence maintains
high accuracy.

Currently, we focus on 1D and 2D HHH detection. \sysname can be
extended for higher dimensions, yet both its space and time complexities
become the multiplication of the depths of all dimensions.  Nevertheless,
practical applications do not need to consider general HHH detection in all
dimensions \cite{Zhang2004}.

In addition, the examples presented in this paper are based on
IPv4 traffic, yet \sysname can also be directly applied to IPv6 traffic
without modification.  We present our evaluation results for IPv6 traffic in
the supplementary file.

%-------------------------------------------------------------------------
% Section IV: Implementation
%-------------------------------------------------------------------------
\section{Implementation}
\label{sec:imp} 

We have implemented \sysname in both software and hardware.  In particular,
our hardware implementation of \sysname is written in P4 \cite{p4}, and
addresses the limitations of the existing hardware-based HHH detection schemes
(\S\ref{sec:introduction}) in the following ways: no reliance on a controller
for counter updates, using SRAM rather than TCAM (which is more expensive and
scarce) for HHH counting, and avoids sampling of packets or aggregation levels
for fast convergence. 

\subsection{Software Implementation}

We have built a software version of \sysname in C++ with around 800~LoC
(including both 1D and 2D HHH detection).  \sysname uses MurmurHash
\cite{murmurhash} as the hash function.  It can be integrated into the data
plane of a software switch \cite{Pfaff2015} for real-time HHH detection.
Specifically, the software switch inserts the IP packet header of each
incoming packet to an in-memory buffer, from which \sysname fetches and
processes the packet headers.  We now implement \sysname in a single thread
running on a single CPU core. 

\subsection{P4 Implementation}

We implement \sysname in P4 \cite{p4} with around 900~LoC and compile it into
a Tofino switch \cite{tofino}.  Our P4 implementation is based on 
the Protocol Independent Switch Architecture (PISA)
\cite{Bosshart2013,Bosshart2014}.  A PISA switch first extracts the header
fields of each incoming packet via a programmable parser. It then passes the
packet to an ingress pipeline of stages, each of which contains a series of
match-action tables.   Each stage matches the extracted header fields of the
packet with the entries in the match-action tables, and applies the matched
actions to modify the extracted header fields and/or update the persistent
state at the stage.  Afterwards, the switch passes the packet to a similar
egress pipeline and emits the packet.  

Programmable switches offer rich computational capability in addition to
packet forwarding, yet they pose stringent hardware resource constraints 
\cite{Basat2018}. For example, a programmable switch typically has limited
SRAM (e.g., few megabytes) and a small number of stateful arithmetic-logic
units (ALUs) per stage.  Such
constraints pose two challenges to our \sysname implementation.  First, each
stage in the ingress/egress pipeline can only access a memory block {\em once}
per packet, with only a single read-modify-write operation.  Second, each stage
only supports an if-else chain with at most two branches. 

In the following, we describe how we address the above constraints and adapt
\sysname into a programmable switch.  Here, we only focus on the 1D-byte HHH
detection due to the limited number of stages available in a switch, while the
detection for other granularities is posed as future work.  For each level
of the 1D-byte hierarchy, we create four register arrays, denoted by $K$, $V$,
$I$, and $C$, which correspond to the four bucket fields in \sysname
(Figure~\ref{fig:mvpipe}). 

\paragraph{Using pairs atoms to update dependent fields.}  One implementation
challenge is that the Push function (Figure~\ref{fig:ops1d}) makes
inter-dependent accesses to $K$ and $I$:  the write to $I$ depends on the
value of $K$ (Lines~6-7), yet the write to $K$ is conditioned on the value of
$I$ (Lines~12-14).  It is infeasible to update both $K$ and $I$ with a single
read-modify-write. 

We resolve the inter-dependency of updating both $K$ and $I$ using the {\em
pairs atoms} \cite{Sivaraman2016}, which are natively supported atomic
operations in PISA.  A pairs atom reads two 32-bit elements from a register
array,
performs conditional branching and primitive arithmetic on both elements, and
writes back the results. At the same time, it can output either the original
value of an element or the computation results to a specified metadata field.
In our case, we pack $K$ and $I$ to the upper 32 bits and lower 32 bits in a
64-bit register array, respectively.  We then update them using a pairs atom. 

\paragraph{Reducing the updates to the indicator counter.}  Referring to the
Push function in Figure~\ref{fig:ops1d}, we split the updates to the bucket
fields into three branches: Lines~6-9, Lines~10-11, and Lines~12-14.  Only the
indicator counter (i.e., the register array $I$) will be updated in each of
the three branches. This needs a three-branch if-else chain to update $I$,
which is infeasible as only a two-branch if-else chain is supported in each
stage. 

To fit the update of $I$ into a two-branch if-else chain, we discard the
update of $I$ in the third branch (i.e., Line~13 in Figure~\ref{fig:ops1d});
in other words, the indicator counter will not be updated if a key is not a
candidate HHH in $K$ but has a count larger than the current indicator counter
$I$.  The rationale here is that the condition rarely happens in skewed
workloads, in which an HHH is quickly tracked in $K$ and is unlikely (albeit
possible) substituted by a different key.  Thus, we (slightly) sacrifice
the accuracy to reduce the update of $I$ to only two (instead of three)
branches.  Now all bucket fields can be updated with a two-branch if-else
chain in a single stage. 

\begin{figure}[t]
\small
\begin{algorithmic}[1]

\Function{Push}{$l$}
\State $(x, v_x) \leftarrow $ (Meta.key$_l$, Meta.val$_l$)
\State // Stage 1: update $V_{l,h(x)}$ and $(K_{l,h(x)}, I_{l,h(x)})$ 
\State $V_{l,h(x)} \leftarrow V_{l, h(x)} + v_{x}$    
\If {$K_{l, h(x)} = x$} \Comment Case 1
    \State $I_{l, h(x)} \leftarrow I_{l, h(x)} + v_x$
\Else                                          \Comment Case 2
    \State $I_{l, h(x)} \leftarrow I_{l, h(x)} - v_x$
\EndIf
\If {$K_{l, h(x)} \neq x$ and $I_{l, h(x)} < v_x$} \Comment Case 3
    \State $K_{l, h(x)} \leftarrow x$
\EndIf
\If {$K_{l, h(x)} = x$ or $I_{l, h(x)} < v_x$} \Comment Case 1 or 3
    \State Meta.key$_{(l+1)}$ $\leftarrow K_{l, h(x)}$
\EndIf
\State // Stage 2: set metadata fields based on the last Stage  
\If {Meta.key$_{(l+1)} = 0$} \Comment Case 2 
    \State Meta.key$_{(l+1)}$ $\leftarrow $ generalize $x$ at level $l+1$
    \State Meta.val$_{(l+1)}$ $\leftarrow v_x$

\Else
    \State Meta.flag$_l \leftarrow $ Meta.key$_{(l+1)}$ xor $x$
\EndIf

\State // Stage 3: update $C_{l, h(x)}$ based on metadata set in Stage 2 
\If {Meta.flag$_l = 0$} \Comment Case 1
\State $C_{l,h(x)} \leftarrow C_{l, h(x)} + v_x$
\Else       \Comment Case 3 
    \State Meta.val$_{(l+1)} \leftarrow C_{l,h(x)}$
    \State $C_{l,h(x)} \leftarrow v_x$
    \State Meta.key$_{(l+1)} \leftarrow $ generalize Meta.key$_{(l+1)}$ at level $l+1$
\EndIf
\EndFunction
\Procedure{Update}{}
\State\textsc{Push}($0$)
\If {Meta.val$_{1} \neq 0$} \State \textsc{Push}($1$)
\EndIf
\If {Meta.val$_{2} \neq 0$} \State \textsc{Push}($2$)
\EndIf
\If {Meta.val$_{3} \neq 0$} \State \textsc{Push}($3$)
\EndIf 
\If {Meta.val$_{4} \neq 0$}
\State $V_{4} \leftarrow V_{4} +$ Meta.val$_{4}$
\EndIf
\EndProcedure
\end{algorithmic}
\vspace{-3pt}
\caption{Pseudo-code of the update operation for 1D-byte HHH detection in P4.}
\label{fig:tofino}
\vspace{-6pt}
\end{figure}

\paragraph{Putting it all together.} Figure~\ref{fig:tofino} shows the
pseudo-code of our P4 implementation.  For each level
$l$, where $0\le l \le 4$, we denote the corresponding register arrays by
$K_l$, $V_l$, $I_l$, and $C_l$, where each element of $K_l$ is initialized
as $-1$ and each element of other arrays is initialized as zero. 
We define three metadata fields for level~$l$, namely Meta.key$_l$,
Meta.value$_l$, and Meta.flag$_l$.  The metadata fields Meta.key$_l$ and
Meta.value$_l$ store the key and its value, respectively, that are pushed to
level~$l$, while Meta.flag$_l$ stores the intermediate result at level~$l$.
All metadata fields are initialized as zeros prior to the processing of each
packet. 

The {\em Push} function pushes the key Meta.key$_l$ and the 
value Meta.val$_l$ to level~$l$ in three stages:
\begin{itemize}[leftmargin=*]
\item \textbf{Stage~1:} update $(K_l, I_l)$ with a pairs atom; 
\item \textbf{Stage~2:} set the metadata value using the results in
Stage~1; and
\item \textbf{Stage~3:} update $C_l$ and prepare the key and value that should 
be pushed to the next level based on the metadata value from Stage~2.
\end{itemize}
We denote the three branches in the Push function (i.e., Lines~6-9,
Lines~10-11, and Lines~12-14 in Figure~\ref{fig:ops1d}) by Case~1, Case~2 and
Case~3, respectively.  In Stage~1, we issue a pairs atom to perform
conditional branching on $K_{l, h_l(x)}$ and $I_{l, h_l(x)}$ and update their
values accordingly. Each time when Case~1 or Case~3 happens, we output the
original value of $K_{l, h_l(x)}$ to Meta.key$_{l+1}$ (Lines~11-12).  Note
that all conditional branches in a pairs atom are executed simultaneously.  In
Stage~2, if Meta.key$_{(l+1)}$ equals zero (i.e., neither Case~1 nor Case~3
happens), we set Meta.key$_{(l+1)}$ as the generalization of $x$ at
level~$l+1$ and Meta.val$_{(l+1)}$ as $v_x$ (Line~14-16).  Otherwise, we set
Meta.flag$_l$ as 0 if $x$ equals Meta.key$_{(l+1)}$ (Case~1), or 1 if they are
different (Case~2) (Line~18). In Stage~3, we update $C_{l, h(x)}$ based on the
value of Meta.flag$_l$ (Lines~20-25).   

To realize the Update procedure of \sysname in P4,  we call 
{\sf Push($0$)} to update each packet from level~$0$ (Line~27). If
Meta.val$_1$ has a non-zero value in {\sf Push($0$)} (i.e., either Case~2
or Case~3 happens), we continue to call {\sf Push($1$)} to update level~$1$
(Lines~28-29).  We have a similar process for level~2 and level~3
(Lines~30-33).  For level~4, we maintain only one register to count the value
of Meta.val$_4$, as there is only one fully generalized key (i.e., any
address) (Lines~34-35). 

%------------------------------------------------------------------------
% Section V: Theoretical Analysis
%------------------------------------------------------------------------
\section{Theoretical Analysis} 
\label{sec:theory}

We present theoretical analysis on \sysname for both 1D and 2D HHH
detection.  Our analysis configures \sysname with $\tfrac{2}{\epsilon}$
buckets per array on average and the number of ancestors being checked in
estimation $t=\log\tfrac{1}{\delta}$ ($t$ is defined in \S\ref{subsec:1d}),
where $\epsilon$ ($0< \epsilon \le \phi <1$) is the approximation parameter,
$\delta$ ($0 < \delta < 1$) is the error probability, and the logarithm base
is 2.  Each key is represented in $\log n$ bits, where $n$ is the maximum
value of a key.  We use the same $n$ for both 1D and 2D cases. 

Our analysis assumes $H\ge t = \log\tfrac{1}{\delta}$ (where $H$ is the number
of nodes in the hierarchy, or the number of arrays in \sysname, as defined in
\S\ref{subsec:data}). That is, \sysname has sufficient memory to cover all
nodes in the hierarchy for accurate HHH detection. 

\subsection{Main Results}

Our goal is to show that \sysname maintains the following two properties.
\begin{itemize}[leftmargin=*]
\item 
\textbf{Accuracy:} $\Pr[\hat{S}(x) - S(x) \le k\epsilon\mathcal{S}] >
1-\tfrac{1}{4k}$, for some constant $k \ge 1$.  This property states that the
estimated count of a key in \sysname is close to its true count with a high
probability. 
\item 
\textbf{Coverage:} For each key $x\notin \mathcal{H}$, $S_{\mathcal{H}}(x) <
\phi\mathcal{S}$.  This property states that any key not in the output
set of HHHs $\mathcal{H}$ must have a conditioned count with respect to
$\mathcal{H}$ less than $\phi\mathcal{S}$.  
\end{itemize}

We first bound the count of a key in a bucket to which the key is hashed.  Let
$\Delta(x)$ be the true count of $x$ tracked in its hashed bucket $B(i,j)$.
Lemma~\ref{lem:bound} gives both the upper and lower bounds of $\Delta(x)$.
Lemma~\ref{lem:est} further shows that $U(x)$ given by the Estimate function
is an upper bound of $\Delta(x)$. 

\begin{lemma} 
\label{lem:bound}
Consider the bucket $B(i,j)$ to which key $x$ is hashed.  If $K_{i,j} $ equals
$x$, then $C_{i,j}\le \Delta(x)\le \tfrac{V_{i,j}+I_{i,j}}{2}$; otherwise,
$0\le \Delta(x) \le \tfrac{V_{i,j}-I_{i,j}}{2}$.
\end{lemma}

\begin{proof} 
We can bound $\Delta(x)$ with the values
of $K_{i,j}$ and $I_{i,j}$ based on the prior analysis
\cite[Lemma~2]{Tang2019} on MJRTY \cite{Boyer1991}.  If $x$ equals $K_{i,j}$, $I_{i,j}\le
\Delta(x) \le \tfrac{V_{i,j}+I_{i,j}}{2}$; otherwise, $0\le \Delta(x) \le
\tfrac{V_{i,j}-I_{i,j}}{2}$.  In \sysname, if $x$ equals $K_{i,j}$, we use
$C_{i,j}$ to track the cumulative count of $x$ since it is stored in
$K_{i,j}$, which implies $\Delta(x)\ge C_{i,j} \ge I_{i,j}$. 
\end{proof}

\begin{lemma} 
\label{lem:est}
The returned estimate $U(x)$ of key $x$ by the Estimate function is an upper
bound of $\Delta(x)$. 
\end{lemma}
\begin{proof} We focus on 1D HHH detection, while
the proof for 2D HHH detection is identical.  Denote the $t$ closest ancestors
of $x$ by $y_i$, where $1\le i\le t$.  Let $U_x$ and $U_i$ be the temporary
estimates calculated for $x$ and $y_i$ in the Estimate function, respectively.
By Lemma~\ref{lem:bound} and the Estimate function, we have $U_x \ge
\Delta(x)$ and $U_i\ge \Delta(y_i)+v\ge \Delta(x)$, where $v$ is the sum of
cumulative counts of $y_i$'s descendants.  Thus, $U(x) =\min_{1\le i\le
t}{\{U_x, U_i\}} \ge \Delta(x)$.
\end{proof}

We first consider 1D HHH detection. Theorem~\ref{the:1dcom} states the space
and time complexities of \sysname. Theorem~\ref{the:1d} shows that \sysname
satisfies both accuracy and coverage properties.  Theorem~\ref{the:1dfp}
further presents the bounds of \sysname for 1D HHH detection under certain
conditions. 

\begin{theorem}
\label{the:1dcom}
In 1D HHH detection,
\sysname finds HHHs in $O(\tfrac{H}{\epsilon}\log{n})$ space. The update time
is $O(H)$. The detection time is
$O(\tfrac{H(d-1)}{\epsilon}\log{\tfrac{1}{\delta}})$.  Note that the space and
time complexities of \sysname are implicitly related to the error
probability $\delta$, as we assume $H\ge \log{\tfrac{1}{\delta}}$.
\end{theorem}

\begin{proof}  
We maintain an array of buckets for each of the $H$ nodes in the hierarchy.
Each bucket stores a $\log n$-bit candidate HHH and three counters.  Thus, the
space usage is $O(\tfrac{H}{\epsilon}\log{n})$.  
Each per-packet update accesses at most $H$ buckets, and hence takes $O(H)$
time in the worst case.  
We traverse all $Hw$ buckets to return the set of HHHs. For each candidate HHH
$x$ in a bucket, we obtain $\hat{S}_{\mathcal{H}}(x)$ by checking the $t$
closest ancestors of $x$. If $\hat{S}_{\mathcal{H}}(x)$ is below the
threshold,  we push $x$ to at most $d-1$ higher-level arrays.  Thus, it takes
$O(\tfrac{H(d-1)}{\epsilon}\log{\tfrac{1}{\delta}})$ time to return all HHHs.
\end{proof}        

\begin{theorem}
\label{the:1d}
The main operations of \sysname in 1D HHH detection (Figure~\ref{fig:ops1d} in
\S\ref{subsec:1d}) satisfy the accuracy and coverage properties.
\end{theorem}
\begin{proof} We first prove the
accuracy property. Let $B(i,j)$ be the bucket to which $x$ is hashed.
Consider the sum of all keys in $B(i,j)$ except $x$. Its expectation is
$\mathbf{E}[V_{i,j}-\Delta(x)] = \mathbf{E}[\sum_{y\neq x,
h_{i}(y)=h_{i}(x)}\Delta(y)] \le \tfrac{\mathcal{S}-\Delta(x)}{w_{i}}\le
\tfrac{\epsilon\mathcal{S}}{2}$. By Markov's inequality, 
\begin{equation} \label{eq:markov} \Pr[V_{i,j}-\Delta(x) \ge
    k\epsilon\mathcal{S}] \le \tfrac{1}{2k}.
\end{equation} 
 
We study the difference between $S(x)$ and $\hat{S}(x)$.  When we start
checking $B(i,j)$ in the Detect operation, the counts of the descendants of $x$
that are not in $\mathcal{H}$ are all pushed to $B(i,j)$. At this time, $S(x)$
consists of two parts in \sysname: $\Delta(x)$; and the sum of the cumulative
counts of $x$'s descendants in $\mathcal{H}$.  That is,
$S(x)=\Delta(x)+\sum_{x'\in\mathcal{H}\wedge x'\prec x}C_{i',h_{i'}(x')}$, where
$i'$ is the level of $x'$.  By the Detect operation, $\hat{S}(x) = U(x) +
\sum_{x'\prec x \wedge x\in\mathcal{H}}C_{i',h_{i'}(x')}$.  We then have
$\hat{S}(x)- S(x) = U(x)-\Delta(x)$.  By Lemma~\ref{lem:bound} and the Estimate
function, if $K_{i,j}$ equals $x$, $U(x)-\Delta(x)\le
\tfrac{V_{i,j}+I_{i,j}}{2}-\Delta(x)\le\tfrac{V_{i,j}-\Delta(x)}{2}$; otherwise,
if $K_{i,j}\neq x$,
$U(x)-\Delta(x)\le\tfrac{V_{i,j}-I_{i,j}}{2}-\Delta(x)\le\tfrac{V_{i,j}-\Delta(x)}{2}$.
Combining both cases, we have $\Pr[U(x)-\Delta(x)\ge k\epsilon\mathcal{S}]\le
\Pr[\tfrac{V_{i,j}-\Delta(x)}{2}\ge k\epsilon\mathcal{S}] \le \tfrac{1}{4k}$ by
Equation~\ref{eq:markov}.  Thus, $\Pr[\hat{S}(x)-S(x)\le
k\epsilon\mathcal{S}]=\Pr[U(x)-\Delta(x) \le k\epsilon\mathcal{S}]  \ge
1-\tfrac{1}{4k}$.

We prove the coverage property by contradiction.  Suppose that
$S_{\mathcal{H}}(x) \ge \phi\mathcal{S}$. As the counts of the descendants of
$x$ that are not in $\mathcal{H}$ must be pushed to $B(i,j)$, we have
${S}_{\mathcal{H}}(x) \le \Delta(x)$. Then, $\phi\mathcal{S}\le
{S}_{\mathcal{H}}(x) \le \Delta(x) \le U(x) = \hat{S}_{\mathcal{H}}(x)$. We do
not report $x$ as an HHH only if $x$ is not stored in $K_{i,j}$. In this case,
the count of $x$ in $B(i,j)$ is further pushed to its ancestors until $x$ is
admitted by an HHH.  Thus, we must add at least one of the ancestors of $x$ to
$\mathcal{H}$.  By the definition of the conditioned count,
$S_{\mathcal{H}}(x)=0$, which is a contradiction. 
\end{proof}

\begin{theorem}
\label{the:1dfp}
In 1D HHH detection, if key $x$ and each of its $t=\log\tfrac{1}{\delta}$
closest ancestors have counts at most 
$(\phi-\tfrac{\epsilon}{2})\mathcal{S}$, \sysname falsely reports $x$ as an
HHH with a probability at most $\delta$; if $x$ is at level~$l$ with
$S_{\mathcal{H}_{l-1}}(x)\ge \phi\mathcal{S}$, \sysname misses $x$ and reports
its ancestor at a level higher than $l+t$ as an HHH with a probability at most
$\delta$. 
\end{theorem}
\begin{proof}
We first show that \sysname reports a small key $x$ that is at
level $l$ with a small probability.
Suppose that $x$ is hashed to bucket $B(l,j)$.  A necessary condition of
reporting $x$ as an HHH is that $\hat{S}_{\mathcal{H}_{l-1}}(x)=U(x) \ge
\phi\mathcal{S}$.  We get $U(x) = \min\{U_k\}$, where $0\le k \le t$ and $U_0,
U_1, \dots, U_t$ are the estimate of $x$ and its $t$ ancestors $ y_1, \dots,
y_t$, respectively, in the Estimate function.  We have
$U_{k}\ge\phi\mathcal{S}$ for each $0\le k \le t$.  Consider $U_0$ first. We
have $U_0-S(x)\ge
\phi\mathcal{S}-(\phi-\tfrac{\epsilon}{2})\mathcal{S}=\tfrac{\epsilon\mathcal{S}}{2}
$.  Then, $ \Pr[U_0-S(x) \ge \tfrac{\epsilon\mathcal{S}}{2}] \le
\Pr[U_0-\Delta(x) \ge\tfrac{\epsilon\mathcal{S}}{2}]\le
\Pr[\tfrac{V_{i,j}-\Delta(x)}{2}\ge \tfrac{\epsilon\mathcal{S}}{2}]\le
\tfrac{1}{2}$ by Equation~\ref{eq:markov} and the proof in Theorem~\ref{the:1d}.
Similarly, we can get $\Pr[U_{k}-S(y_{k})\ge\tfrac{\epsilon\mathcal{S}}{2}]\le
\tfrac{1}{2}$.    Thus, $\Pr[U(x) \ge \phi\mathcal{S}] =
\prod_{k=0}^{t}\Pr[U_{k}-S(y_{k})\ge \tfrac{\epsilon\mathcal{S}}{2}] \le
\tfrac{1}{2^{t+1}} \le \delta$. 

We show that \sysname misses an HHH $x$ at level $l$ but reports
its ancestor at much higher levels with a small probability.  Given
$S_{\mathcal{H}_{l-1}}(x)\ge \phi\mathcal{S}$, we have
$\hat{S}_{\mathcal{H}_{l-1}}(x)=U(x)\ge\Delta(x)\ge S_{\mathcal{H}_{l-1}}(x)\ge
\phi\mathcal{S}$ by the Detect operation and Lemma~\ref{lem:est}.  We do not
report $x$ as an HHH when checking its hashed bucket $B(l,j)$ if $x$ is not
stored in $K_{l,j}$. By MJRTY, the count of $x$ in $B(l,j)$ does not account for
more than half of the total count in that bucket.  We have
$\Pr[\Delta(x)\le\tfrac{V_{l,j}}{2}] = \Pr[V_{l,j} - \Delta(x) \ge \Delta(x)]
\le \Pr[V_{l,j}-\Delta(x) \ge S_{\mathcal{H}}(x)] \le \Pr[V_{l,j}-\Delta(x) \ge
\phi\mathcal{S}] \le \Pr[V_{l,j}-\Delta(x) \ge \epsilon\mathcal{S}] \le
\tfrac{1}{2}$ by Equation~\ref{eq:markov}.  Similarly, the probability that we
miss the next ancestor of $x$ is also smaller than $\tfrac{1}{2}$.  Thus, the
probability that we miss all $t-1$ closest ancestors of $x$ is smaller than
$\tfrac{1}{2^{t}}=\delta$.
\end{proof}

We also consider 2D HHH detection.  Theorem~\ref{the:2dcom} states the space
and time complexities.  Theorem~\ref{the:2d} shows that \sysname satisfies
both the accuracy and coverage properties.  In the interest of space, we
present the major operations of 2D HHH detection and the proofs of the
theorems in the supplementary file. 

\begin{theorem}
\label{the:2dcom}
In 2D HHH detection, \sysname finds HHHs in $O(\tfrac{H}{\epsilon}\log{n})$
space. The update time is $O(d)$ in the worst case. The detection time is
$O(\tfrac{(d-1)H}{\epsilon}\log{\tfrac{1}{\delta}})$.
\end{theorem}

\begin{theorem}
\label{the:2d}
The main operations of \sysname in 2D HHH detection satisfy the accuracy and
coverage properties.
\end{theorem}

\subsection{Comparisons with Existing Schemes}
\label{subsec:comparisons}

\noindent
{\bf Update time.} Compared to previous studies, except for RHHH
\cite{BenBasat2017} with $O(1)$ update time, \sysname achieves the same or even
a lower update time complexity even in the worst case.  For example, the schemes
in \cite{Mitzenmacher2012} process each packet at each of the $H$ nodes in the
hierarchy, and the update time complexity is $O(H\log\tfrac{1}{\epsilon})$ (for
heap-based implementation) and $O(H)$ (using unitary updates).  The schemes in
\cite{Cormode2008} have an amortized update time complexity
$O(H\log{(\epsilon\mathcal{S})})$. 

Note that our update time analysis for \sysname in Theorem~\ref{the:1dcom}
focuses on the worst case (i.e., each packet update takes $O(H)$ time), yet
our evaluation (i.e.  Experiments~4 and 5 in \S\ref{sec:evaluation})
shows that \sysname can terminate the Update operation for most packets at low
levels under the skewness of IP traffic. The amortized update time of \sysname
in our evaluation for real-world traces is much smaller than $O(H)$.

\paragraph{Accuracy.} 
Prior studies pose strong accuracy guarantees for the estimated count
$\hat{S}(x)$ of key $x$.  For example, the HHH detection schemes in 
\cite{Cormode2008, Mitzenmacher2012} guarantee that the error between
$\hat{S}(x)$ and $S(x)$ is at most $\epsilon\mathcal{S}$, while randomized
HHH (RHHH) \cite{BenBasat2017} keeps the error within $\epsilon\mathcal{S}$
with a probability of at least $1-\delta$.  In contrast, \sysname relaxes the
accuracy guarantee for high update performance, in which $\hat{S}(x)$ is close
to $S(x)$ and deviates much from $S(x)$ with a small probability.  

We can improve the accuracy of \sysname by maintaining multiple rows of
buckets for each node of the hierarchy, at the expense of degrading the
processing speed and increasing the update complexity.   The reasons are 
two-fold.  First, the process of each key in an array of \sysname requires
multiple memory accesses, as \sysname needs to insert each processed key into
multiple arrays to reduce the errors caused by hash collisions. Second, there
could be multiple keys being kicked out from an array after a key is updated, and
if it happens, we need to further check whether each of the kicked-out keys
remains in the other rows of the array except for the row where the key is
kicked out. When we push the kicked-out keys that are no longer in the array
to the subsequent higher-level arrays, each of these keys may also kick out
multiple keys.  Thus, maintaining multiple rows of buckets for each node of
the hierarchy incurs high update overhead.  Nevertheless, even though the
one-row-array design of \sysname relaxes the accuracy guarantee, our
evaluation (i.e., Experiments~1 and 2 in \S\ref{sec:evaluation}) shows
that \sysname achieves high accuracy.

%------------------------------------------------------------------------
% Section VI: Evaluation
%------------------------------------------------------------------------
\section{Evaluation} 
\label{sec:evaluation}

We compare \sysname with
six state-of-the-art HHH detection schemes, including: trie-based HHH
detection (TRIE) \cite{Zhang2004}, full ancestry (FULL) \cite{Cormode2008},
partial ancestry (PARTIAL) \cite{Cormode2008}, heap-based Space Saving (HSS)
\cite{Mitzenmacher2012}, unitary-update-based Space Saving (USS)
\cite{Mitzenmacher2012}, and randomized HHH (RHHH) \cite{BenBasat2017}; the
first five schemes are streaming-based, while RHHH is sampling-based
(\S\ref{sec:introduction}).  We show that \sysname achieves (i) high detection
accuracy, (ii) high update throughput, (iii) small convergence time, and (iv)
limited resource usage in a Tofino switch \cite{tofino}.  

\subsection{Methodology}
\label{subsec:methodology}

\noindent
{\bf Traces.} We use the real-world traces from CAIDA \cite{caida}, captured on
an OC-192 backbone link in January 2019. Note that CAIDA traces are also used
for evaluating HHH detection in both networking \cite{BenBasat2017} and
database \cite{Cormode2008,Mitzenmacher2012} communities. 
By default, we use the first five minutes of
the traces for evaluation and divide them into five one-minute epochs, each of
which has 36.7\,M packets and 1.1\,M unique IPv4 addresses on average; in
Experiment~6, we vary the number of epochs and the epoch length. We perform
HHH detection in each epoch and obtain the average results over all epochs.  

We also consider IPv6 traffic from both the CAIDA traces and the
IPv6 traces from MAWI's WIDE project\cite{wide}.  We find that \sysname shows
similar trends on both IPv4 and IPv6 traffic compared with state-of-the-arts.
For brevity, we focus on IPv4 traffic in the CAIDA traces in this section, and
report the findings for IPv6 traffic in the supplementary file.

\paragraph{Parameter settings.} We configure the number of buckets (i.e.,
$w_i$) in each array $A_i$ of \sysname for a given available memory size,
where $0\le i< H$.  We first calculate the average number of buckets, denoted
by $w_{avg}$, for each array of \sysname based on the bucket size (e.g, a
bucket consumes 16~bytes, with four bytes for each field, for 1D-byte HHH
detection) and the number of nodes $H$ in the hierarchy (e.g., five nodes for
1D-byte HHH detection).  We set $w_i$ for each array $A_i$, starting from the
top level~$H-1$ of the hierarchy.  If a level has a key space size smaller
than $w_{avg}$ (e.g., the highest level $H-1$ has only the wildcard element),
we set $w_i$ as the key space size and update $w_{avg}$ by averaging the
residual available memory size among the remaining arrays; otherwise, we set
$w_i=w_{avg}$. 

We configure the memory sizes for HSS, USS, and RHHH based on the fractional
threshold $\phi$ to ensure that there is enough memory to store the maximum
possible number of HHHs; a large $\phi$ implies a small number of true HHHs in
an epoch, which also implies less memory usage to store all HHHs.  We also
configure the maximum memory sizes for FULL and PARTIAL based on $\phi$.  Note
that the memory sizes for FULL and PARTIAL vary in an epoch, as they
dynamically expand and shrink their counter arrays during packet processing to
keep only the large keys in the arrays.

We consider 1D-byte, 1D-bit, 2D-byte, and 2D-bit HHH detection.  
We only present the results for packet counting (i.e., $v_f=1$) in the interest
of space, while \sysname shows similar results for byte counting.  For TRIE, we
only evaluate it for 1D cases, due to its high space complexity and low
accuracy for 2D cases. We implement hash functions using MurmurHash
\cite{murmurhash} in all schemes.  

\subsection{Results}
\label{subsec:results}

\begin{figure*}[!t]
\centering
\begin{tabular}{@{\ }c@{\ }c@{\ }c@{\ }c}
\multicolumn{4}{c}{\includegraphics[width=3.2in]{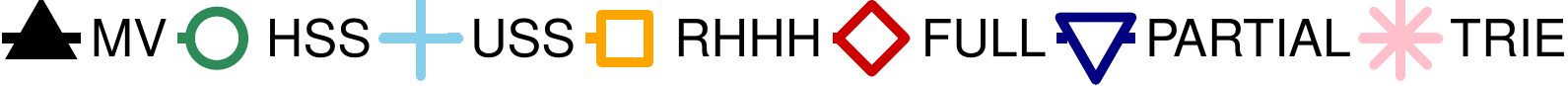}} \\
\includegraphics[width=1.7in]{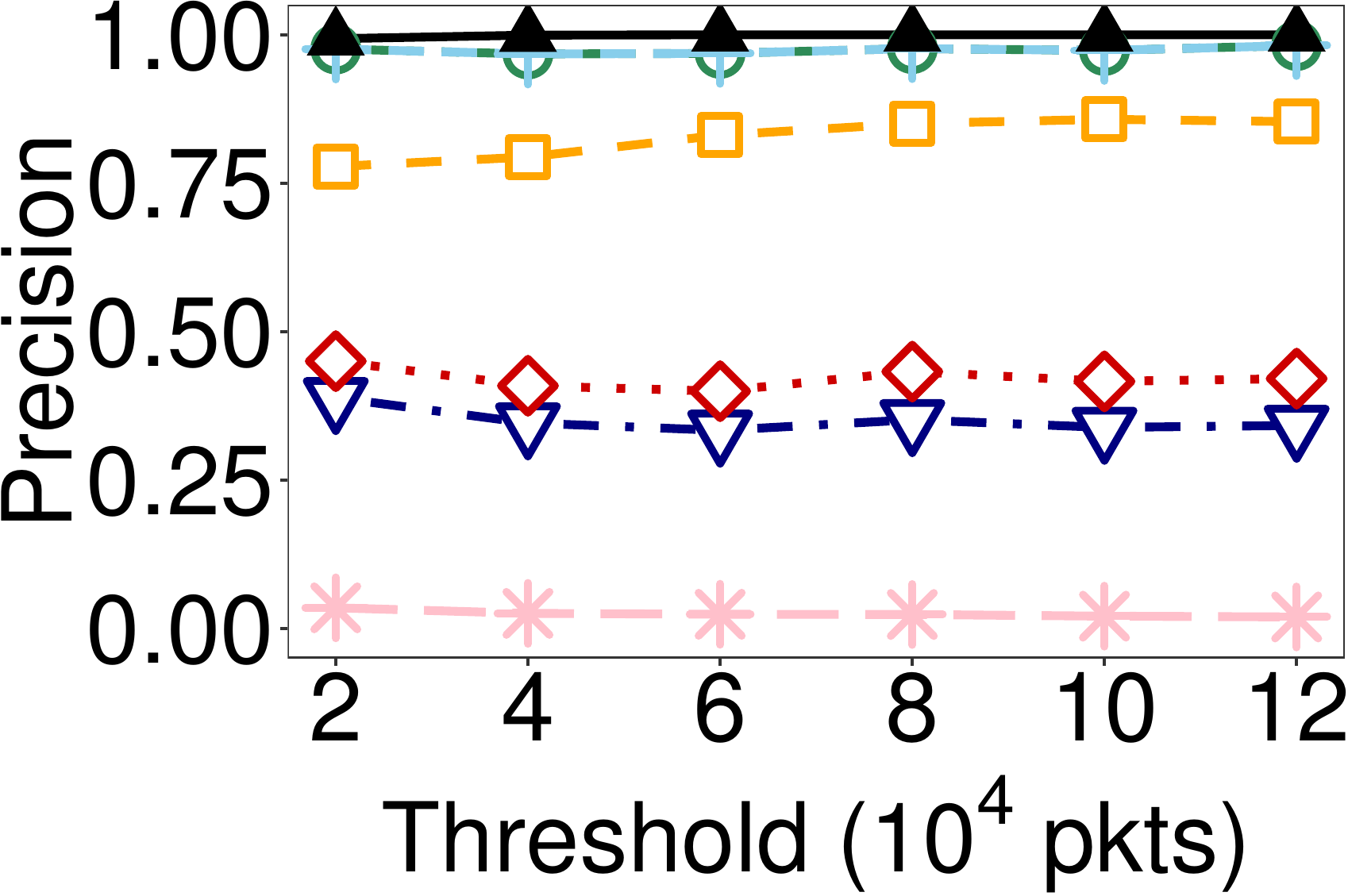} &
\includegraphics[width=1.7in]{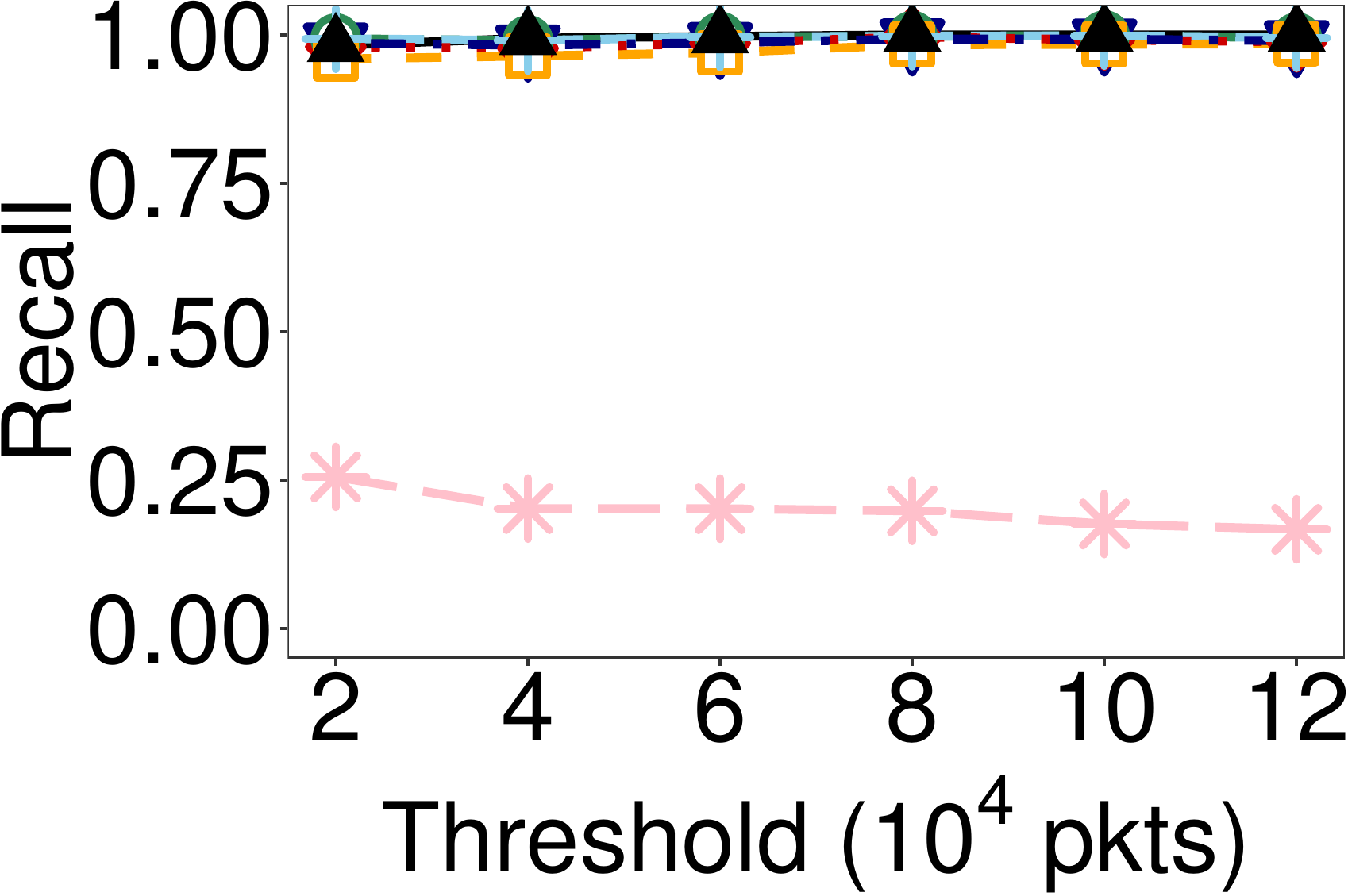} &
\includegraphics[width=1.7in]{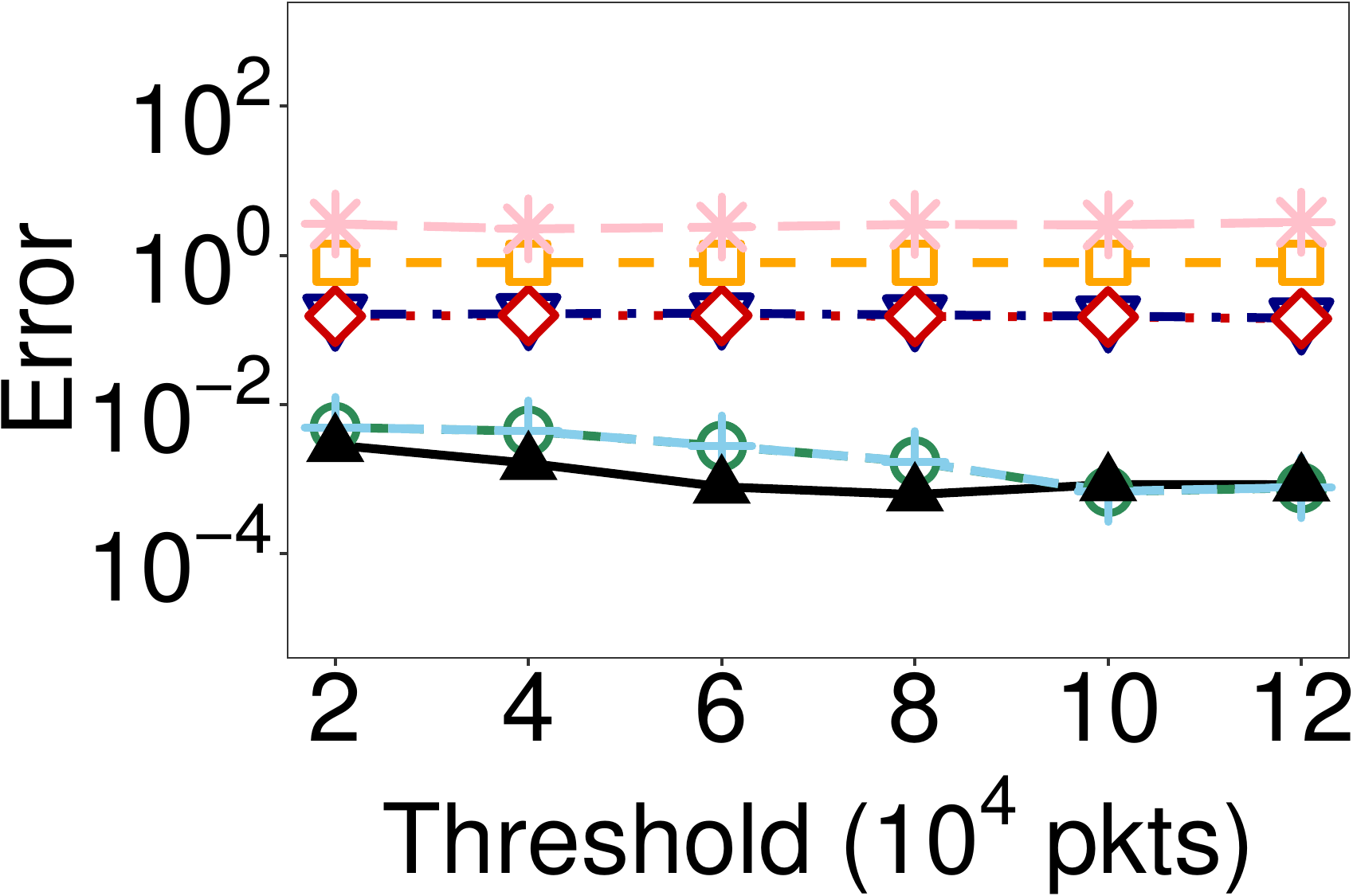} &
\includegraphics[width=1.7in]{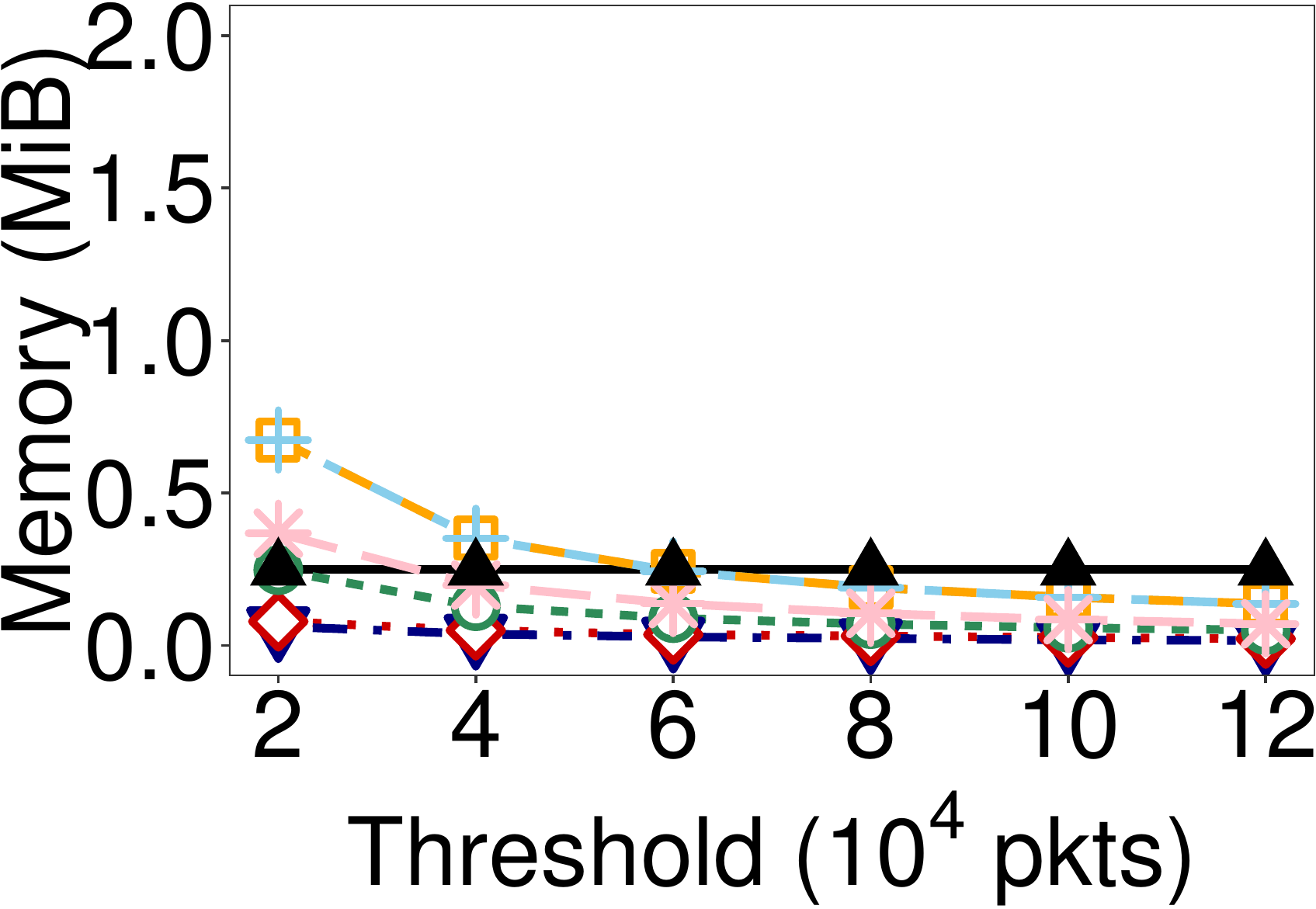}
\vspace{-3pt}\\
{\small (a) Precision for 1D-byte} & 
{\small (b) Recall for 1D-byte} &
{\small (c) Error for 1D-byte} &
{\small (d) Memory for 1D-byte}  
\vspace{3pt}\\
\includegraphics[width=1.7in]{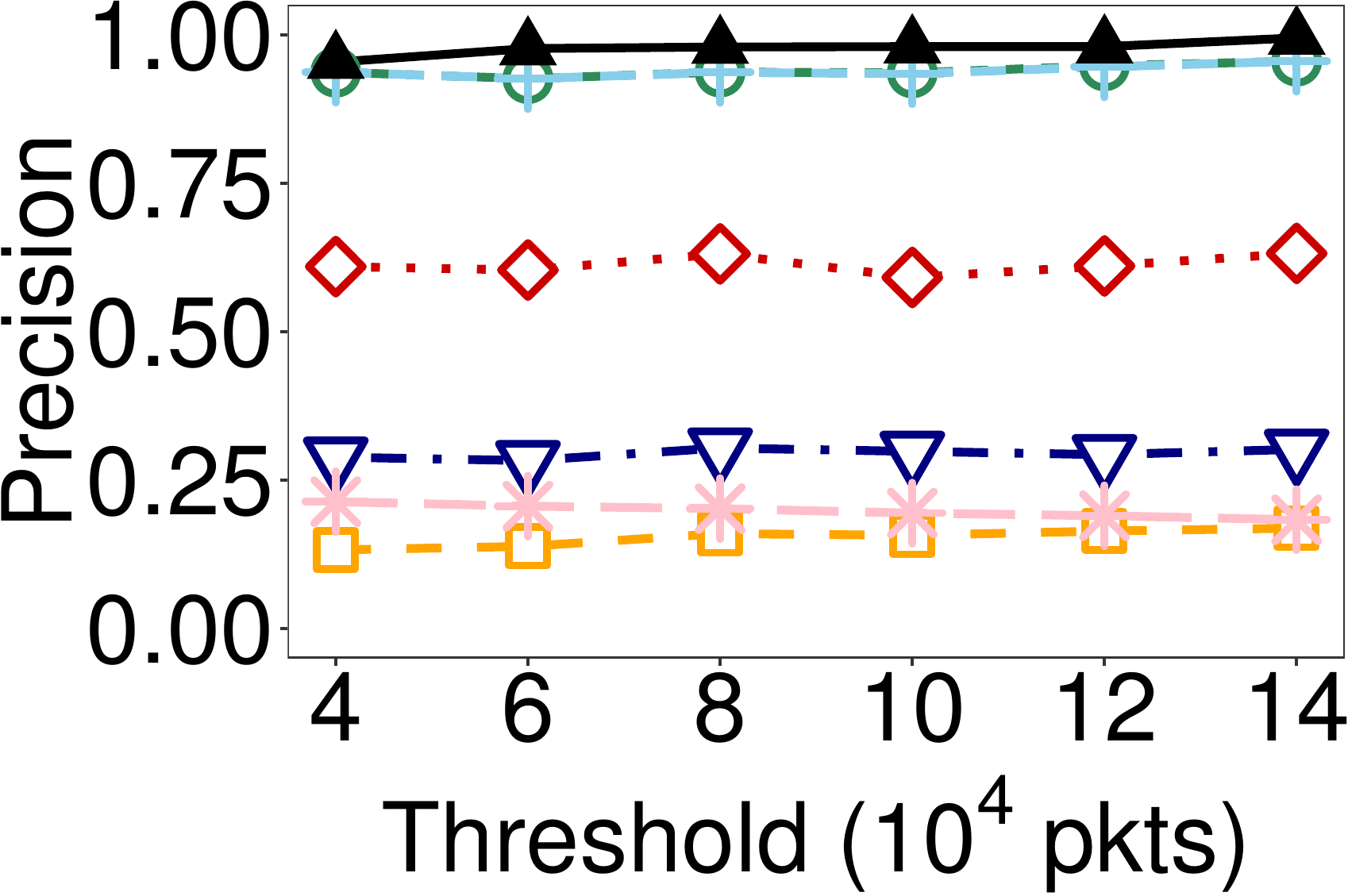} &
\includegraphics[width=1.7in]{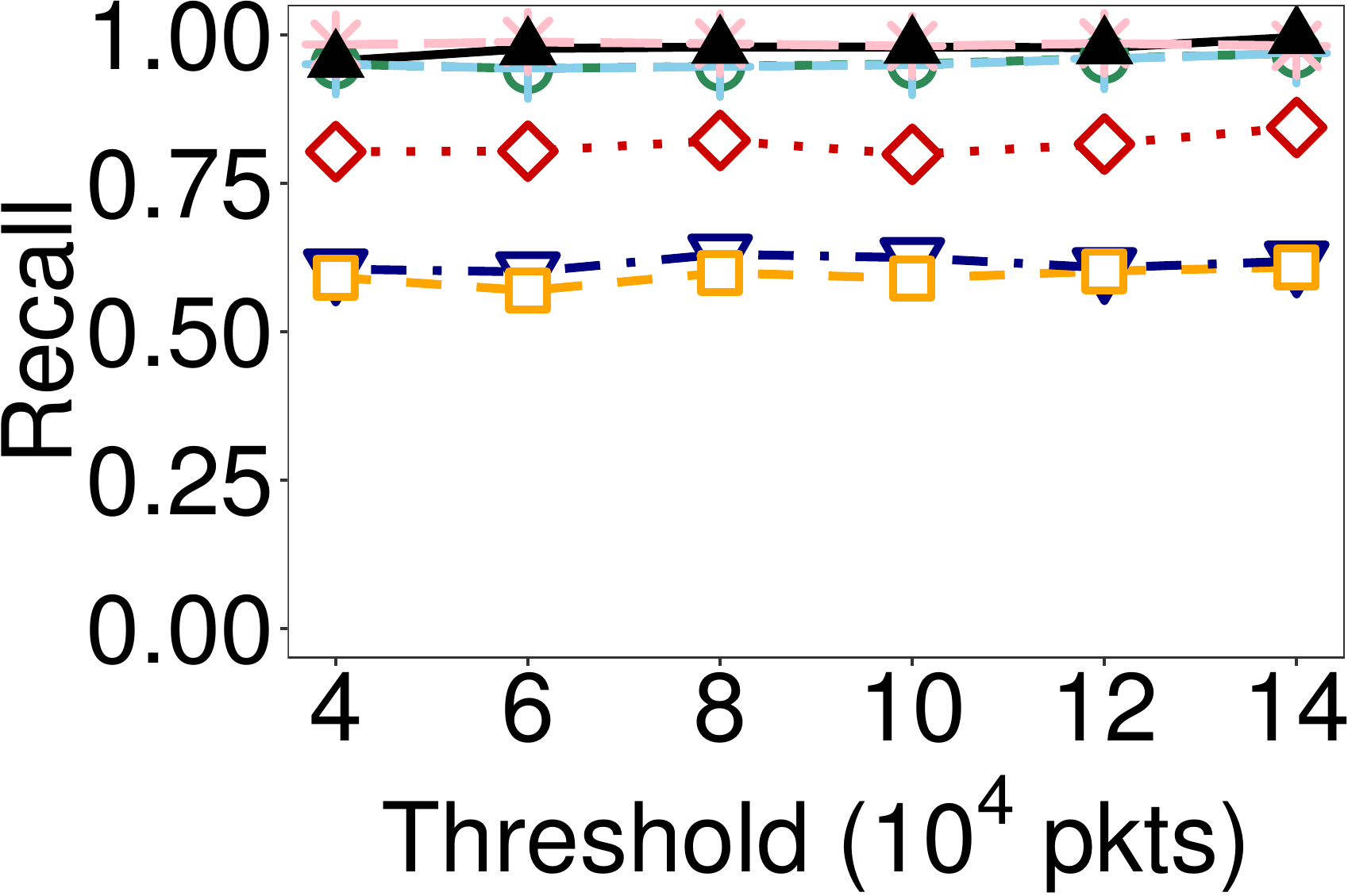} &
\includegraphics[width=1.7in]{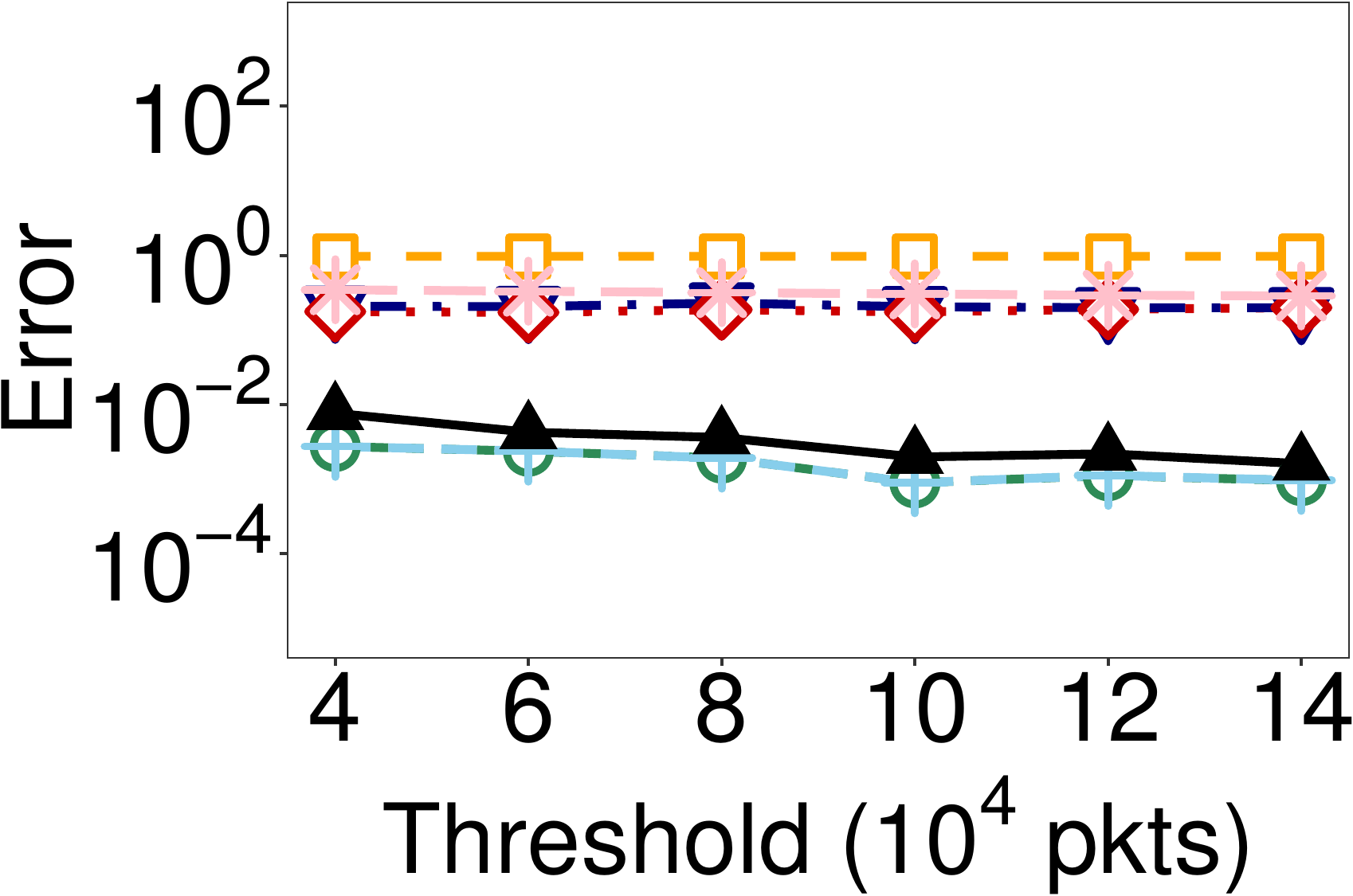} &
\includegraphics[width=1.7in]{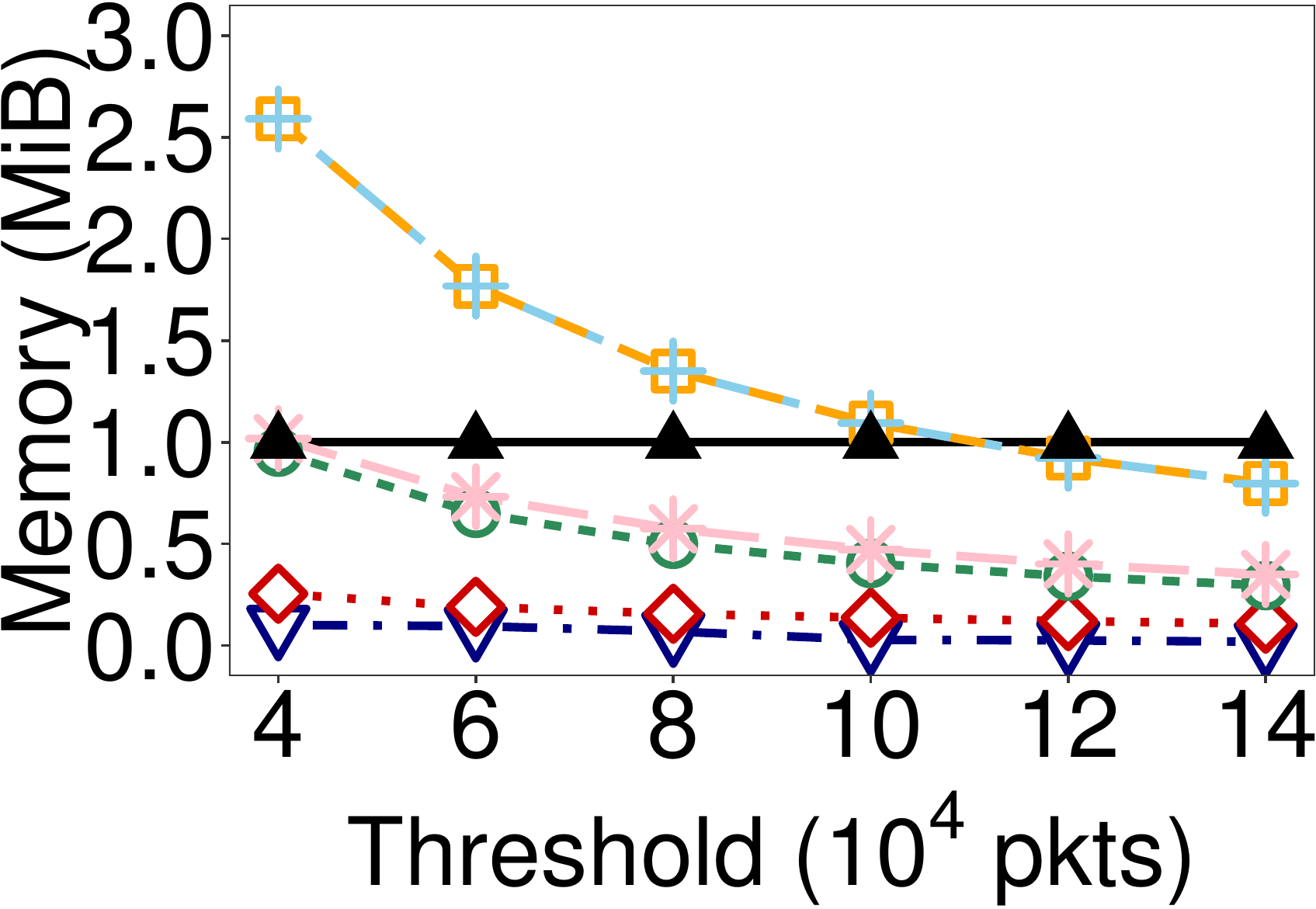}
\vspace{-3pt}\\
{\small (e) Precision for 1D-bit} & 
{\small (f) Recall for 1D-bit} &
{\small (g) Error for 1D-bit} &
{\small (h) Memory for 1D-bit}  
\vspace{3pt}\\
\includegraphics[width=1.7in]{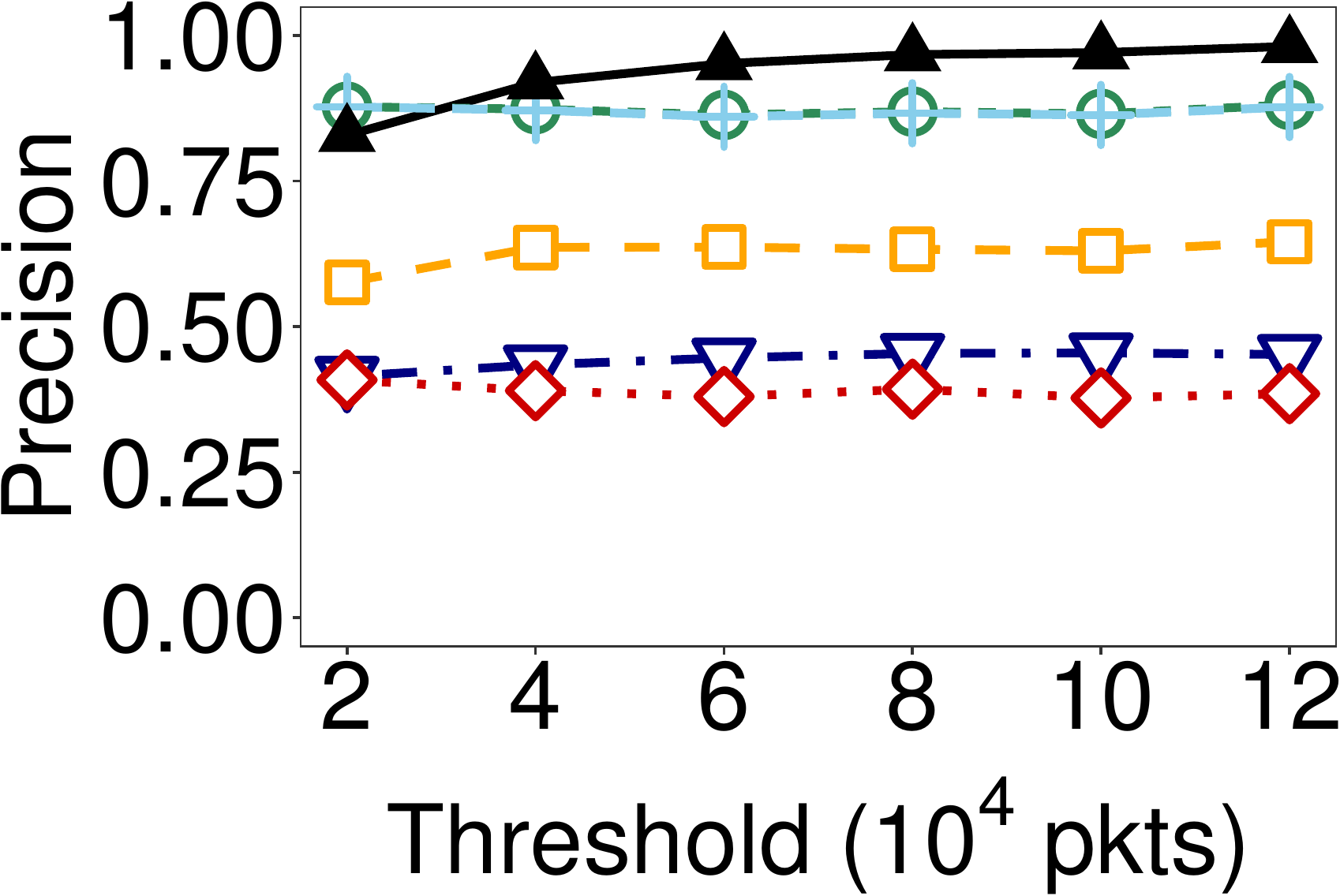} &
\includegraphics[width=1.7in]{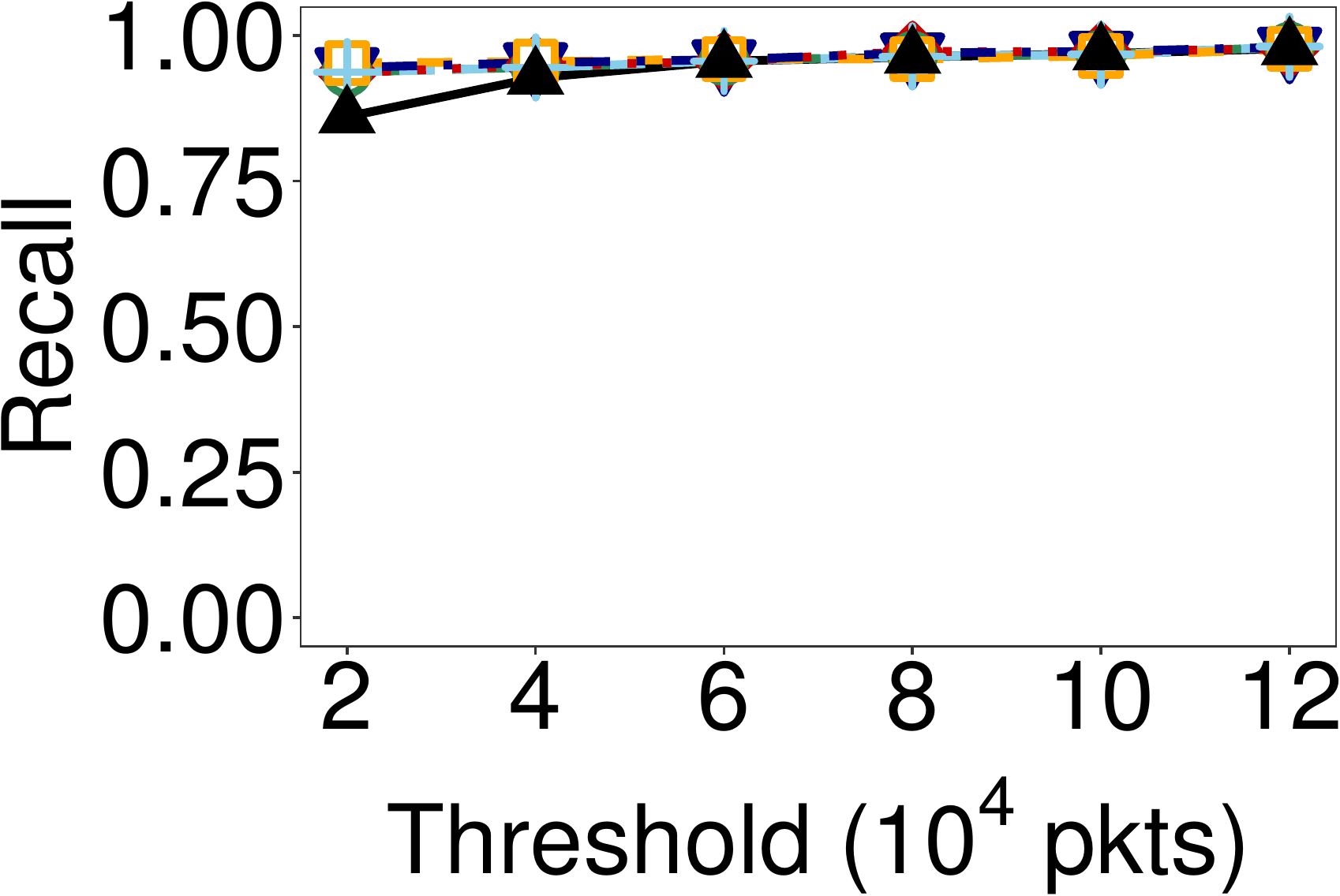} &
\includegraphics[width=1.7in]{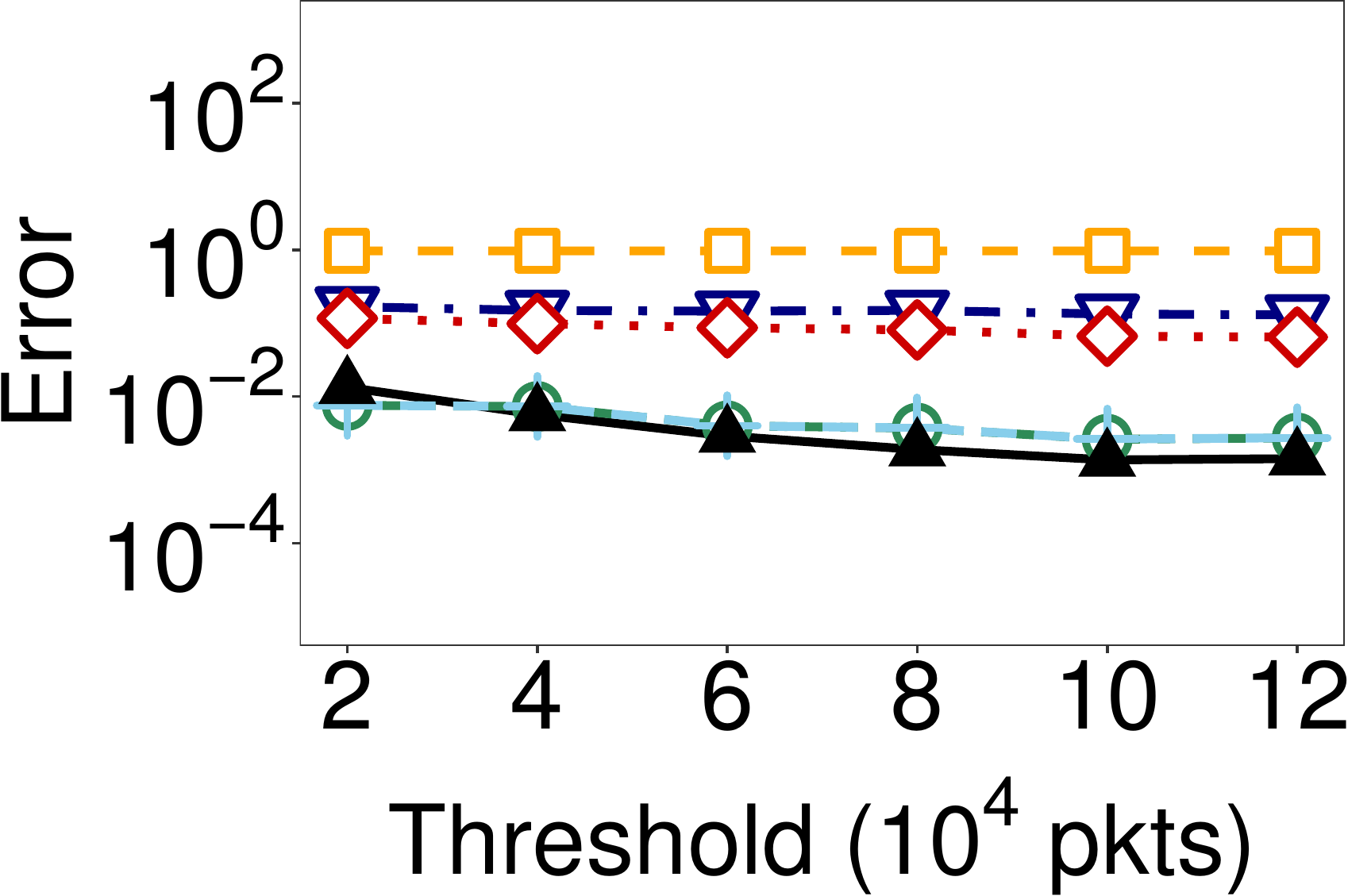} &
\includegraphics[width=1.7in]{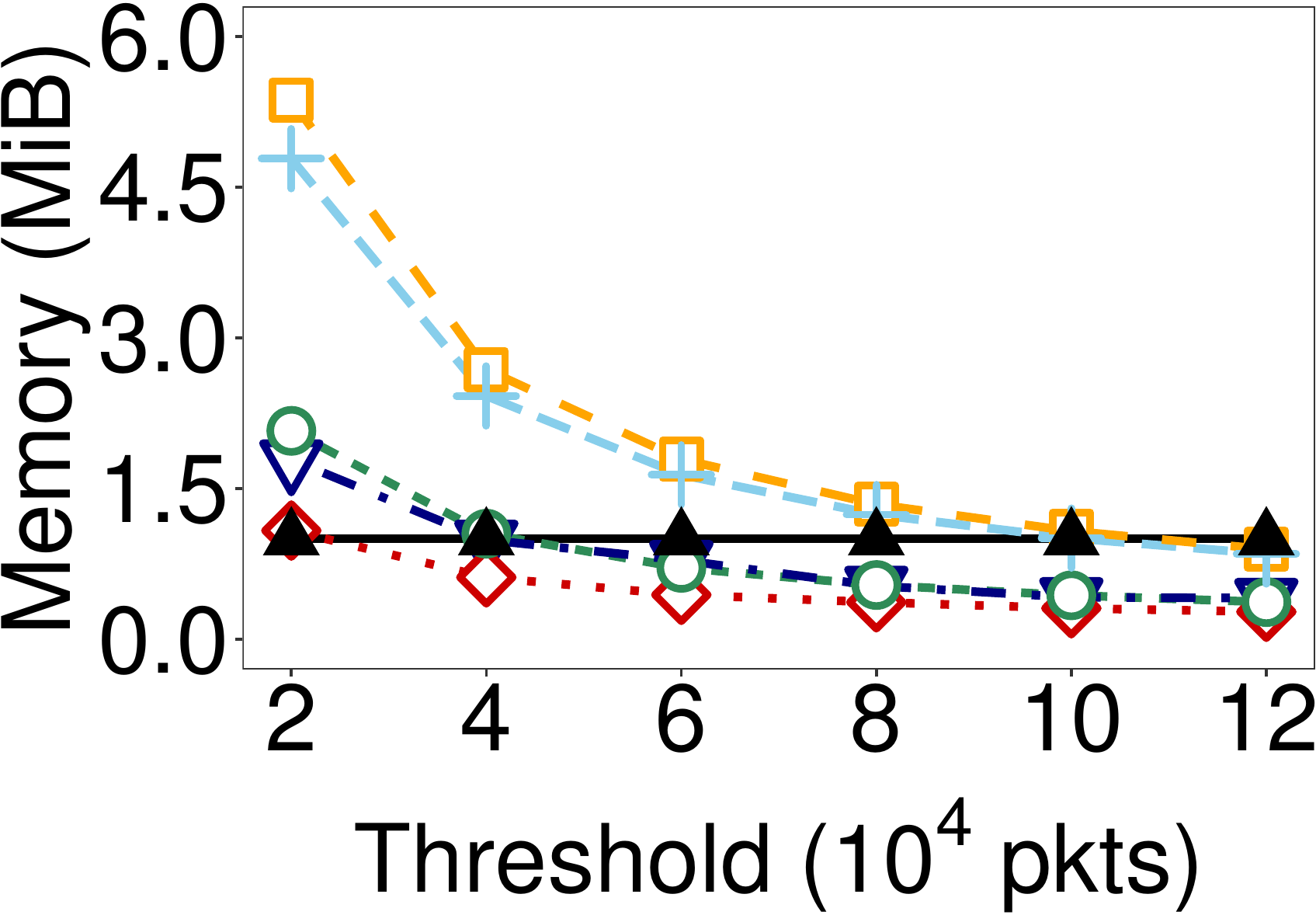}
\vspace{-3pt}\\
{\small (i) Precision for 2D-byte} & 
{\small (j) Recall for 2D-byte} &
{\small (k) Error for 2D-byte} &
{\small (l) Memory for 2D-byte}  
\vspace{3pt}\\
\includegraphics[width=1.7in]{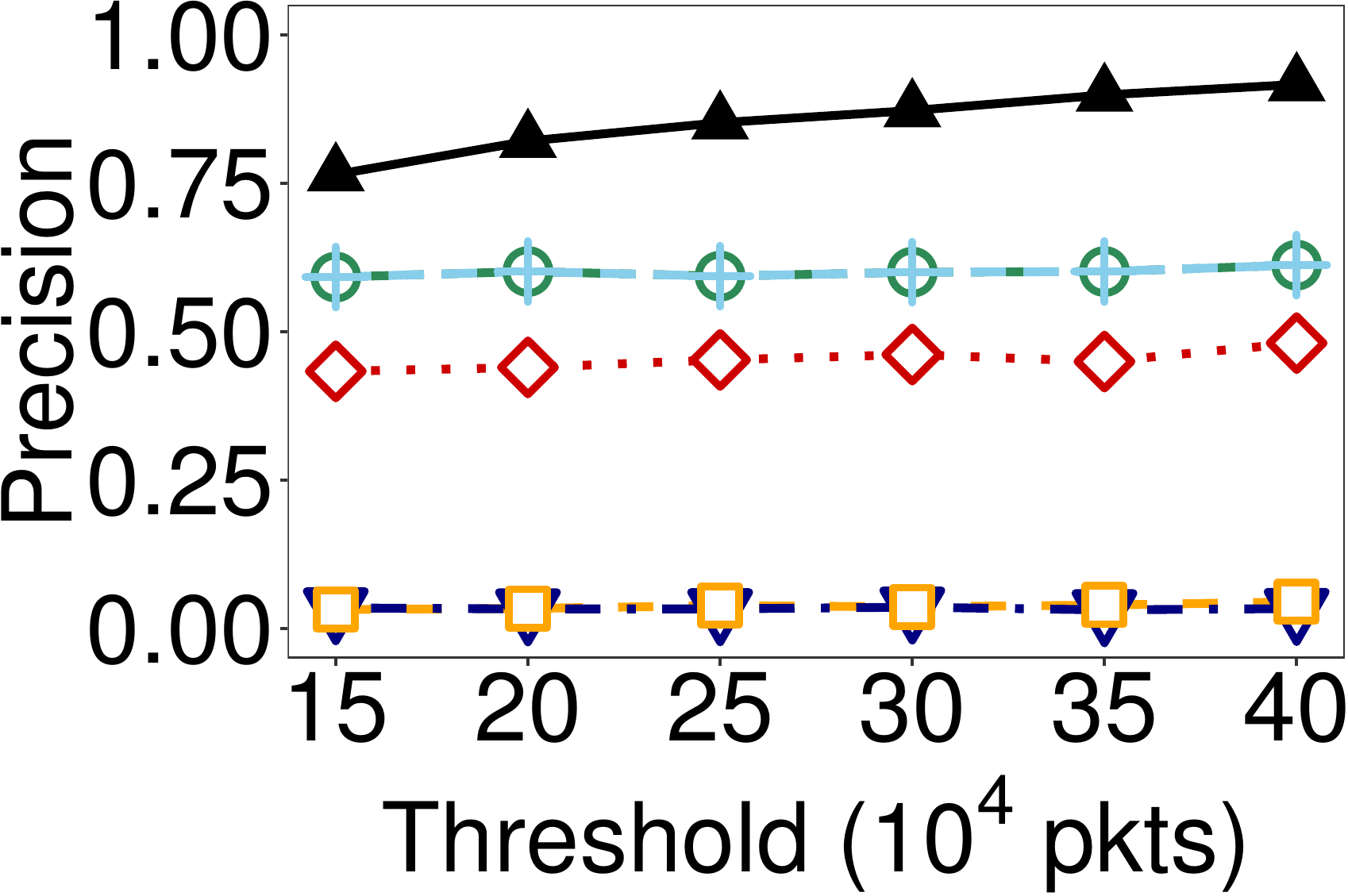} &
\includegraphics[width=1.7in]{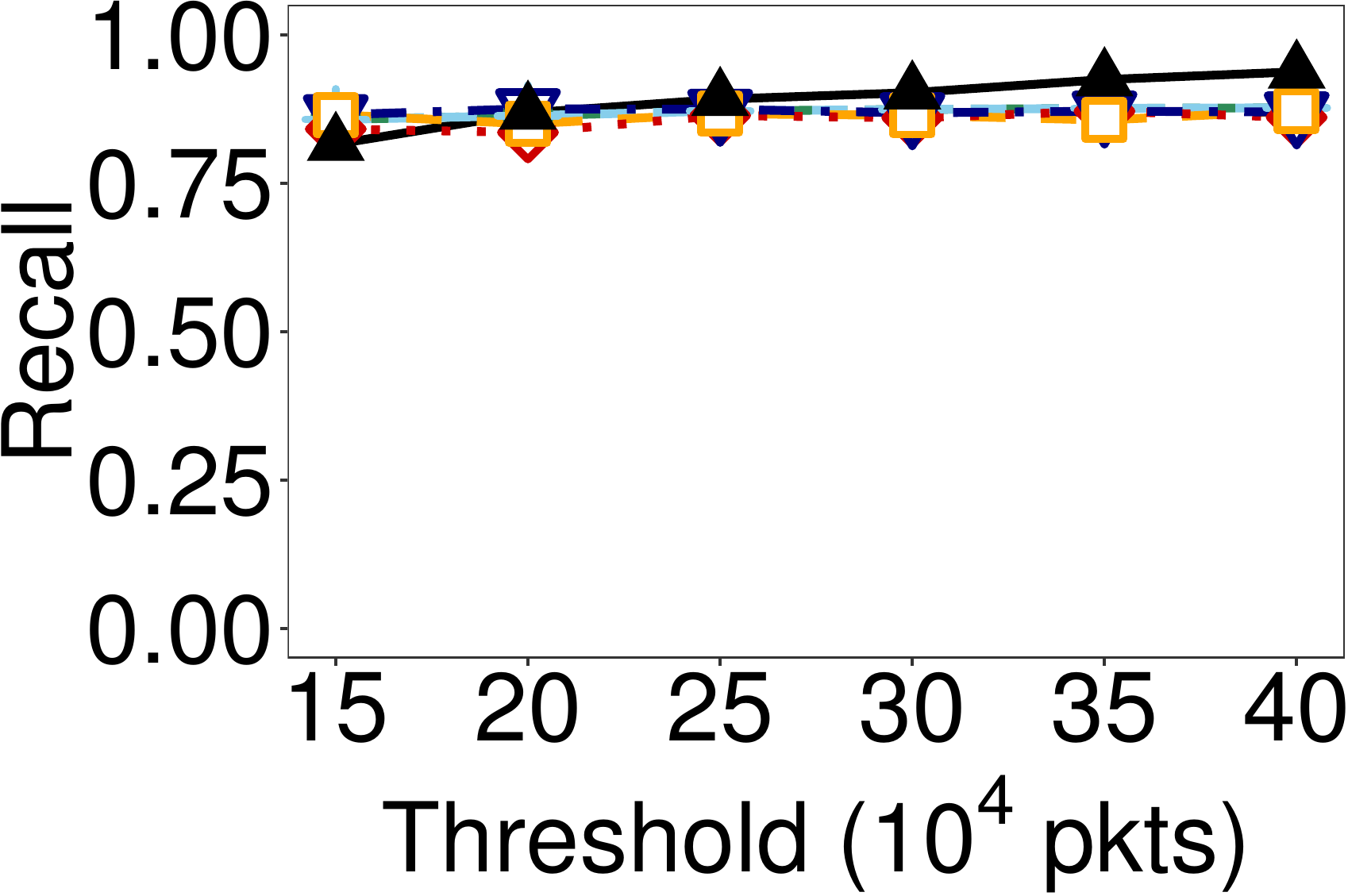} &
\includegraphics[width=1.7in]{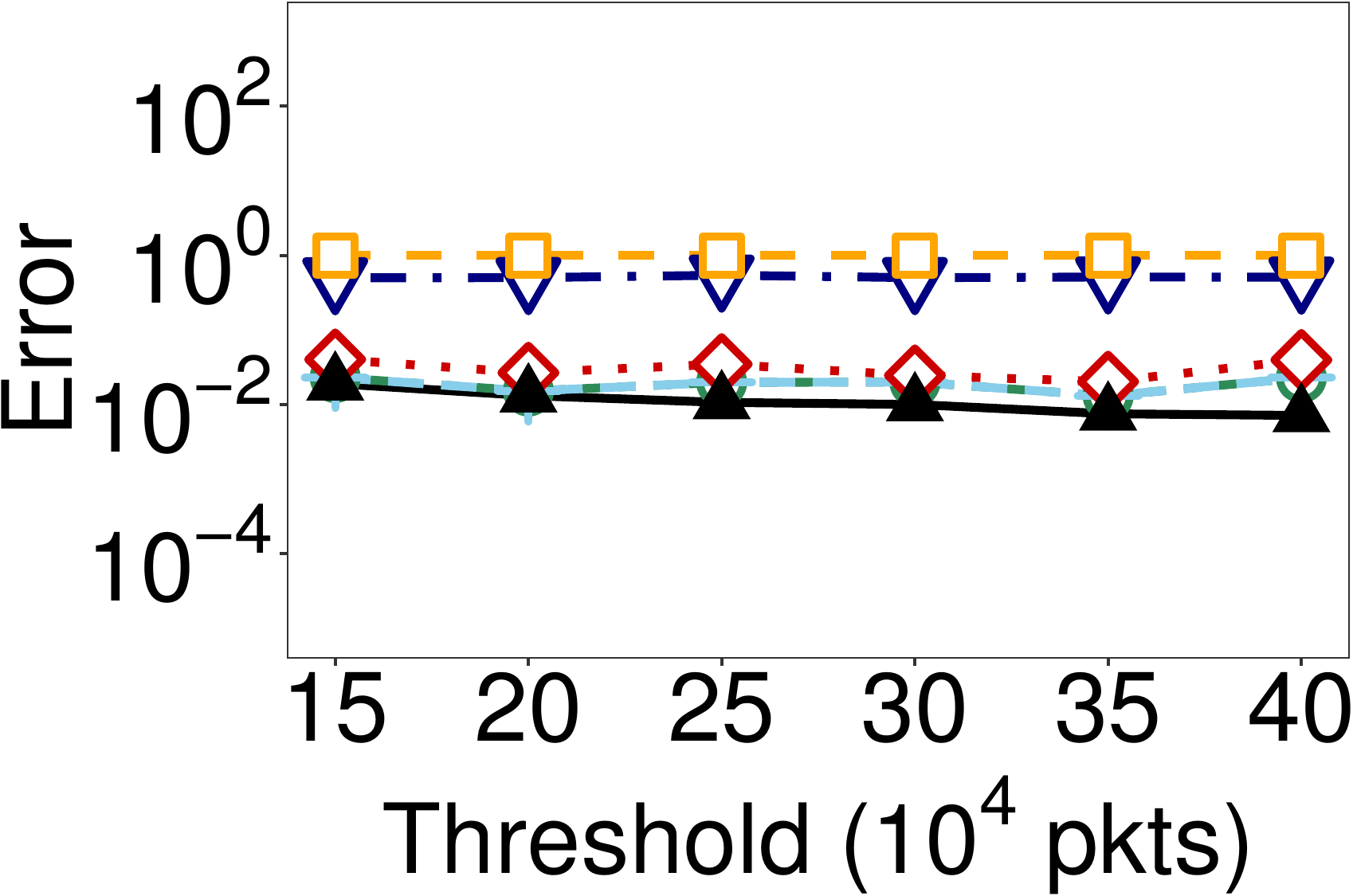} &
\includegraphics[width=1.7in]{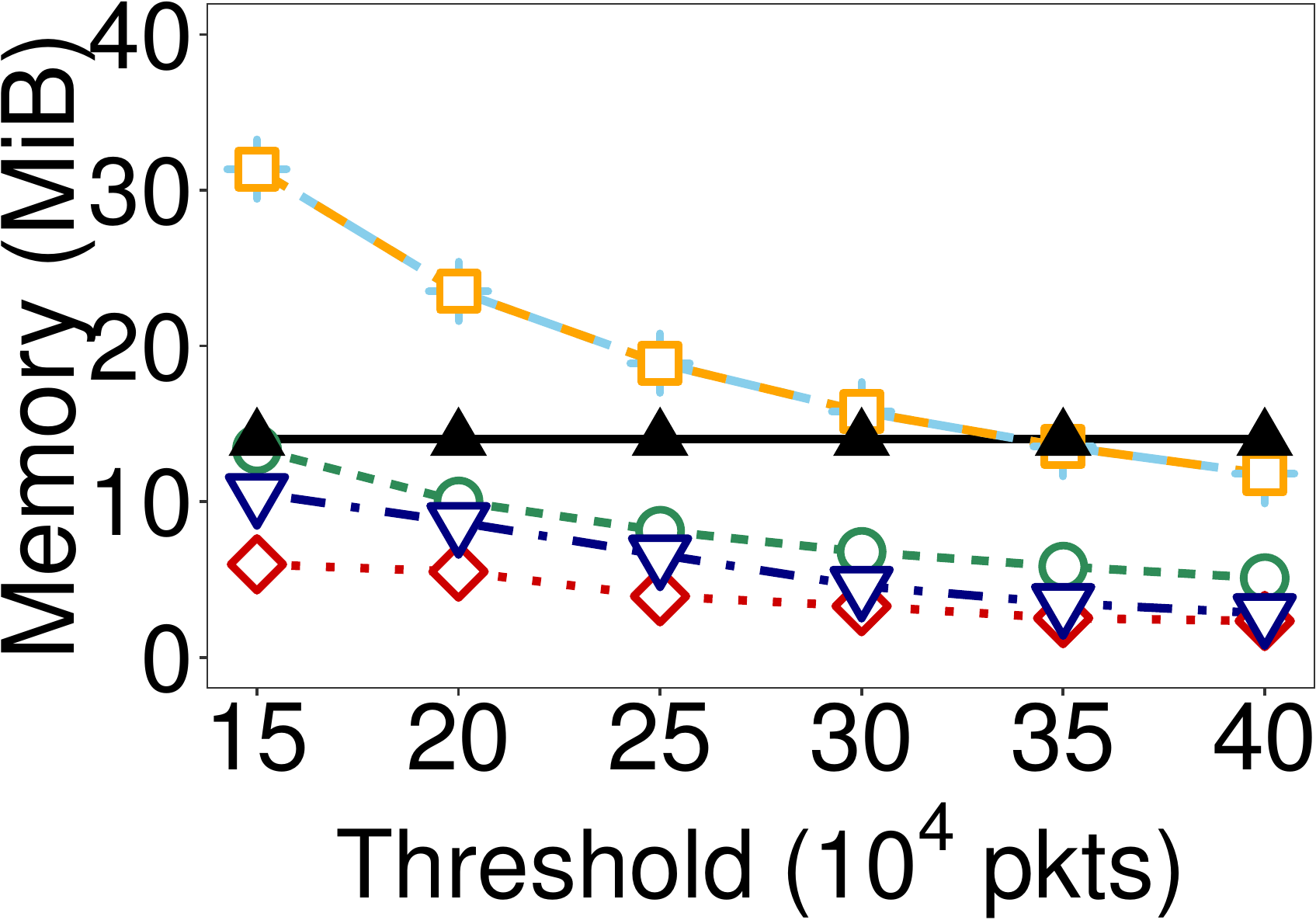}
\vspace{-3pt}\\
{\small (m) Precision for 2D-bit} & 
{\small (n) Recall for 2D-bit} &
{\small (o) Error for 2D-bit} &
{\small (p) Memory for 2D-bit}  
\end{tabular}
\vspace{-3pt}
\caption{(Experiment~1) Accuracy comparisons.} 
\label{fig:exp-1}
\vspace{-9pt}
\end{figure*}

%1-D #HHHs 1350 - 330
%2-D #HHHs 1540 - 290
\noindent
{\bf (Experiment~1) Accuracy comparisons.} We compare different HHH
detection schemes on accuracy versus different values of the 
absolute threshold $\phi\mathcal{S}$. We fix the memory space of \sysname as
256\,KiB, 1\,MiB, 1\,MiB, and 16\,MiB for 1D-byte, 1D-bit, 2D-byte, and 2D-bit
HHH detection, respectively.  We consider different absolute thresholds, such
that the number of true HHHs per epoch varies between 200 and 1,000.  

We consider three accuracy metrics: (i) {\em precision}, the ratio of true
HHHs reported over all reported HHHs (the denominator includes all true and
false HHHs); (ii) {\em recall}, the ratio of true HHHs reported over
all true HHHs (the denominator includes all reported and non-reported true
HHHs); and (iii) {\em relative error}, defined as
$\tfrac{1}{|\mathcal{H}|}\sum_{x\in \mathcal{H}}
\tfrac{|\hat{S}(x)-S(x)|}{S(x)}$, where $\mathcal{H}$ is the set of true HHHs
reported.  Note that an HHH is identified by both its prefix and subnet mask.
For example, it is treated as an error if an HHH 1.2.3.4/32 is reported as
1.2.3.4/31 in 1D-bit HHH detection. 

We also measure the memory usage of each scheme based on the number of
counters allocated in its data structure.  As both FULL and PARTIAL
dynamically allocate memory space in each epoch, we report their peak memory
usage.  

Figure~\ref{fig:exp-1} shows the results.  \sysname achieves higher accuracy
in most cases compared to others in all cases.
RHHH achieves a precision below 0.85 and 0.25 for byte-level and
bit-level HHH detection, respectively, with a relative error of around 100\%.
The reason is that RHHH has slow convergence and needs to process sufficient
packets in order to converge to high accuracy (see Experiment~6 for further
analysis).  Both HSS and USS have comparable accuracy to \sysname in 1D-byte
and 1D-bit HHH detection, yet their precisions are significantly lower than
\sysname in 2D-byte and 2D-bit HHH detection (e.g., their precisions are
around 0.6 in 2D-bit precision), mainly because they estimate the conditioned
count of a key in a more conservative way.  TRIE, FULL, and PARTIAL have low
accuracy in all settings. We observe that the accuracy of
\sysname increases with the threshold (i.e., fewer HHHs), while those of other
schemes remain almost the same for all thresholds.  The reason is that
\sysname adopts static memory allocation and its memory size is fixed for all
thresholds, while the memory sizes of other schemes decrease as the threshold
increases (\S\ref{subsec:methodology}).

For memory usage, \sysname maintains a medium size of memory usage among
all schemes.  RHHH and USS have the highest memory usage in most cases, as
they implement multiple Space Saving instances \cite{Metwally2005}, each of
which comprises a hash table and multiple doubly linked lists.  FULL and
PARTIAL have the smallest memory usage, as they dynamically kick out small
keys and keep only large keys in their counter arrays; however, such dynamic
memory allocation incurs high update overhead (Experiment~3). 

\paragraph{(Experiment~2) Robustness of \sysname under various memory sizes.}
We evaluate \sysname versus the absolute threshold
$\phi\mathcal{S}$ by varying the memory size allocated for \sysname.
We configure the memory size in the range from 256\,KiB to
2\,MiB for 1D-byte, 1D-bit, and 2D-byte HHH detection, while increasing the
memory size to the range from 8\,MiB to 14\,MiB for 2D-bit HHH detection for 
tracking many more nodes in the 2D-bit hierarchy.

Figure~\ref{fig:exp-2} shows the results. As expected, \sysname achieves
higher accuracy with larger memory sizes.  Also, the accuracy of \sysname is
fairly robust in different cases. For example, with a memory size of 1\,MiB,
both the precision and recall of \sysname are above 0.9 for most of the
absolute threshold settings in 1D-bit, 1D-byte, and 2D-byte HHH detection.  

\begin{figure*}[!t]
\centering
\begin{tabular}{@{\ }c@{\ }c@{\ }c@{\ }c}
\includegraphics[width=1.7in]{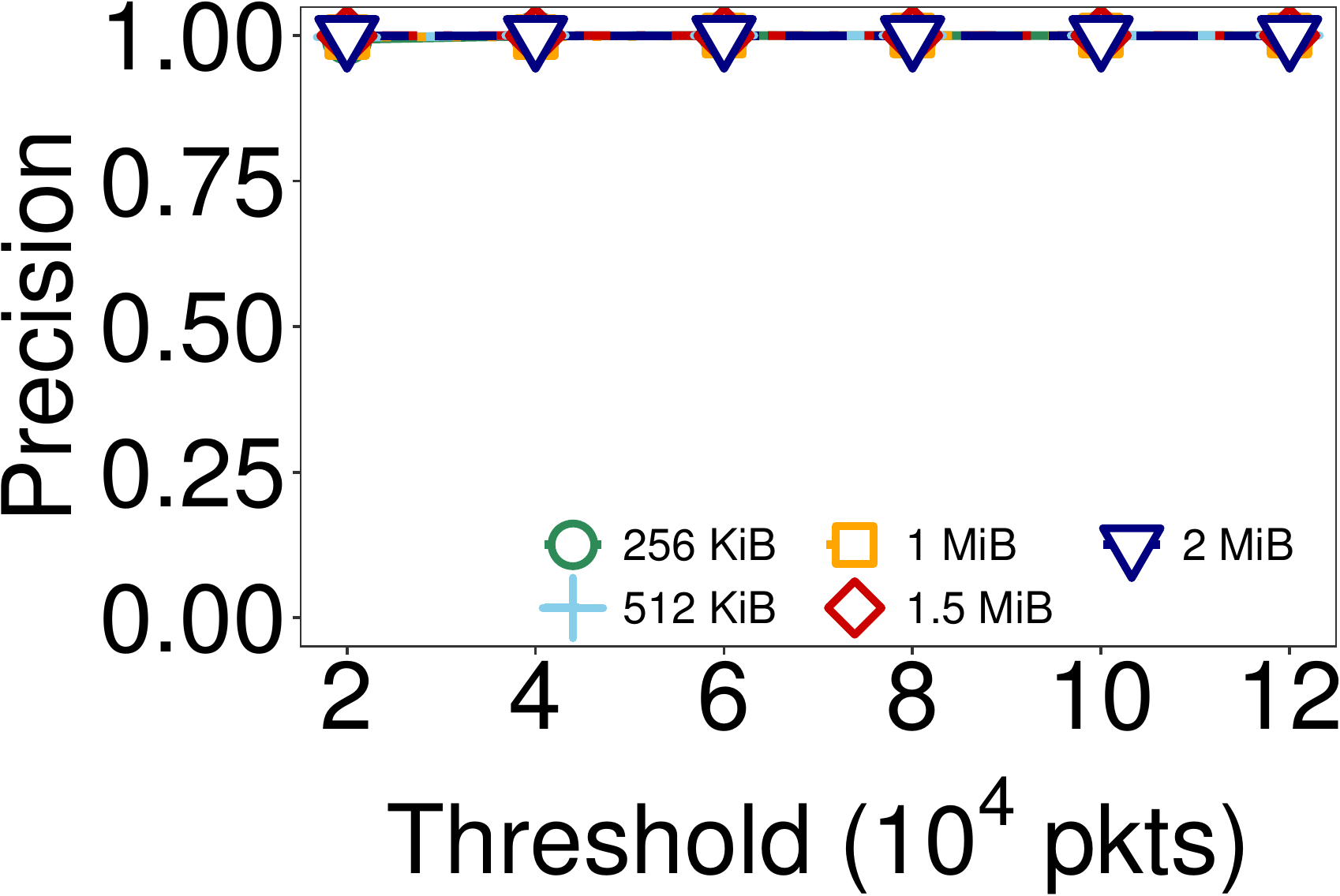} &
\includegraphics[width=1.7in]{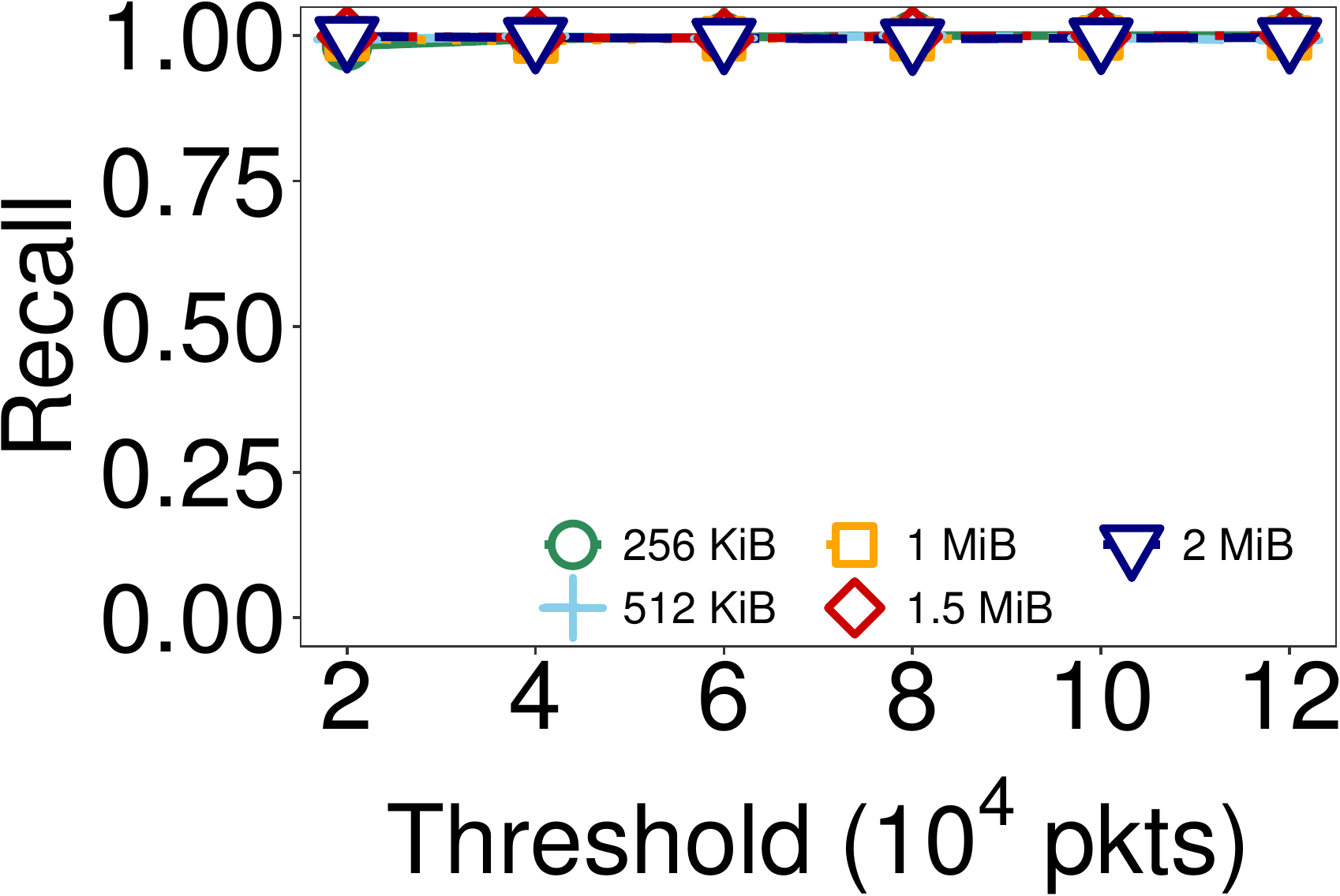} & 
\includegraphics[width=1.7in]{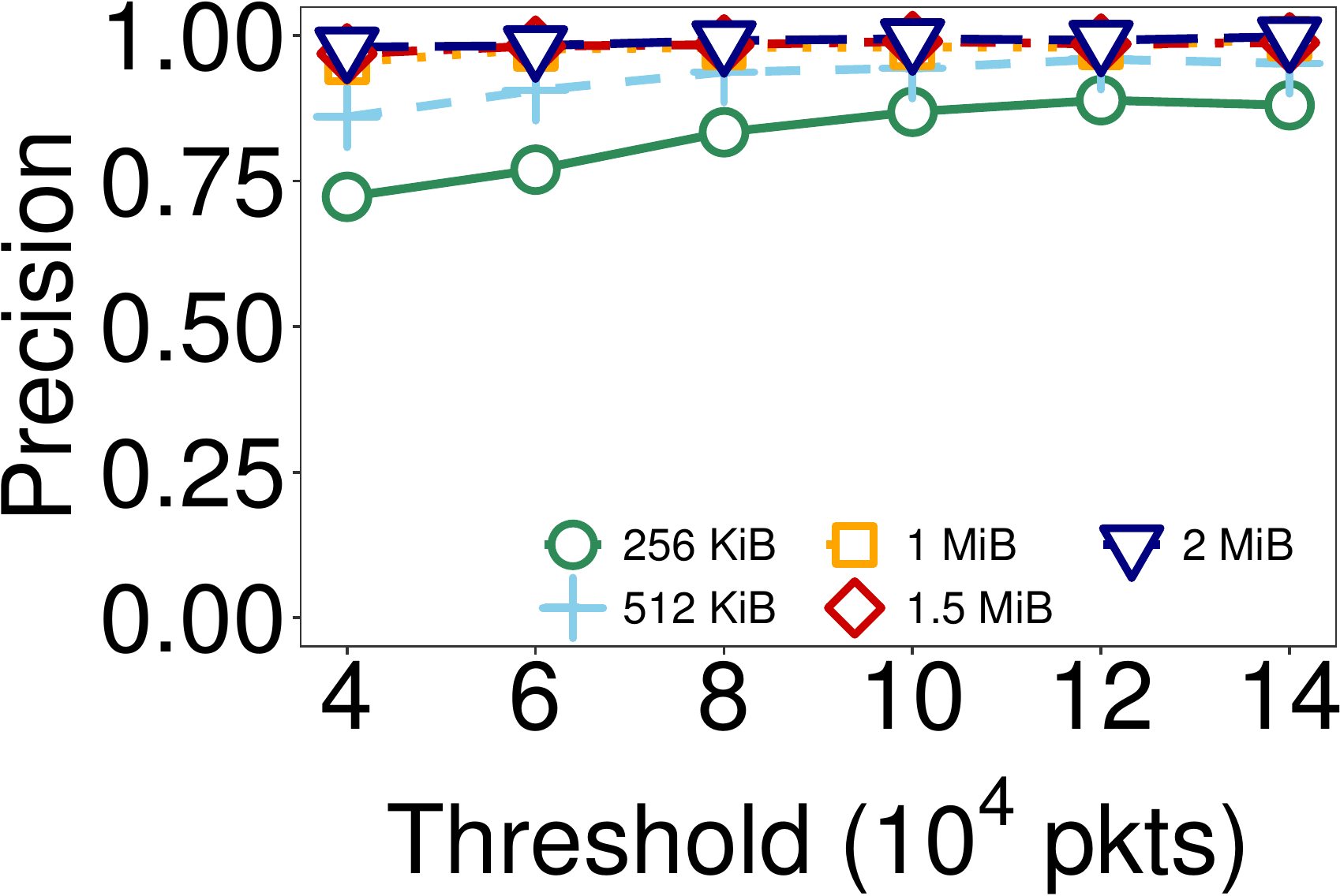} &
\includegraphics[width=1.7in]{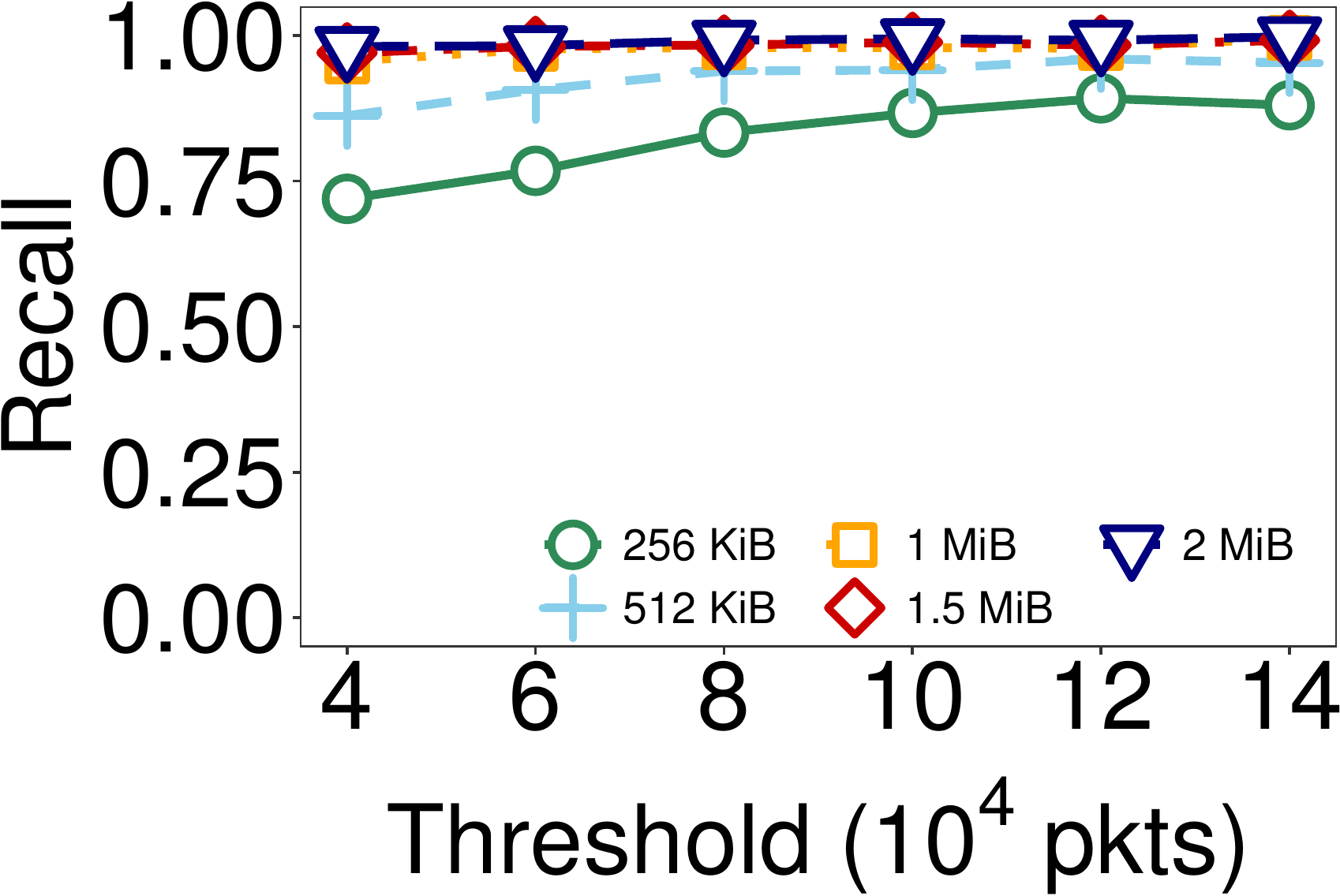} 
\vspace{-3pt}\\
{\small (a) Precision for 1D-byte} & 
{\small (b) Recall for 1D-byte} & 
{\small (c) Precision for 1D-bit} & 
{\small (d) Recall for 1D-bit} 
\vspace{3pt}\\
\includegraphics[width=1.7in]{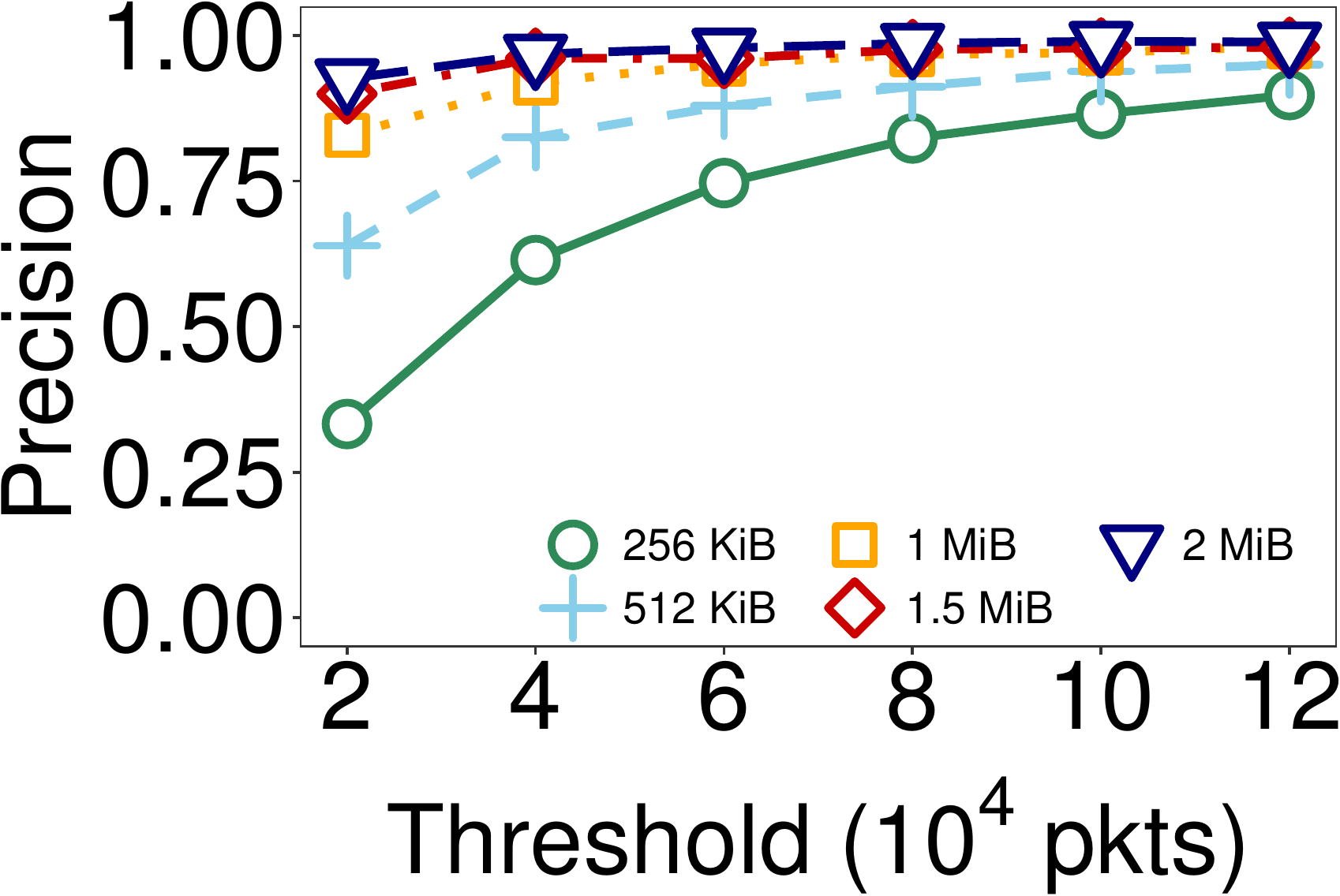} &
\includegraphics[width=1.7in]{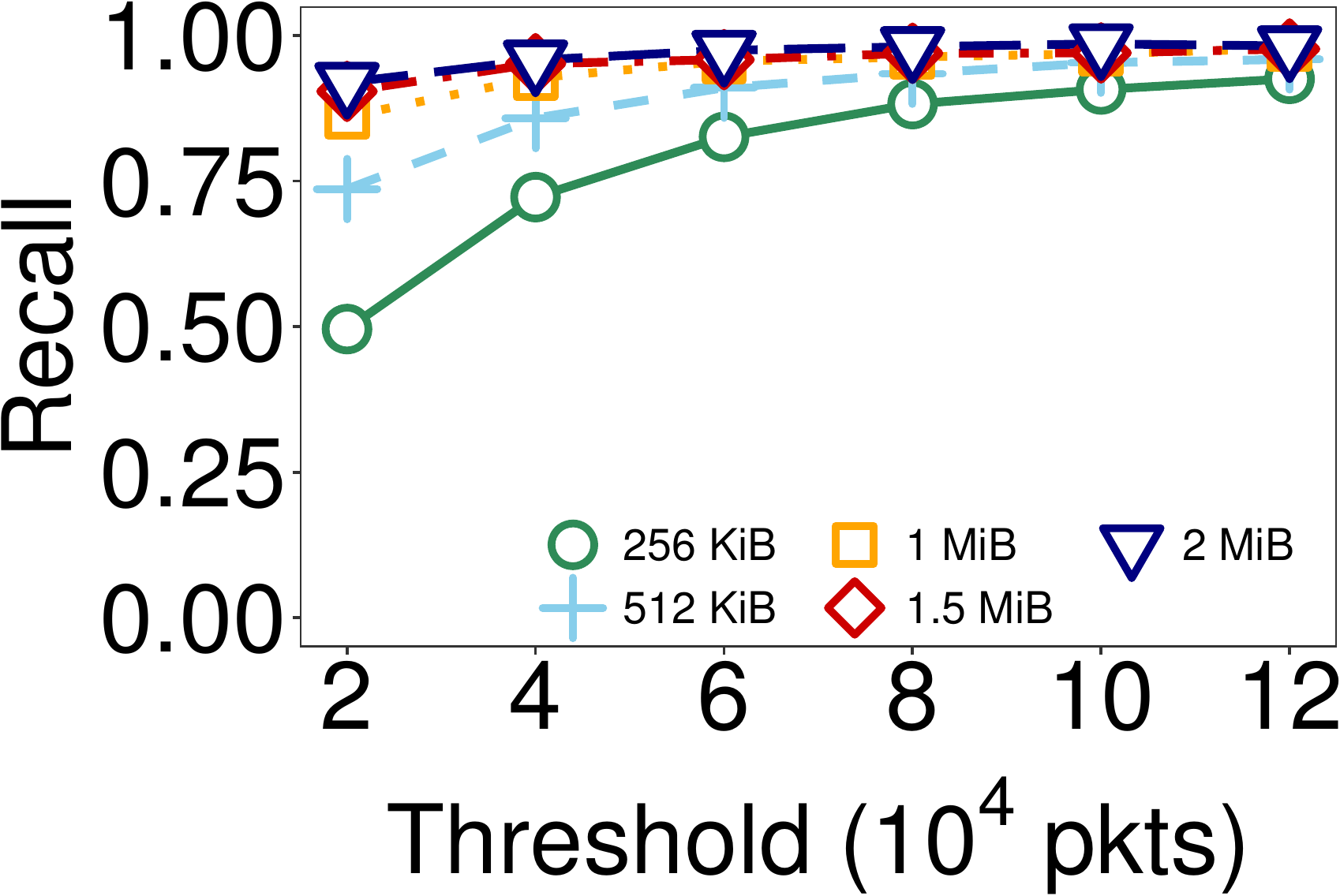} & 
\includegraphics[width=1.7in]{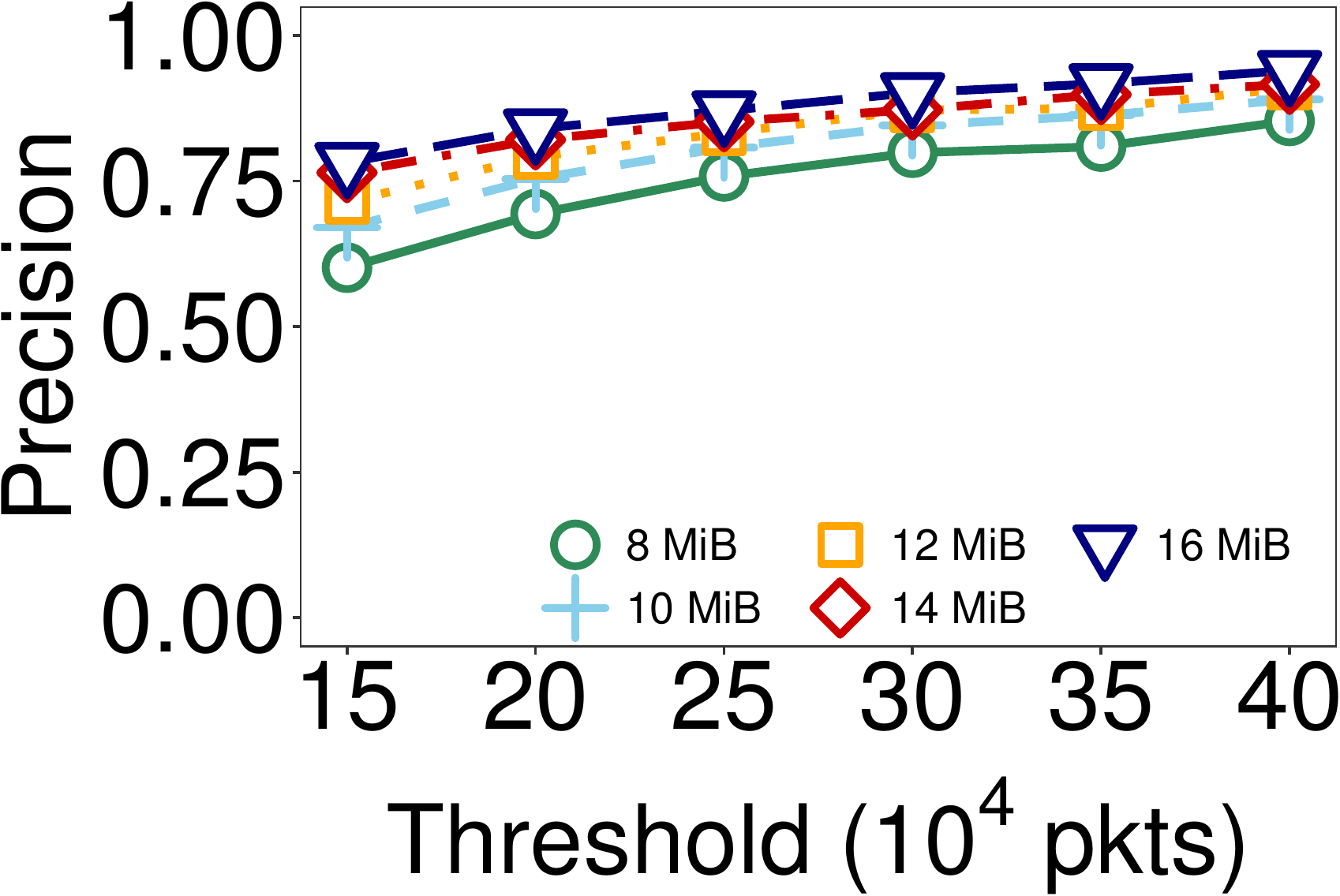} &
\includegraphics[width=1.7in]{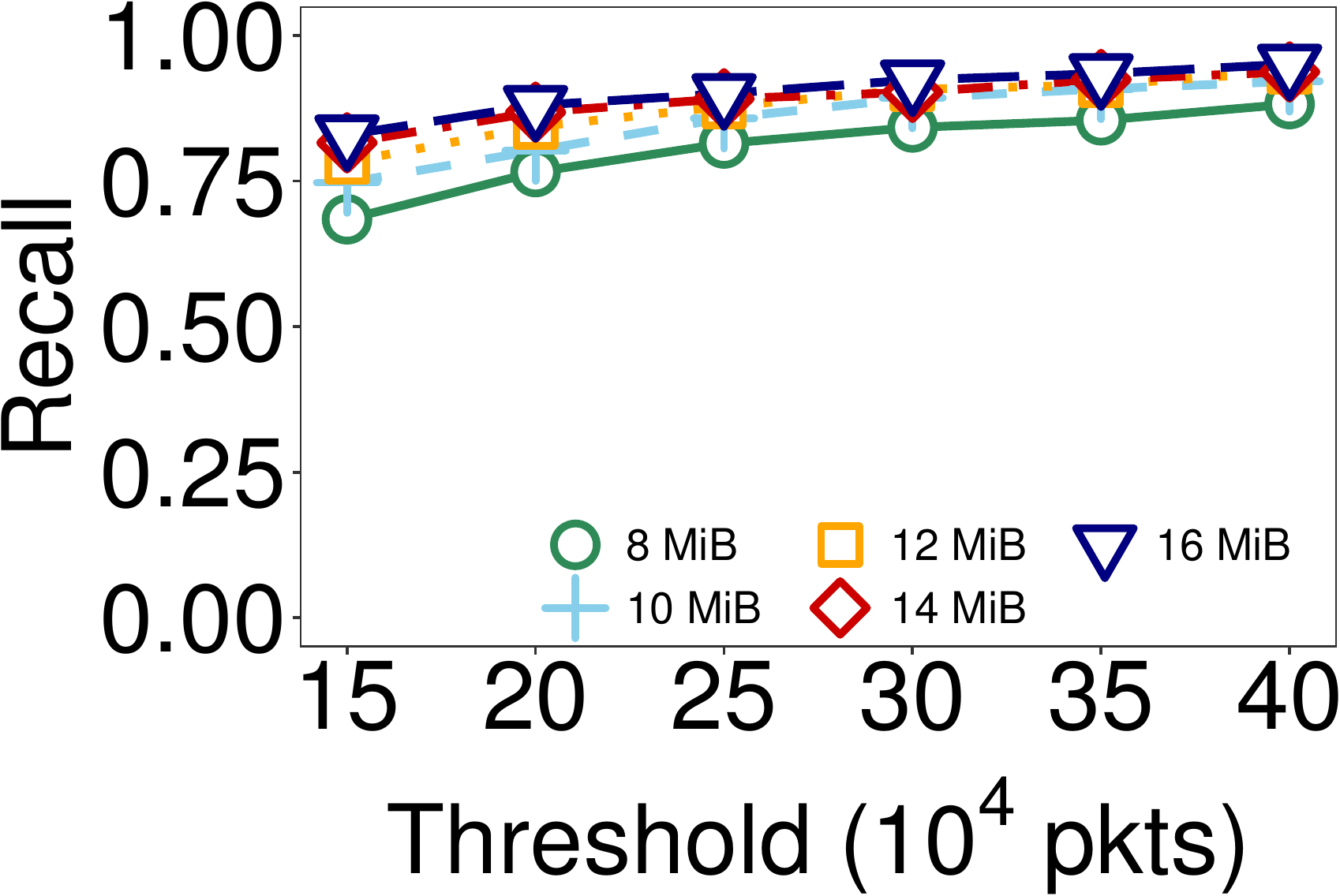} 
\vspace{-3pt}\\
{\small (e) Precision for 2D-byte} & 
{\small (f) Recall for 2D-byte} & 
{\small (f) Precision for 2D-bit} & 
{\small (h) Recall for 2D-bit} 
\end{tabular}
\vspace{-3pt}
\caption{(Experiment~2) Robustness of \sysname under various memory sizes.}
\label{fig:exp-2}
\vspace{-6pt}
\end{figure*}

\paragraph {(Experiment~3) Update throughput.} We benchmark the update
throughput of all HHH detection schemes on a server equipped with an Intel Xeon
E5-1630 3.70\,GHz CPU and 16\,GiB RAM. The server runs Ubuntu 14.04.5. To
exclude disk I/O overhead and stress-test each scheme, we first load the whole
trace into memory before running the experiment, and then process the trace as
fast as possible. Here, we focus on 1D-byte and 1D-bit HHH detection, while
similar performance trends are observed for 2D-byte and 2D-bit HHH detection.
We keep the same memory size setting for \sysname as in Experiment~1 and fix
the absolute threshold as 100,000 packets.

\begin{figure}[!t]
\centering
\begin{tabular}{@{\ }c@{\ }c}
\includegraphics[width=1.7in]{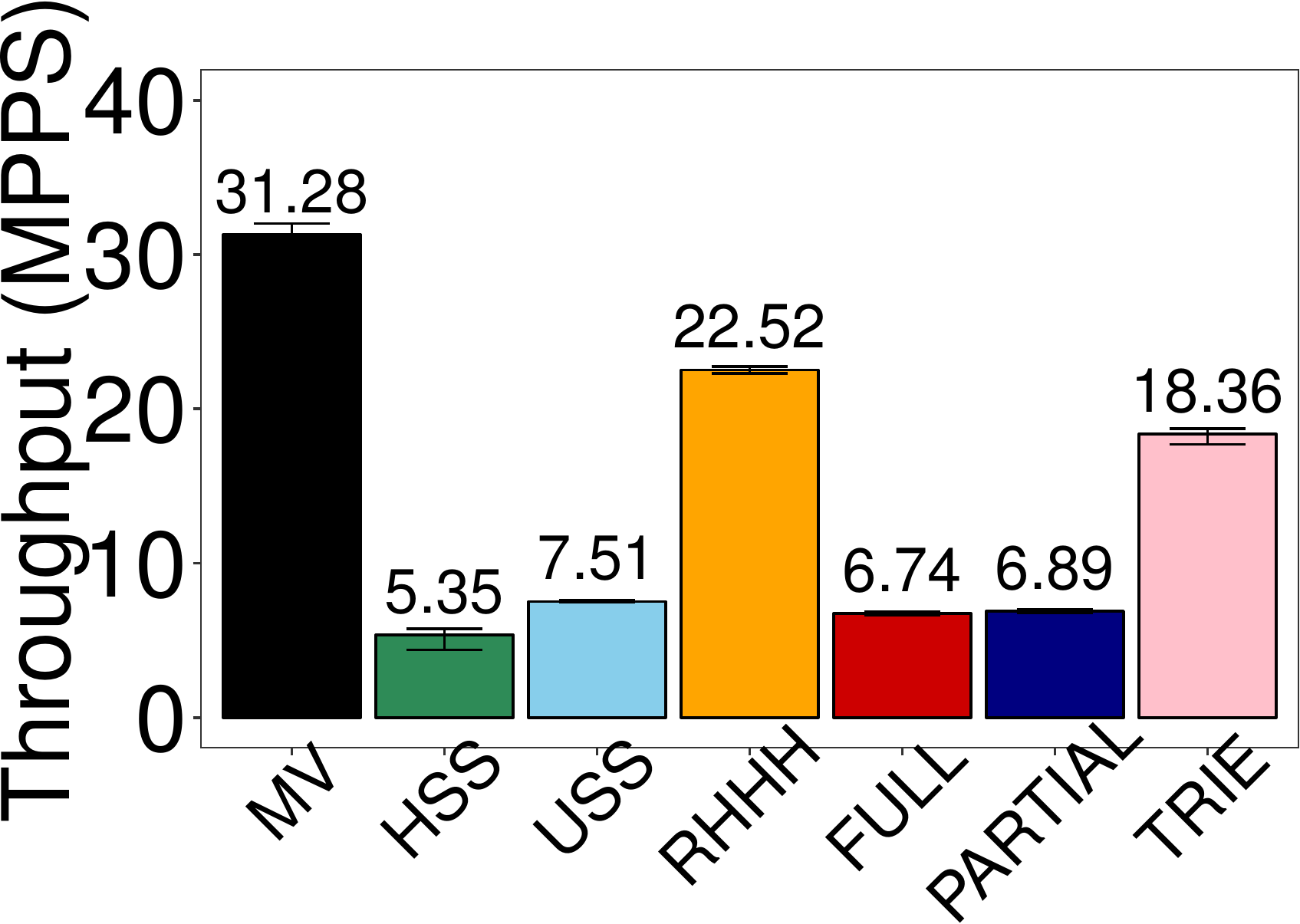} &
\includegraphics[width=1.7in]{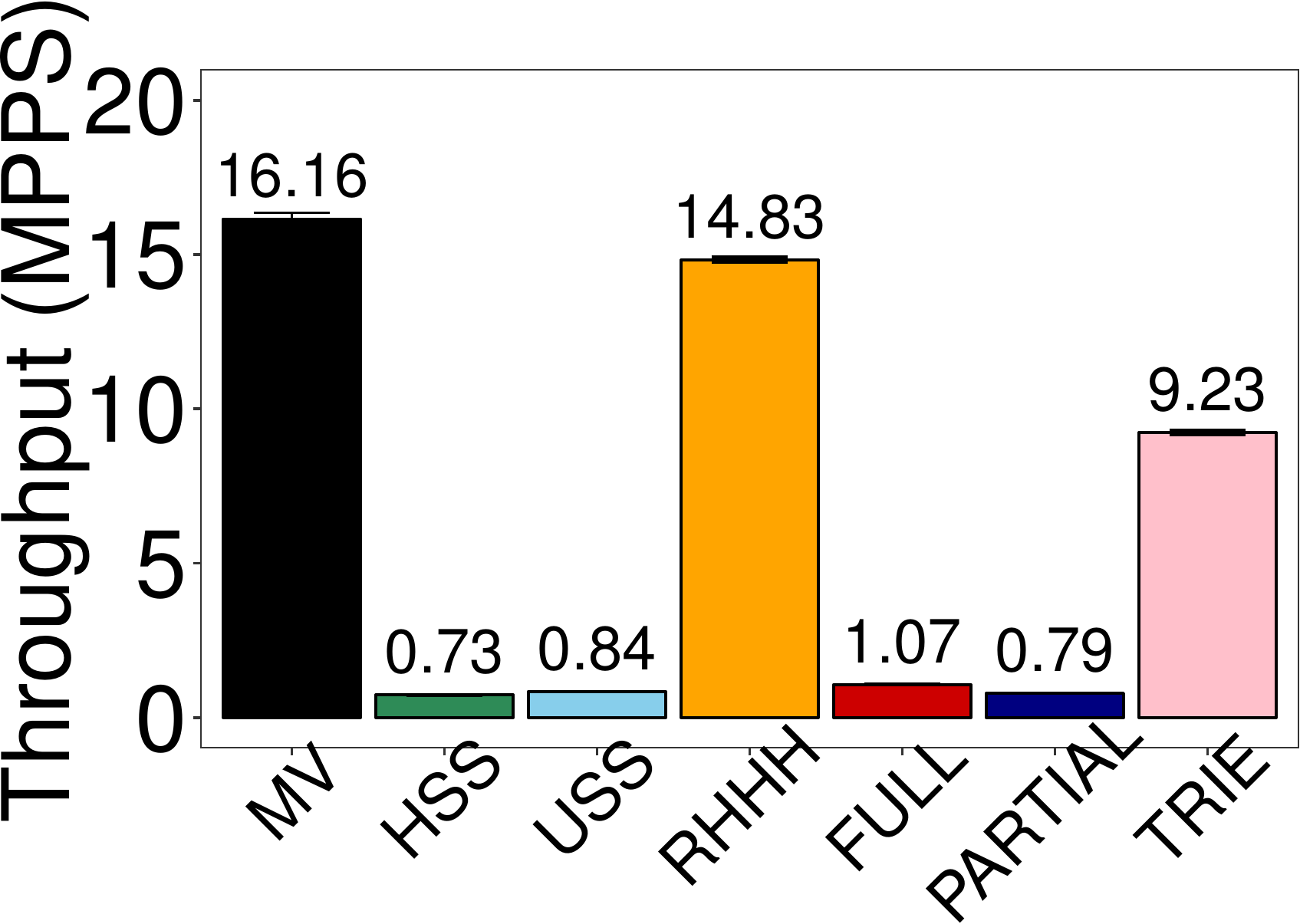} 
\vspace{-3pt}\\
{\small(a) 1D-byte } &
{\small(b) 1D-bit }
\end{tabular}
\vspace{-3pt}
\caption{(Experiment~3) Update throughput.}
\label{fig:exp-3-1}
\vspace{-6pt}
\end{figure}

Figures~\ref{fig:exp-3-1}(a) and \ref{fig:exp-3-1}(b) show the update
throughput of all schemes in million packets per second (MPPS) for 1D-byte and
1D-bit HHH detection, respectively; each error bar shows the maximum and
minimum throughput across different epochs for each scheme. 
\sysname achieves the highest throughput with up to $5.84\times$ and
$22.13\times$ throughput gain for byte-level and bit-level HHH detection,
respectively.  Both HSS and USS have the lowest throughput as they update the
sketch instance for every node in the hierarchy for each packet.  FULL,
PARTIAL, and TRIE also have low throughput, as they dynamically expand or
shrink their data structures during packet updates.  

Although RHHH supports constant-time updates per packet \cite{BenBasat2017},
it has lower throughput than \sysname in 1D HHH detection. The reason is that
for each packet update, RHHH accesses a single Space Saving instance, but
may incur multiple pointer assignments to update the linked lists in the Space
Saving data structure \cite{Metwally2005}.  RHHH can increase its throughput
via packet sampling (e.g., 10\% of packets in 10-RHHH \cite{BenBasat2017}),
but it increases the convergence time and has low accuracy. 

\begin{figure*}[!t]
\centering
\begin{minipage}{.5\textwidth}
\centering
\begin{tabular}{@{\ }cc}
\multicolumn{2}{c}{\includegraphics[width=2in]{fig/legend.pdf}} \\
\includegraphics[width=1.6in]{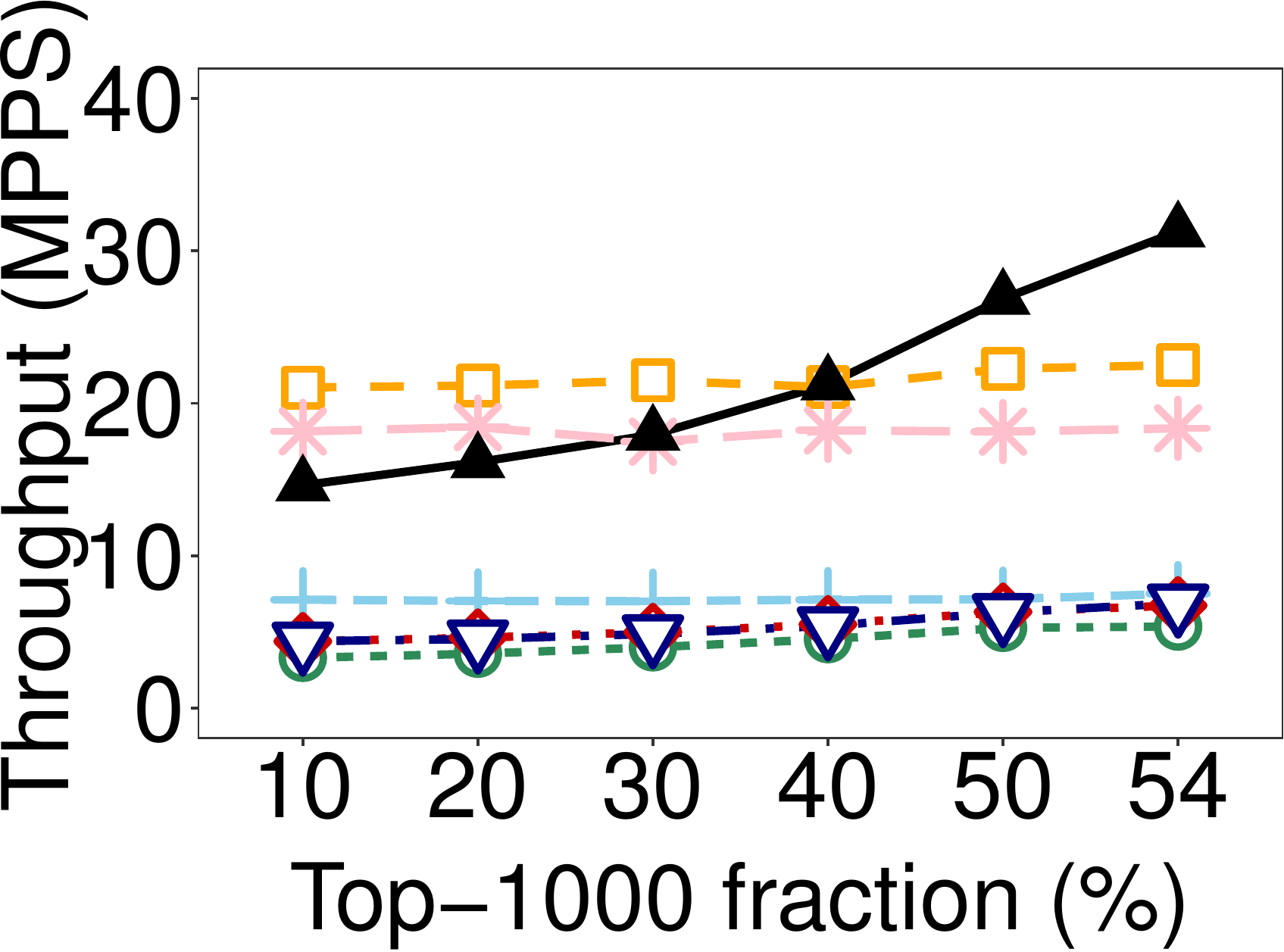} &
\includegraphics[width=1.6in]{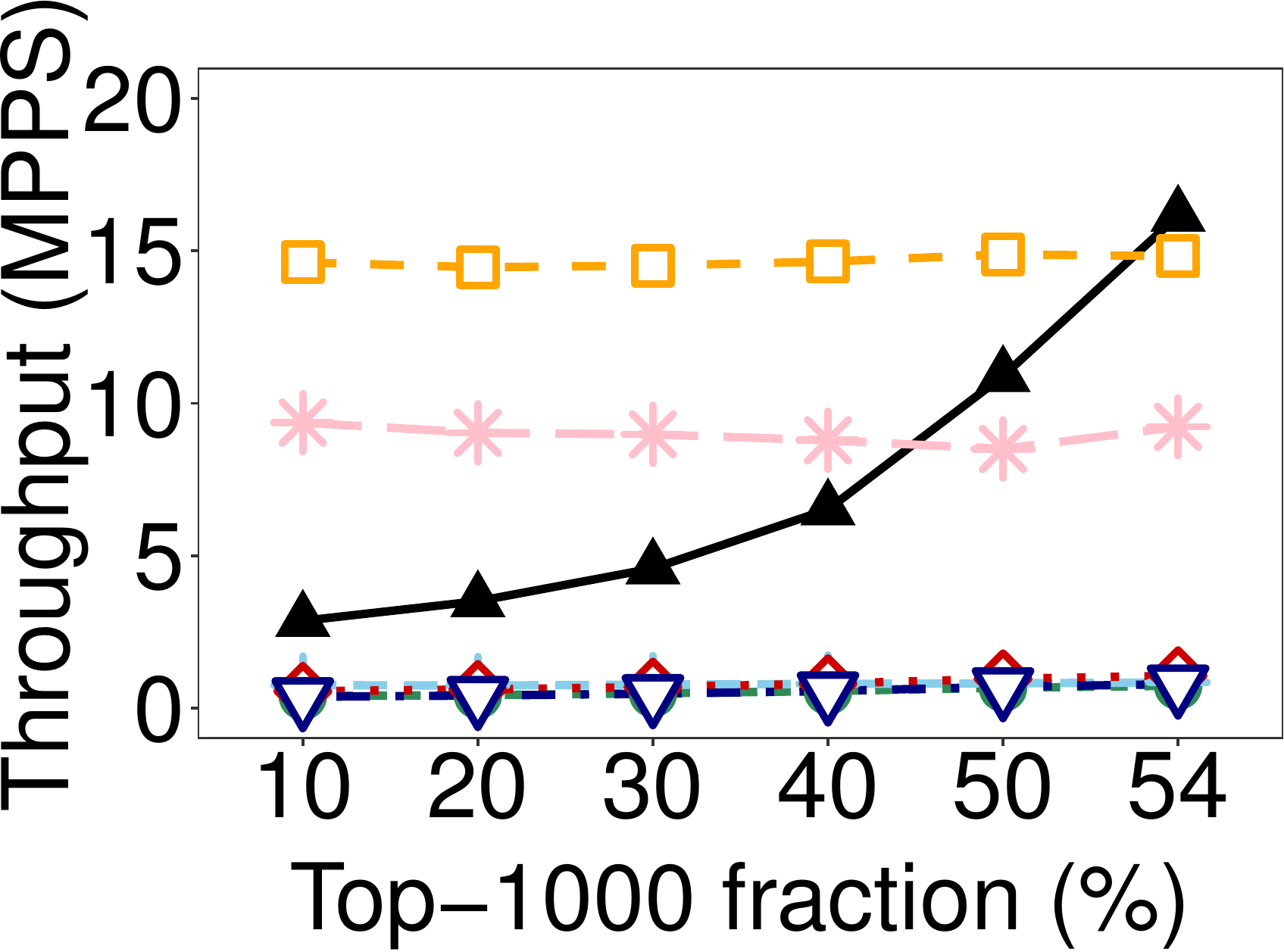}  
\vspace{-3pt}\\
{\small (a) 1D-byte} & 
{\small (b) 1D-bit} 
\end{tabular}
\vspace{-3pt}
\captionof{figure}{(Experiment~4) Throughput vs. skewness.}
\label{fig:exp-4}
\end{minipage}%
\begin{minipage}{.5\textwidth}
\centering
\begin{tabular}{@{\ }c@{\ }c}
\multicolumn{2}{c}{\includegraphics[width=3in]{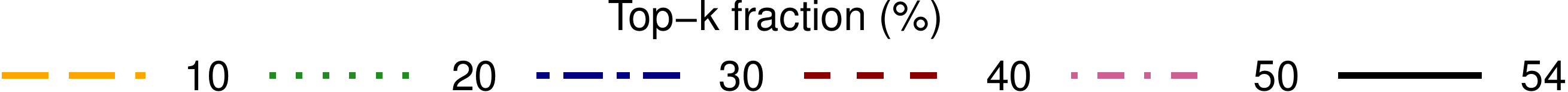}} \\
\includegraphics[width=1.7in]{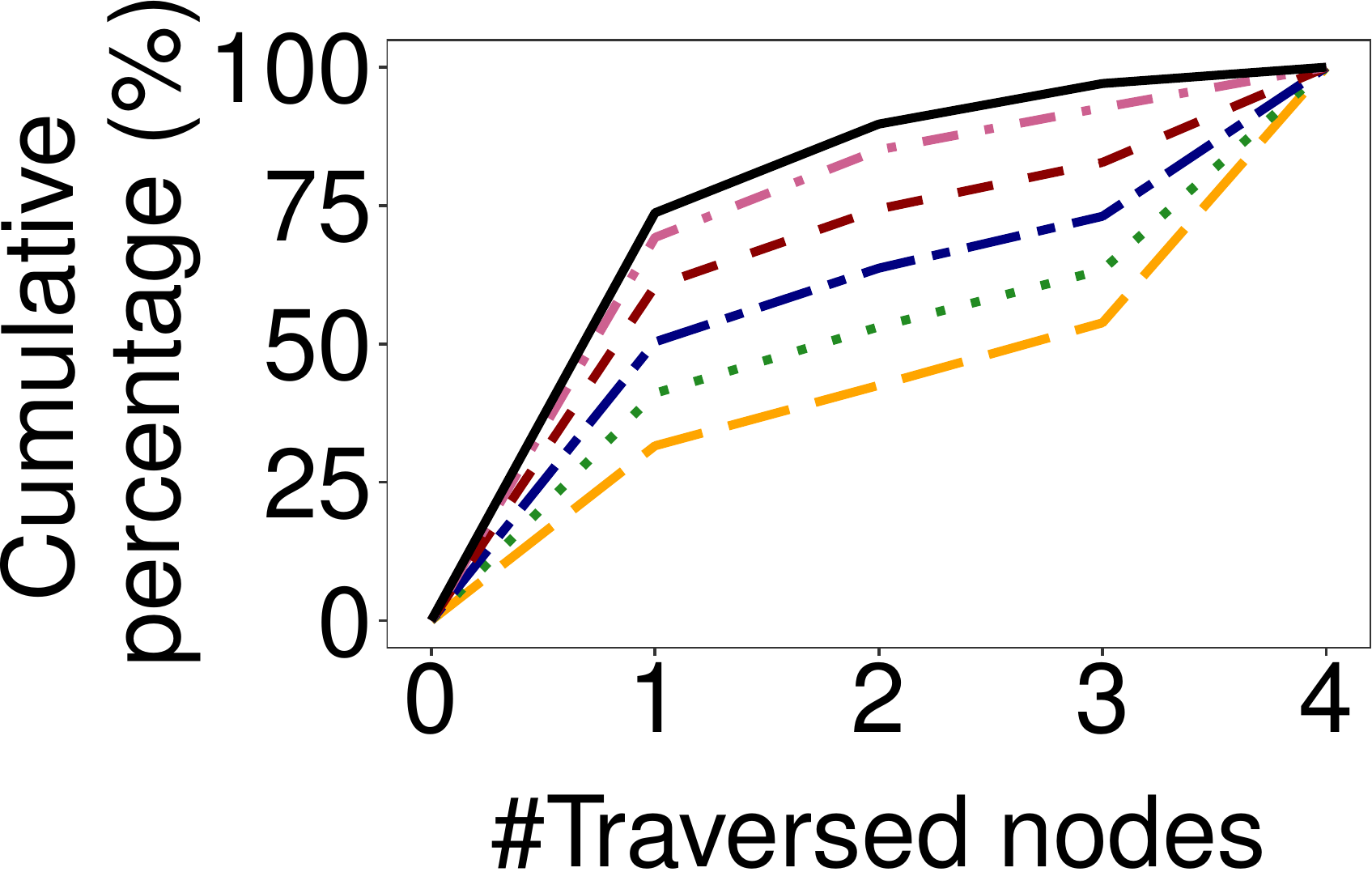} &
\includegraphics[width=1.7in]{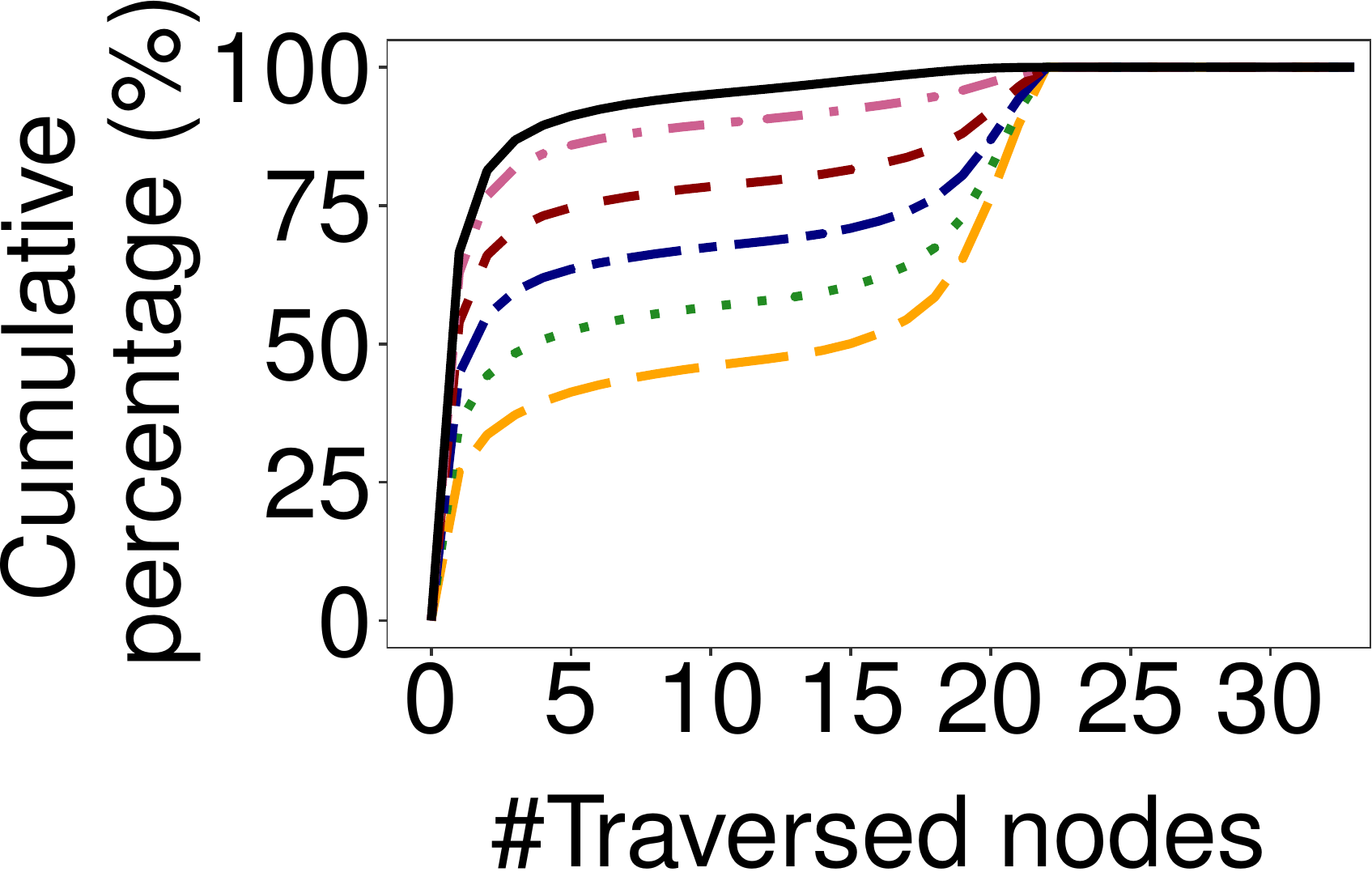}
\vspace{-3pt}\\
{\small (a) 1D-byte} & 
{\small (b) 1D-bit}  
\end{tabular}
\vspace{-3pt}
\captionof{figure}{(Experiment~5) Number of traversed nodes.}
\label{fig:exp-5}
\end{minipage}
\end{figure*}

\begin{figure*}[!t]
\centering
\begin{tabular}{@{\ }c@{\ }c@{\ }c@{\ }c}
\includegraphics[width=1.7in]{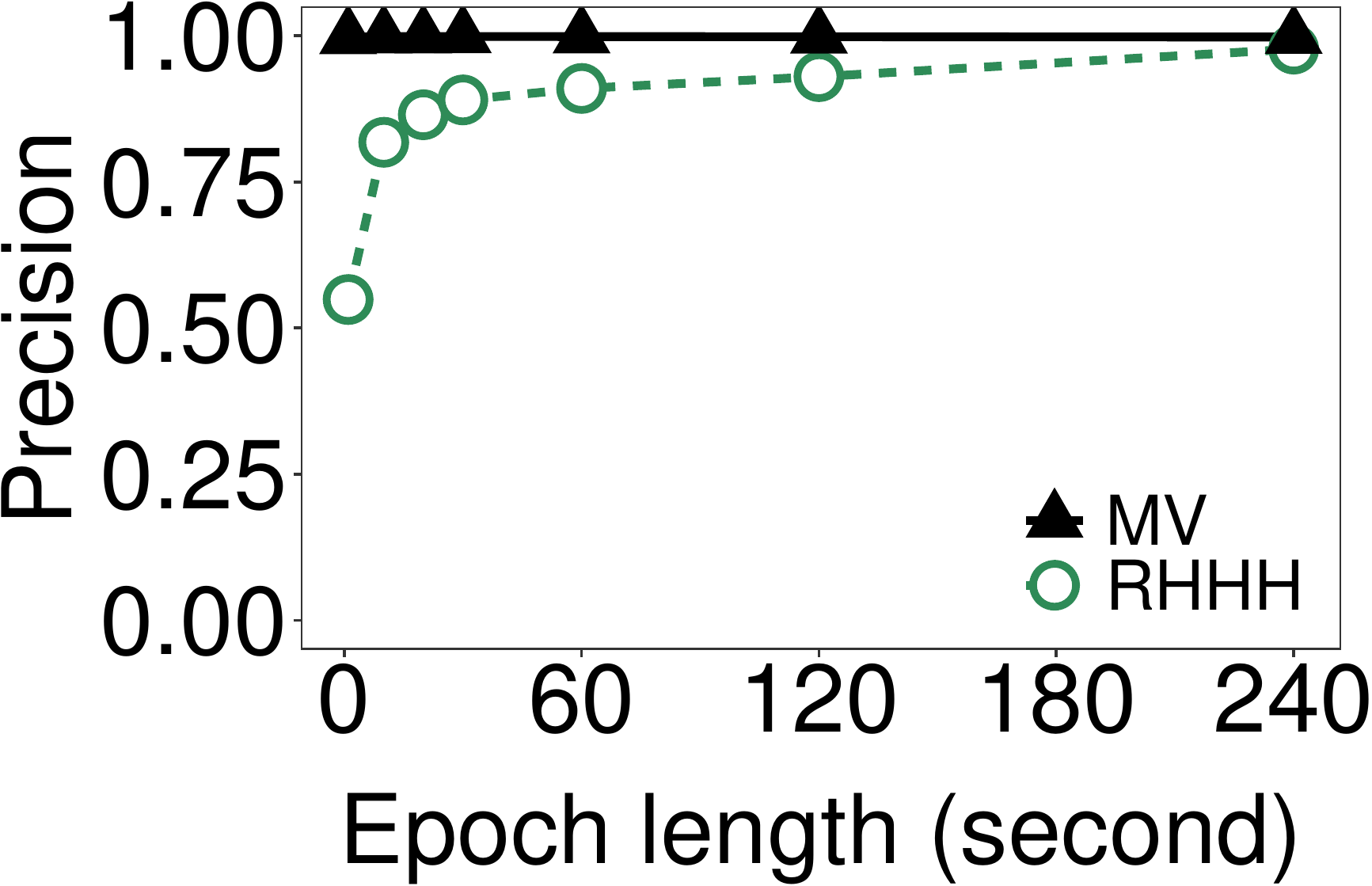} &
\includegraphics[width=1.7in]{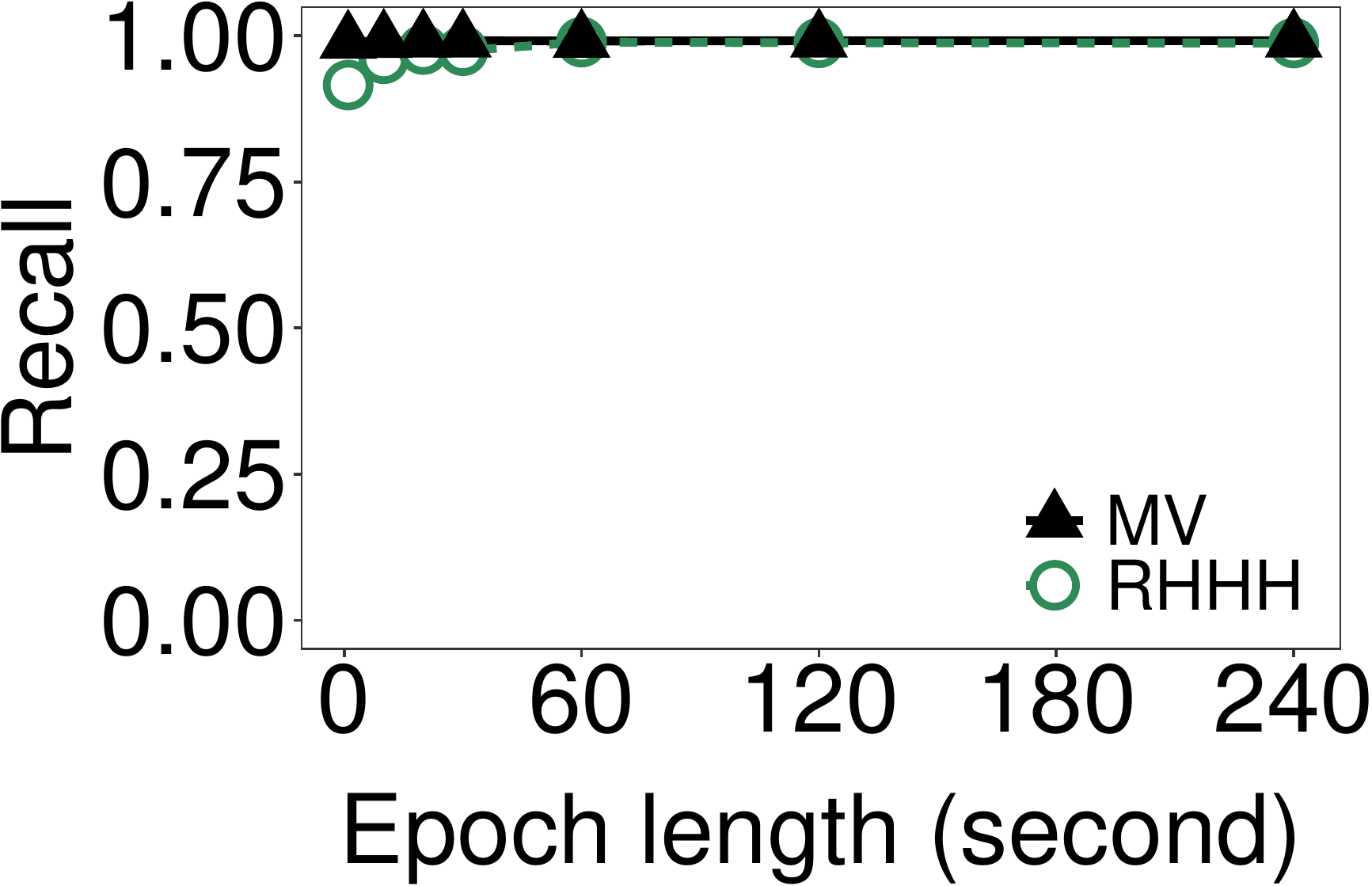} & 
\includegraphics[width=1.7in]{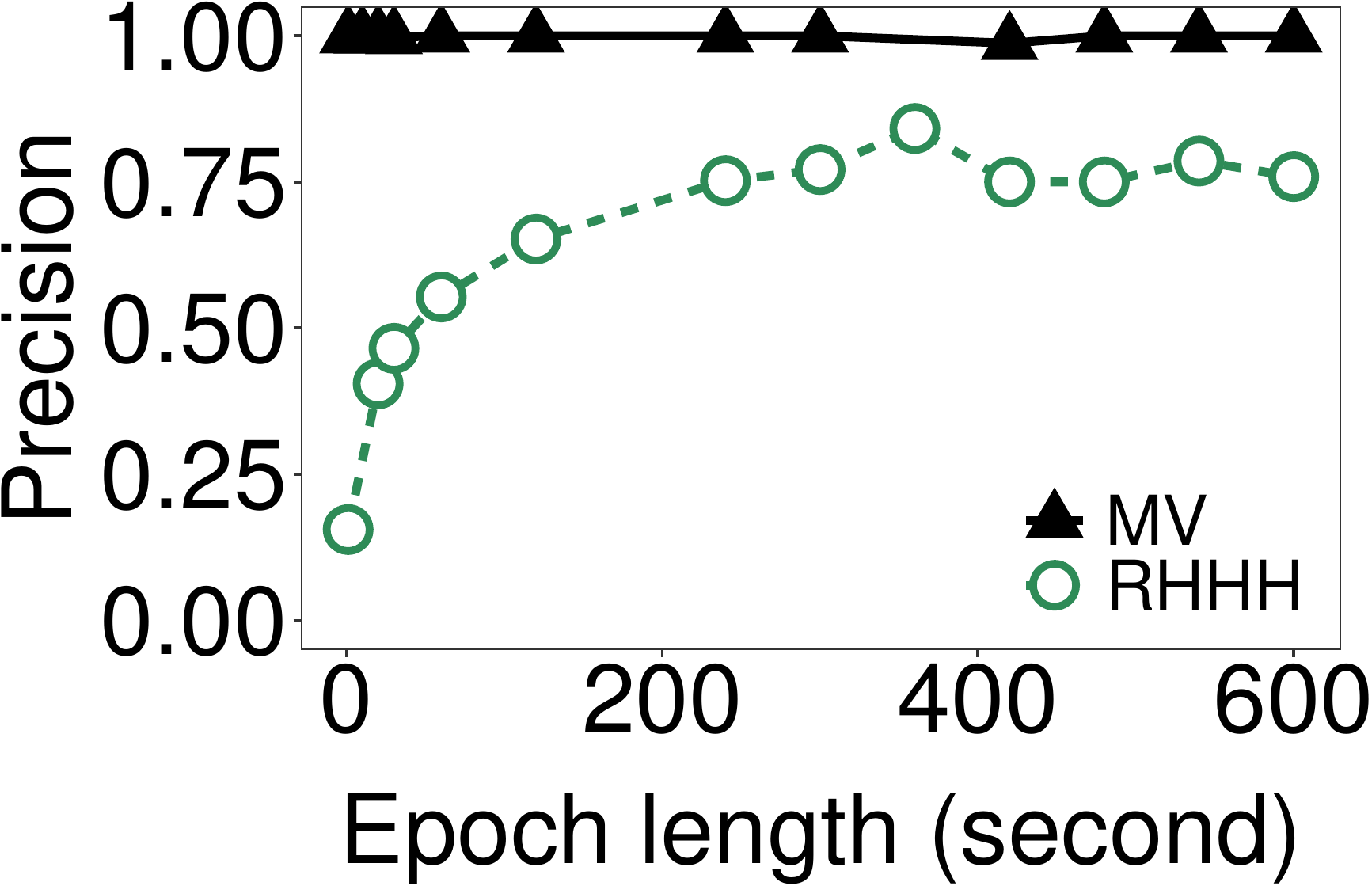} &
\includegraphics[width=1.7in]{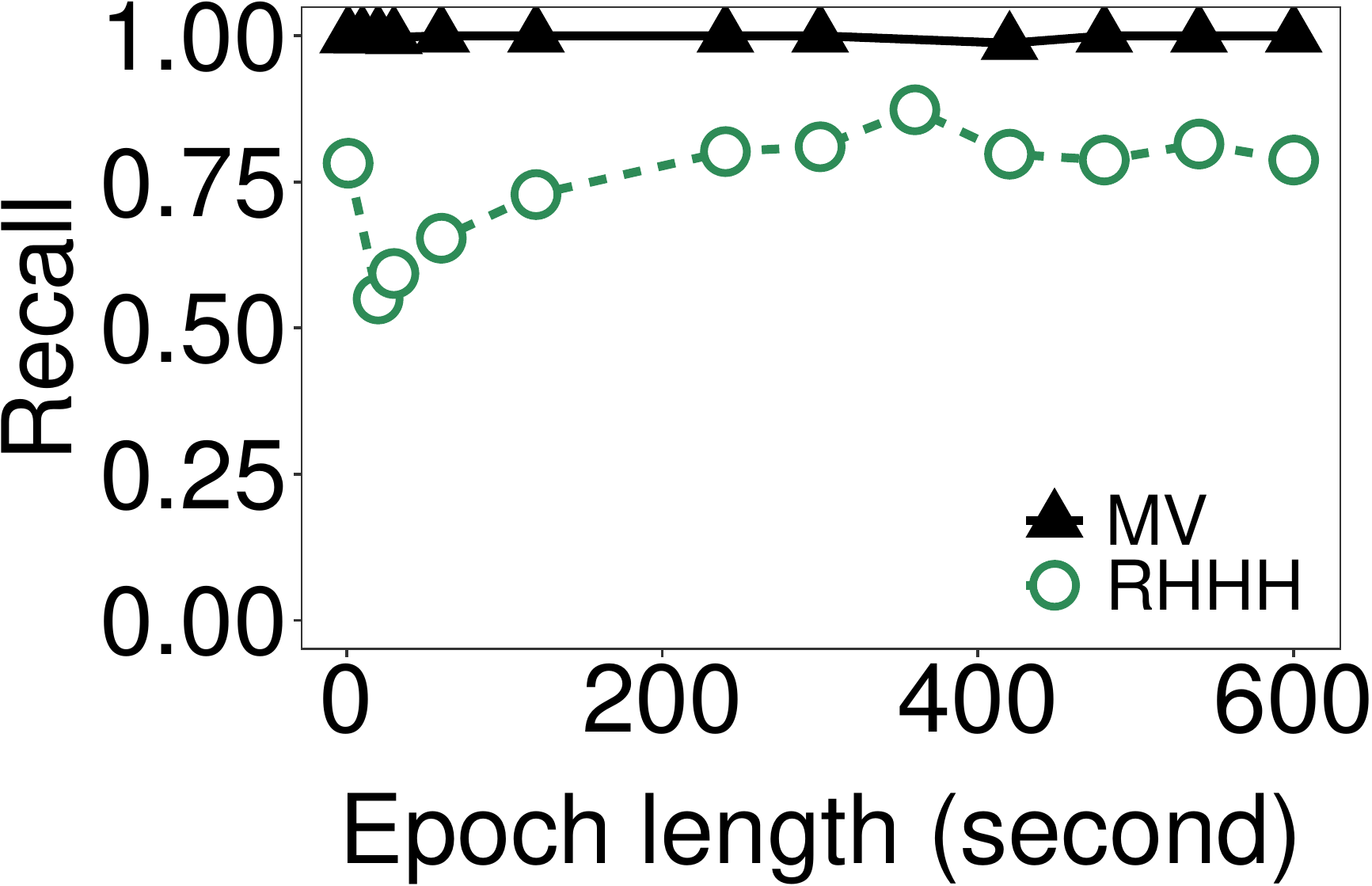} 
\vspace{-3pt}\\
{\small (a) Precision for 1D-byte} & 
{\small (b) Recall for 1D-byte} &
{\small (c) Precision for 1D-bit} & 
{\small (d) Recall for 1D-bit}
\end{tabular}
\vspace{-3pt}
\caption{(Experiment~6) Convergence.}
\label{fig:exp-6}
\vspace{-6pt}
\end{figure*}

\paragraph{(Experiment 4) Throughput versus skewness.}  While \sysname
is designed for highly skewed workloads, we evaluate the update
throughput of \sysname for less skewed workloads by varying the skewness
degree of the CAIDA traces.  In the original CAIDA traces used in our
evaluation, the top-1000 flows account for 54\% of the total number of packets
in the traces. We vary the skewness degree of the traces by controlling the
fraction of the total number of packets occupied by the top-1000 flows in each
epoch.  Specifically, we replace some packets of the top-1000 flows with new
packets that have randomly generated source and destination IP addresses, such
that the top-1000 flows account for a specified fraction (varied from 10\% to
50\%) of the total number of packets in each epoch.  A smaller specified
fraction implies a less skewed workload.  Here, we set $w_0$ as 5,000 and
3,000 in 1D-byte and 1D-bit detection, respectively. 

Figure~\ref{fig:exp-4} shows the update throughput of all schemes under
various skewness degrees for 1D-byte and 1D-bit HHH detection.  The throughput
of \sysname drops quickly as the specified fraction decreases (i.e., less
skewed), as more packets need to be pushed to higher levels.
Although \sysname's throughput decreases for less skewed workloads, its
throughput remains higher than other schemes except for RHHH and TRIE.

\paragraph{(Experiment 5) Number of traversed nodes.} To understand the update
throughput of \sysname, we collect the number of nodes being traversed by
\sysname in a hierarchy for each packet update in both 1D-byte and 1D-bit
hierarchies for different skewness degrees as specified in Experiment~4. 

Figure~\ref{fig:exp-5} shows the cumulative percentage of packets versus the
number of traversed nodes by \sysname for different skewness degrees.  We
first examine the results for the original CAIDA traces (i.e., the top-1000
fraction is 54\%).  In 1D-byte HHH detection, 73\% of packet updates traverse
only one node in the hierarchy, where each packet update on average traverses
only 1.39 nodes.  In 1D-bit HHH detection, the number of traversed nodes
slightly increases: only 66\% of packet updates traverse one node, while each
packet update on average traverses 2.36 nodes.  The reason is that as the
number of nodes increases in the 1D-bit hierarchy, each packet update
generally needs to traverse more nodes in order to be admitted by a candidate
HHH.  Thus, \sysname has lower throughput in 1D-bit HHH detection than in
1D-byte HHH detection. Nevertheless, since each packet update traverses only
one node in a hierarchy in most cases, it justifies the high update throughput
of \sysname (see Figure~\ref{fig:exp-3-1} in Experiment~3).  

We examine the results when the skewness degree decreases. As the
top-1000 fraction decreases from 54\% to 10\%, the fraction of packet updates
traversing only one node decreases from 73\% to 31\% for 1D-byte HHH
detection, and from 66\% to 26\% for 1D-bit HHH detection. Correspondingly,
the average number of traversed nodes per update increases from 1.39 to 2.72,
and from 2.36 to 11.62, respectively.  This explains the throughput drop of
\sysname for less skewed workloads. 

\paragraph {(Experiment~6) Convergence.}  We study the convergence 
by comparing the accuracy between \sysname and RHHH \cite{BenBasat2017} for
various epoch lengths.  
We use the first twelve minutes of the CAIDA traces and vary
the epoch length from one second to ten minutes (our default is one minute),
where the number of packets in each epoch on average varies from 0.5\,M to
401\,M.  For each epoch length, we divide the traces into multiple epochs (if
the epoch length is larger than six minutes, we consider one epoch only).  A
small epoch length (e.g., one second) implies a small number of packets in the
epoch, and any scheme that requires sufficient packets for convergence may
have low accuracy.  We set the absolute threshold for each epoch as
$\phi\mathcal{S}$, where we fix $\phi=0.01$ and $\mathcal{S}$ as the total
number of packets in that epoch.  We keep the same memory usage of \sysname
and RHHH.  We focus on 1D HHH detection, and similar observations are made for
the 2D cases. 

Figure~\ref{fig:exp-6} shows the results.  The accuracy of RHHH drops in small
epoch lengths, due to its slow convergence.  For example, its precision is less
than 0.9 if the epoch length is less than 30\,s in the 1D-byte case;
in the 1D-bit case, both its precision and recall converge to around 0.8 after
300 seconds (conforming to the results in the original paper
\cite{BenBasat2017}). In contrast, the precision and recall of \sysname are
higher than 0.99 in all settings. 

\paragraph{(Experiment~7) \sysname in hardware.}
We evaluate \sysname for 1D-byte HHH detection in a Tofino switch
\cite{tofino}.  We configure the number of buckets from arrays $A_0$ to $A_4$
as 2048, 2048, 2048, 256, and 1, respectively.  In this case, both the
precision and recall of \sysname are above 0.9 for an epoch length of one
second in our traces.   

Table~\ref{tab:tofino} summarizes the resource usage of \sysname in the Tofino
switch, in terms of the number of physical stages used, SRAM consumption,
the number of stateful ALUs consumed, and the message size overhead across
stages in the packet header vector (PHV).  \sysname occupies all 12 physical
stages of the switch.  Nevertheless, its average resource consumption per
stage is small, and the remaining resources in each occupied stage can still be
made available for other applications. For example, \sysname consumes only
2.81\% of SRAM and 27.18\% of stateful ALUs of the switch. The total size of
messages across stages, including packet header information and metadata needed
by \sysname, is 132~bytes, among which only 48~bytes are due to the metadata
from \sysname.  

We also validate that \sysname's throughput now achieves
100\,Gb/s in our testbed (bounded by our packet generation rate), and it does
not have any packet resubmission or recirculation.  As \sysname incurs limited
switch resource overhead, we conjecture that its throughput in switch hardware
can be even higher in production deployment.

\begin{table}[t]
\centering
\renewcommand{\arraystretch}{1.1}
\small
\begin{tabular}[c]{|c|c|c|c|}
\hline
\textbf{No. stages} & \textbf{SRAM (KiB)} & \textbf{No. SALUs} & \textbf{PHV
size (byte)}\\
\hline
12 (100\%) & 432 (2.81\%) & 13 (27.08\%)  & 132 (17.18\%) \\
\hline
\end{tabular}
%\vspace{3pt}
\caption{(Experiment~7) Switch resource usage of \sysname; percentages are
fractions of total resources.}
\label{tab:tofino}
\vspace{-6pt}
\end{table}

%-------------------------------------------------------------------------
% Section VII: Related Work
%-------------------------------------------------------------------------
\section{Related Work}
\label{sec:related} 

\paragraph{Dynamic data structures.} To maintain memory efficiency in HHH
detection, prior studies propose dynamic data structures that insert or delete
keys of interest on-the-fly.  Trie-based HHH detection
\cite{Zhang2004,Truong2009} tracks keys in trie nodes and
dynamically spawns new child nodes if a trie node has a byte count above a
splitting threshold.  Cormode et al.  \cite{Cormode2008} propose 
{\em full ancestry} and {\em partial ancestry}, both of which build
on Lossy Counting \cite{Manku2002} with hierarchy awareness.  Both algorithms
maintain a lattice structure that dynamically adds or removes nodes.  In
contrast, \sysname uses static memory allocation and incurs no dynamic memory
management overhead.

\paragraph{Extensions of HH detection.} Several studies extend existing
heavy-hitter-based solutions for HHH detection.  Lin et al. \cite{Lin2007} adapt
Space Saving \cite{Metwally2005} to improve the accuracy of 1D HHH detection.
Mitzenmacher et al.  \cite{Mitzenmacher2012} further extend Space Saving with
better space efficiency.  Randomized HHH (RHHH) \cite{BenBasat2017} extends the
solution by Mitzenmacher et al.  \cite{Mitzenmacher2012} with randomization: it
maintains a Space Saving instance for each level of the hierarchy and randomly
updates only one of the instances for each packet.  While RHHH achieves high
update throughput, it has slow convergence.  In contrast, \sysname preserves
the invertibility and static memory allocation of MV-Sketch and adopts a
pipelined design to achieve both lightweight updates and fast convergence in
HHH detection.

\paragraph{TCAM-based solutions.} Some studies \cite{Jose2011, Moshref2013,
Popescu2017} leverage TCAM counters in hardware switches for 1D HHH detection,
by matching and counting packets in the data plane and adapting the monitoring
rules for different prefixes based on counter values.  They rely on a
centralized controller to decide the rules on which specific aggregation
levels are monitored.  In contrast, \sysname can work entirely in the data
plane for general aggregation levels. 

\paragraph{Others.}  Some HHH detection solutions specifically address the
practical requirements of network measurement.  AutoFocus \cite{Estan2003} is
an offline traffic analysis tool for identifying large traffic clusters.  Cho
\cite{Cho2017} proposes a recursive partitioning approach for tractable HHH
detection from an operational perspective. 

%--------------------------------------------------------------------------
% Section VIII: Conclusions
%--------------------------------------------------------------------------
\section{Conclusions}
\label{sec:conclusion} 

We revisit the HHH detection problem in network measurement.  We present
\sysname, a novel invertible sketch that supports both lightweight updates and
fast convergence in HHH detection and can be feasibly deployed in programmable
switches.  \sysname builds on the skewness property of IP traffic and the
pipelined executions of majority voting.  Theoretical analysis and prototype
evaluation in both software and hardware justify the design properties of
\sysname: high accuracy, high update throughput, fast convergence, and
resource efficiency in P4-based switch deployment. 

\bibliographystyle{abbrv}
\bibliography{paper}

\begin{thebibliography}{10}

\bibitem{murmurhash}
A.~Appleby.
\newblock \url{https://github.com/aappleby/smhasher}, 2016.

\bibitem{BenBasat2020}
R.~B. Basat, X.~Chen, G.~Einziger, and O.~Rottenstreich.
\newblock {Designing Heavy-Hitter Detection Algorithms for Programmable
  Switches}.
\newblock {\em IEEE/ACM Transactions on Networking}, 2020.

\bibitem{BenBasat2017}
R.~B. Basat, G.~Einziger, R.~Friedman, M.~C. Luizelli, and E.~Waisbard.
\newblock {Constant Time Updates in Hierarchical Heavy Hitters}.
\newblock In {\em Proc. of ACM SIGCOMM}, pages 127--140, 2017.

\bibitem{Basat2018}
R.~Ben-Basat, X.~Chen, G.~Einziger, and O.~Rottenstreich.
\newblock {Efficient Measurement on Programmable Switches Using Probabilistic
  Recirculation}.
\newblock In {\em Proc. of IEEE ICNP}, pages 313--323, 2018.

\bibitem{Bosshart2014}
P.~Bosshart, D.~Daly, G.~Gibb, M.~Izzard, N.~McKeown, J.~Rexford,
  C.~Schlesinger, D.~Talayco, A.~Vahdat, G.~Varghese, et~al.
\newblock {P4: Programming Protocol-Independent Packet Processors}.
\newblock {\em ACM SIGCOMM Computer Communication Review}, 44(3):87--95, 2014.

\bibitem{Bosshart2013}
P.~Bosshart, G.~Gibb, H.-S. Kim, G.~Varghese, N.~McKeown, M.~Izzard, F.~Mujica,
  and M.~Horowitz.
\newblock {Forwarding Metamorphosis: Fast Programmable Match-Action Processing
  in Hardware for SDN}.
\newblock In {\em Proc. of ACM SIGCOMM}, number~4, 2013.

\bibitem{Boyer1991}
R.~S. Boyer and J.~S. Moore.
\newblock {MJRTY -- A Fast Majority Vote Algorithm}.
\newblock In {\em Automated Reasoning}, pages 105--117. Springer, 1991.

\bibitem{caida}
CAIDA.
\newblock \url{http://www.caida.org/data/passive/trace_stats/}, 2022.

\bibitem{Cho2017}
K.~Cho.
\newblock {Recursive Lattice Search: Hierarchical Heavy Hitters Revisited}.
\newblock In {\em Proc. of ACM IMC}, pages 283--289, 2017.

\bibitem{wide}
K.~Cho, K.~Mitsuya, and A.~Kato.
\newblock {Traffic Data Repository at the WIDE Project}.
\newblock In {\em {USENIX 2000 FREENIX Track}}, June 2000.
\newblock \url{https://mawi.wide.ad.jp/mawi/}.

\bibitem{Cormode2008}
G.~Cormode, F.~Korn, S.~Muthukrishnan, and D.~Srivastava.
\newblock {Finding Hierarchical Heavy Hitters in Streaming Data}.
\newblock {\em ACM Trans. on Knowledge Discovery from Data}, 1(4):2, 2008.

\bibitem{Estan2003}
C.~Estan, S.~Savage, and G.~Varghese.
\newblock {Automatically Inferring Patterns of Resource Consumption in Network
  Traffic}.
\newblock In {\em Proc. of ACM SIGCOMM}, pages 137--148, 2003.

\bibitem{Estan2002}
C.~Estan and G.~Varghese.
\newblock {New Directions in Traffic Measurement and Accounting}.
\newblock In {\em Proc. of ACM SIGCOMM}, 2002.

\bibitem{Fang1999}
W.~Fang and L.~Peterson.
\newblock {Inter-AS Traffic Patterns and Their Implications}.
\newblock In {\em Proc. of IEEE GLOBECOM}, 1999.

\bibitem{Fayaz2015}
S.~K. Fayaz, Y.~Tobioka, V.~Sekar, and M.~Bailey.
\newblock {Bohatei: Flexible and Elastic DDoS Defense}.
\newblock In {\em Proc. of USENIX Security Symposium}, pages 817--832, 2015.

\bibitem{Gupta2018}
A.~Gupta, R.~Harrison, M.~Canini, N.~Feamster, J.~Rexford, and W.~Willinger.
\newblock {Sonata: Query-Driven Streaming Network Telemetry}.
\newblock In {\em Proc. of ACM SIGCOMM}, pages 357--371, 2018.

\bibitem{Jose2011}
L.~Jose, M.~Yu, and J.~Rexford.
\newblock {Online Measurement of Large Traffic Aggregates on Commodity
  Switches}.
\newblock In {\em Proc. of ACM Hot-ICE}, 2011.

\bibitem{Kallitsis2016}
M.~Kallitsis, S.~A. Stoev, S.~Bhattacharya, and G.~Michailidis.
\newblock {AMON: An Open Source Architecture for Online Monitoring, Statistical
  Analysis, and Forensics of Multi-Gigabit Streams}.
\newblock {\em IEEE Journal on Selected Areas in Communications},
  34(6):1834--1848, 2016.

\bibitem{Karp2003}
R.~M. Karp, S.~Shenker, and C.~H. Papadimitriou.
\newblock {A Simple Algorithm for Finding Frequent Elements in Streams and
  Bags}.
\newblock {\em ACM Trans. on Database Systems}, 28(1):51--55, 2003.

\bibitem{Kuvcera2020}
J.~Ku{\v{c}}era, D.~A. Popescu, H.~Wang, A.~Moore, J.~Ko{\v{r}}enek, and
  G.~Antichi.
\newblock {Enabling Event-Triggered Data Plane Monitoring}.
\newblock In {\em Proc. of ACM SOSR}, pages 14--26, 2020.

\bibitem{Lin2007}
Y.~Lin and H.~Liu.
\newblock {Separator: Sifting Hierarchical Heavy Hitters Accurately from Data
  Streams}.
\newblock {\em Proc, of ACM ADMA}, pages 170--182, 2007.

\bibitem{Manku2002}
G.~S. Manku and R.~Motwani.
\newblock {Approximate Frequency Counts over Data Streams}.
\newblock In {\em Proc. of ACM VLDB}, pages 346--357, 2002.

\bibitem{Metwally2005}
A.~Metwally, D.~Agrawal, and A.~E. Abbadi.
\newblock {Efficient Computation of Frequent and Top-K Elements in Data
  Streams}.
\newblock In {\em Proc. of ICDT}, pages 398--412, 2005.

\bibitem{Mitzenmacher2012}
M.~Mitzenmacher, T.~Steinke, and J.~Thaler.
\newblock {Hierarchical Heavy Hitters with the Space Saving Algorithm}.
\newblock In {\em Proc. of ACM ALENEX}, pages 160--174, 2012.

\bibitem{Moraney2020}
J.~Moraney and D.~Raz.
\newblock {On the Practical Detection of Hierarchical Heavy Hitters}.
\newblock In {\em Proc. of IFIP Networking}, pages 37--45, 2020.

\bibitem{Moshref2013}
M.~Moshref, M.~Yu, and R.~Govindan.
\newblock {Resource/Accuracy Tradeoffs in Software-Defined Measurement}.
\newblock In {\em Proc. of ACM HotSDN}, pages 73--78, 2013.

\bibitem{Narayana2017}
S.~Narayana, A.~Sivaraman, V.~Nathan, P.~Goyal, V.~Arun, M.~Alizadeh,
  V.~Jeyakumar, and C.~Kim.
\newblock {Language-Directed Hardware Design for Network Performance
  Monitoring}.
\newblock In {\em Proc. of ACM SIGCOMM}, pages 85--98, 2017.

\bibitem{p4}
{P4 Open Source Programming Language}.
\newblock \url{https://p4.org}, 2022.

\bibitem{Pfaff2015}
B.~Pfaff, J.~Pettit, T.~Koponen, E.~J. Jackson, A.~Zhou, J.~Rajahalme,
  J.~Gross, A.~Wang, J.~Stringer, P.~Shelar, et~al.
\newblock {The Design and Implementation of Open vSwitch}.
\newblock In {\em Proc. of USENIX NSDI}, 2015.

\bibitem{Popescu2017}
D.~A. Popescu, G.~Antichi, and A.~W. Moore.
\newblock {Enabling Fast Hierarchical Heavy Hitter Detection using Programmable
  Data Planes}.
\newblock In {\em Proc. of ACM SOSR}, pages 191--192, 2017.

\bibitem{Sarrar2012}
N.~Sarrar, S.~Uhlig, A.~Feldmann, R.~Sherwood, and X.~Huang.
\newblock {Leveraging Zipf's Law for Traffic Offloading}.
\newblock {\em ACM SIGCOMM Computer Communication Review}, 42(1):16--22, 2012.

\bibitem{Sekar2006}
V.~Sekar, N.~Duffield, O.~S. amd Kobus~{van der Merwe}, and H.~Zhang.
\newblock {LADS: Large-scale Automated DDoS Detection System}.
\newblock In {\em Proc. of USENIX ATC}, 2006.

\bibitem{Sivaraman2016}
A.~Sivaraman, A.~Cheung, M.~Budiu, C.~Kim, M.~Alizadeh, H.~Balakrishnan,
  G.~Varghese, N.~McKeown, and S.~Licking.
\newblock {Packet Transactions: High-Level Programming for Line-Rate Switches}.
\newblock In {\em Proc. of ACM SIGCOMM}, pages 15--28, 2016.

\bibitem{Sivaraman2017}
V.~Sivaraman, S.~Narayana, O.~Rottenstreich, S.~Muthukrishnan, and J.~Rexford.
\newblock {Heavy-Hitter Detection Entirely in the Data Plane}.
\newblock In {\em Proc. of Proc. of SOSR}, pages 164--176, 2017.

\bibitem{Tang2019}
L.~Tang, Q.~Huang, and P.~P. Lee.
\newblock {A Fast and Compact Invertible Sketch for Network-Wide Heavy Flow
  Detection}.
\newblock {\em IEEE/ACM Trans. on Networking}, 28(5):2350--2363, Oct 2020.

\bibitem{tofino}
Tofino.
\newblock
  \url{https://www.intel.com/content/www/us/en/products/network-io/programmable-ethernet-switch/tofino-series/tofino.html},
  2022.

\bibitem{Truong2009}
P.~Truong and F.~Guillemin.
\newblock {Identification of Heavyweight Address Prefix Pairs in IP Traffic}.
\newblock In {\em Proc. of IEEE International Teletraffic Congress}, pages
  1--8, 2009.

\bibitem{Zhang2002}
Y.~Zhang, L.~Breslau, V.~Paxson, and S.~Shenker.
\newblock {On the Characteristics and Origins of Internet Flow Rates}.
\newblock In {\em Proc. of ACM SIGCOMM}, 2002.

\bibitem{Zhang2004}
Y.~Zhang, S.~Singh, S.~Sen, N.~Duffield, and C.~Lund.
\newblock {Online Identification of Hierarchical Heavy Hitters: Algorithms,
  Evaluation, and Applications}.
\newblock In {\em Proc. of ACM IMC}, pages 101--114, 2004.

\end{thebibliography}

\end{document}